\pdfoutput=1
\documentclass[a4paper, 11pt]{article}
\usepackage{amsmath,a4wide}
\usepackage{amssymb}
\usepackage{amscd}
\usepackage{epsfig}
\usepackage{graphics}
\usepackage{multirow}
\usepackage{graphicx}

\usepackage{subfigure}
\usepackage{epstopdf}
\usepackage{booktabs}
\usepackage{pgfplots}
\usepackage[T1]{fontenc}
\usepackage{mathpazo}
\usepackage{microtype}

\usepackage{ifpdf}
\usepackage[a4paper,top=1in, bottom=1in, left=0.75in, right=0.75in]{geometry}
\usepackage{changes}
\usepackage[makeroom]{cancel}
\usepackage{bbm}

\usepackage{amsthm}
\usepackage{color}
\definecolor{amaranth}{rgb}{0.9, 0.17, 0.31}
\definecolor{myblue}{rgb}{0.2,0,0.9}
\definecolor{auburn}{rgb}{0.43, 0.21, 0.1}
\definecolor{bittersweet}{rgb}{1.0, 0.44, 0.37}
\definecolor{blue-violet}{rgb}{0.54, 0.17, 0.89}
\definecolor{takadd}{rgb}{0.0, 0.25, 0.75}
\definecolor{takdelete}{rgb}{0.8, 0.0, 0.0}
\definecolor{takcomment}{rgb}{0.0, 0.45, 0.2}
\usepackage[numbers,sort&compress]{natbib}

\usepackage{hyperref}
\usepackage[nameinlink]{cleveref}

\hypersetup{
	colorlinks,linkcolor=myblue,citecolor=blue-violet,filecolor=magenta,urlcolor=myblue}

\newtheorem{pro}{Proposition}[section]
\newtheorem{pr}{Problem}[section]

\newtheorem{thm}{Theorem}[section]

\newtheorem{lem}{Lemma}[section]

\newtheorem{as}{Assumption}[section]
\newtheorem{rem}{Remark}[section]

\def\r{\rho}

\newcommand{\bdryinv}{\mathcal{Z}_{\mathcal{I}}}
\newcommand{\bdrydisinv}{\mathcal{Z}_{\mathcal{D}}}

\numberwithin{equation}{section}
\begin{document}

\title{\LARGE \textbf{Finite-Horizon Reversible Investment under Multi-Factor Dynamics}\thanks{Junkee Jeon was supported  by the National Research Foundation of Korea (NRF) grant funded by the Korea government (MSIT) (Grant No. RS-2026-25591924). Takwon Kim was supported by the National Research Foundation of Korea (NRF) grant funded by the Korea government (MSIT) (Grant No. RS-2024-00351151). Jinwan Park was supported by Basic Science Research Program through the National Research Foundation of Korea (NRF) funded by the Ministry of Education (Grant No. RS-2025-25415411). A.\ Max Reppen acknowledges financial support from Boston University Hariri Institute for Computing and Computational Science. \\
}}

\author{
Junkee Jeon\footnote{E-mail: \href{mailto:junkeejeon@khu.ac.kr}{\tt junkeejeon@khu.ac.kr};
Department of Applied Mathematics, Kyung Hee University, Yongin, Korea.}
\and
Takwon Kim\footnote{E-mail: \href{mailto:takwon@sungshin.ac.kr}{\tt takwon@sungshin.ac.kr};
School of Mathematics, Statistics and Data Science, Sungshin Women's University, Seoul, Korea.}
\and
Jinwan Park\footnote{E-mail: \href{mailto:jinwan@kongju.ac.kr}{\tt jinwan@kongju.ac.kr};
Department of Applied Mathematics, Kongju National University, Gongju, Korea.}
\and
A. Max Reppen\footnote{Corresponding author. E-mail: \href{mailto:amreppen@bu.edu}{\tt amreppen@bu.edu};
Questrom School of Business, Boston University, 595 Commonwealth Avenue, Boston,
MA 02215, USA.}
}

	\date{\today}

	\maketitle \pagestyle{plain} \pagenumbering{arabic}

\abstract{
We study a finite-horizon reversible investment problem in which a risk-neutral firm adjusts capacity at a proportional purchase cost and a lower salvage value under multi-factor geometric Brownian motion. Via the singular control--optimal switching correspondence, the marginal value of capacity solves a family of parabolic double-obstacle problems. We prove existence, uniqueness and local Sobolev regularity of the strong solution, characterize investment, waiting and disinvestment regions by continuous, strictly separated free boundaries, and verify optimality of the reflected capacity process. Numerically, joint demand improvements shift both boundaries super-additively, 1.5--2.7 times as strongly at the disinvestment boundary, depending on factor correlation.
}

\vspace{0.2cm}

\vspace{0.2cm}

\vspace{0.2cm}

\vspace{0.2cm}

\noindent{\em Running title}: Multi-factor reversible investment.
\newpage

\section{Introduction}

Investment under uncertainty is a classical setting in which stochastic
control, optimal stopping, and free-boundary methods interact naturally. A firm
that expands productive capacity under uncertainty must balance the value of
waiting against the cost of delayed investment: investing too early creates the
risk of excess capacity, whereas waiting too long may forgo operating profits.
The trade-off becomes two-sided when investment is reversible but costly.
Capacity may then be increased or decreased, but the purchase cost of capital
exceeds its salvage value, creating an endogenous inaction region.

This paper studies a finite-horizon reversible investment problem in which a
risk-neutral firm adjusts productive capacity through two finite-variation
controls representing investment and disinvestment. Upward adjustment incurs a
constant proportional cost, while downward adjustment produces a smaller
proportional salvage value. The economic environment is driven by an
\(n\)-dimensional geometric Brownian motion, and the operating revenue may
depend jointly on all of its components.  Because the free-boundary
analysis assumes that marginal revenue is nondecreasing in every state
coordinate, each coordinate is interpreted as a favourable economic index.
Examples include output demand or prices, productivity, inverse input-price
indices, energy availability, favourable financial conditions, and sectoral
or regional demand shocks.  A scalar state variable
adequately represents such exposures only when they move together or enter the
revenue function through a one-dimensional sufficient statistic. Our
formulation addresses settings in which no such reduction is available.

The multi-factor specification permits economically relevant substitution,
complementarity, and production-switching effects.  Smooth CES and
Cobb--Douglas specifications are leading examples.  Exact Leontief and
production-switching specifications are nonsmooth benchmarks, and are covered
as well: the free-boundary assumption is stated in terms of a comparable
smooth majorant of marginal revenue, precisely so that marginal revenue itself
need only be locally Lipschitz.

The mathematical difficulty is genuinely multi-dimensional. Many reversible
investment models rely on a one-dimensional ordering, homogeneity, or an
explicit dimension reduction. These tools are unavailable when the economic
state is a multi-dimensional diffusion and marginal revenue depends on the full
state vector. The intervention boundaries are then hypersurfaces rather than
scalar trigger curves. Moreover, the marginal value of capacity is constrained
by two obstacles: the investment cost from above and the disinvestment value
from below. The resulting equation is therefore a parabolic double obstacle
problem, rather than the one-sided obstacle problem associated with a standard
American option.

Our analysis has three parts.  The marginal capacity problem is reduced to a
family of two-regime switching problems indexed by capacity level, whose value
difference solves a single \(n\)-dimensional parabolic double obstacle problem;
we establish existence, uniqueness, boundedness and local Sobolev regularity of
its strong solution.  Under monotonicity and growth conditions on marginal
revenue we then identify the investment, waiting and disinvestment regions and
obtain continuous capacity-level free boundaries together with their terminal
behaviour.  Finally, these boundaries define a time-dependent inaction interval
for capacity, and the Skorokhod problem on that interval produces a reflected
capacity process whose induced level-crossing controls we verify to be optimal.
No integrability restriction is placed on the control class: a growth condition
imposed directly on operating revenue makes the objective well defined for
every finite-variation control, and gives every control of finite value the
integrability that the decomposition requires.

The main results, and the assumptions on which each
rests, are summarized below:\\
\begin{tabular}{@{}p{0.12\textwidth}p{0.55\textwidth}p{0.16\textwidth}@{}}
\toprule
Result & Content & Assumptions \\
\midrule
Theorem~\ref{thm:Q}
  & Marginal value of capacity ${\cal Q}$: existence, uniqueness,
    boundedness and $W^{1,2}_{p,\mathrm{loc}}$ regularity;
    $\partial_t{\cal Q}\le0$
  & \ref{as:global}(i) \\
\midrule
Theorem~\ref{thm:free-boundary-continuity}
  & Continuity of free boundaries: ${\cal Z}_{\cal I}$ and ${\cal Z}_{\cal D}$
  & \ref{as:global}, \ref{as:free_boundary} \\
Theorem~\ref{thm:capacity-boundary-regularity}
  & Free boundary smoothness: ${\cal Z}_{\cal I},{\cal Z}_{\cal D}\in C^1$ \newline
    \phantom{Free boundary smoothness:} ${\cal Z}_{\cal I},{\cal Z}_{\cal D}\in C^\infty$
  & \ref{as:global}, \ref{as:free_boundary}, \ref{as:capacity-first-order-data} \newline
    \ref{as:global}, \ref{as:free_boundary},
    \ref{as:capacity-first-order-data}--\ref{as:capacity-exact-compensation} \\
\midrule
Theorem~\ref{thm:verification-singular-control}
  & Optimality of the reflected control; representation of the value
    function
  & \ref{as:global}, \ref{as:free_boundary} \\
\bottomrule
\end{tabular}\\
Assumption~\ref{as:global} restricts operating revenue and
Assumption~\ref{as:free_boundary} restricts marginal revenue.
Assumptions~\ref{as:capacity-first-order-data}--\ref{as:capacity-exact-compensation}
are needed only for the higher regularity of the free boundaries, and the
verification uses only their continuity and strict separation.

The probabilistic construction uses two principal tools. The correspondence
between singular control and consistent families of optimal switching problems
developed by \citet{Guo2008} provides the level-wise decomposition. The theory
of the Skorokhod problem on a time-dependent interval developed by
\citet{BurdzyKangRamanan2009} converts the free boundaries into a reflected
capacity process. Together with the PDE analysis and the coercive no-loss
argument, these tools identify and verify an optimal finite-variation
investment--disinvestment policy.

The paper is organized as follows. Section~\ref{sec:Model} formulates the
finite-horizon reversible investment problem and states the assumptions on the
operating revenue. The subsequent section derives the family of optimal
switching problems through the layer-cake decomposition and the
singular-control--switching correspondence, and shows that the associated value
difference satisfies the parabolic double obstacle problem.
Section~\ref{sec:pdo} establishes existence, uniqueness, and local Sobolev
regularity, first on bounded domains and then on the unbounded state space.
Section~\ref{sec:fb} characterizes the intervention regions and their
capacity-level free boundaries. Section~\ref{sec:verification} recovers the
switching value pair, constructs the reflected capacity process, proves the
required a-posteriori admissibility, and verifies optimality for the switching
and singular control problems. Section~\ref{sec:Numeric} illustrates the
resulting multi-dimensional intervention geometry. Section~\ref{sec:implications}
quantifies the interaction between the demand factors at the free boundaries and its dependence on the factor correlation. These computations illustrate
the geometry and economic mechanisms generated by the model rather than supply
additional evidence for the analytical results. Section~\ref{sec:conclusion}
concludes.

\medskip
\noindent\textbf{Related literature.}
The economic foundation of irreversible investment and real options goes back
to \citet{Arrow1968}, \citet{Pindyck1988,Pindyck1991}, and
\citet{DixitPindyck1994}. Costly reversibility and the resulting inaction-band
policy were developed by \citet{AbelEberly1996}. In mathematical finance and
stochastic control, related irreversible and reversible investment problems
have been studied through singular control, optimal stopping, and
monotone-follower methods; see, among others,
\citet{BaldurssonKaratzas1997}, \citet{Oksendal2000},
\citet{RiedelSu2011}, and \citet{Steg2012}. These papers provide the economic
and stochastic-control foundations for investment under irreversibility or
partial reversibility, whereas our focus is the finite-horizon, two-sided
problem under genuinely multi-factor dynamics.

Our work is also connected to the literature on the reduction of singular
control to optimal stopping or optimal switching. Classical connections
between singular control and optimal stopping were developed by
\citet{KaratzasShreve1984,KaratzasShreve1985}, and
\citet{ElKarouiKaratzas1991}; see also \citet{BankElKaroui2004} for a
stochastic representation approach. For reversible investment,
\citet{GuoPham2005} and \citet{MerhiZervos2007} study partially reversible
capacity expansion models, while \citet{Guo2008} establishes the
correspondence between singular control and consistent families of optimal
switching problems used here. We combine this correspondence with a
multi-dimensional double-obstacle analysis and a time-dependent Skorokhod
reflection to construct and verify the optimal finite-variation policy. The
one-dimensional setting of \citet{MerhiZervos2007} is closely related to a
reduced case of our framework. Our construction does not posit intervention
thresholds in advance: the double obstacle problem identifies the boundaries,
and the switching and Skorokhod representations turn them into the candidate
control.

Our problem is structurally related to the transaction-cost literature, where
proportional costs turn the frictionless Merton problem
\citep{Merton1969,Merton1971} into a two-sided singular control problem whose
optimal policy keeps the portfolio inside a no-trading region by minimal
reflection at its boundaries; this structure was introduced by
\citet{MagillConstantinides1976}, characterized by \citet{DavisNorman1990},
and developed through viscosity methods by \citet{ShreveSoner1994}, with dual
and shadow-price approaches in \citet{CvitanicKaratzas1996},
\citet{KallsenMuhleKarbe2010}, \citet{ChoiSirbuZitkovic2013}, and
\citet{GuasoniMuhleKarbe2013}.  Closest to our methodology are the
finite-horizon and multi-asset versions.  \citet{DaiYi2009} reduce a
finite-horizon transaction-cost problem to a parabolic double obstacle problem
for the derivative of the value function, \citet{DaiJiangLiYi2009} extend this
to investment and consumption, and \citet{BelakSass2019} construct and verify
optimal finite-horizon strategies, which is conceptually close to our pathwise
construction of the reflected capacity process; multi-asset models are studied
by \citet{AkianMenaldiSulem1996}, \citet{KabanovKluppelberg2004},
\citet{ChenDai2013}, and \citet{HobsonTseZhu2019}, and small-cost asymptotics
and mixed fixed-proportional costs by \citet{JanecekShreve2004},
\citet{SonerTouzi2013}, \citet{PossamaiSonerTouzi2015}, and
\citet{AltaroviciReppenSoner2017}.  Unlike these portfolio models, our
controlled state is productive capacity under a general multi-factor revenue
function, and we do not rely on homotheticity, portfolio-ratio reductions, or a
special utility specification.

Finite-horizon reversible investment problems have been analyzed by
\citet{DeAngelisFerrari2014}, \citet{FedericoPham2014},
\citet{HanYiZhang2024}, \citet{HanYi2025}, and \citet{JeonKim2026}.
These works obtain detailed structural results in one-dimensional,
two-dimensional, or degenerate settings. Here the operating revenue may depend
on the full state vector, and the marginal problem is a genuinely
\(n\)-dimensional parabolic double obstacle problem. Accordingly, the
investment and disinvestment boundaries are treated as hypersurfaces rather
than scalar or reduced-dimensional trigger curves.

Finally, our work is related to free-boundary problems arising from American
options on several assets. American options lead to one-sided obstacle
problems; see \citet{McKean1965}, \citet{Bensoussan1984},
\citet{Karatzas1988}, \citet{BensoussanLions1982},
\citet{JailletLambertonLapeyre1990}, \citet{Kim1990},
\citet{Jacka1991}, and \citet{CarrJarrowMyneni1992}. For multi-asset
American options, see \citet{BroadieDetemple1997} and
\citet{Villeneuve1999}. \citet{LaurenceSalsa2009} prove first-order
regularity of the free boundary for American options on several assets;
further developments include \citet{KlimsiakRozkosz2016} and
\citet{ReppenSonerTissotDaguette2025}. In contrast, our marginal problem has
two interacting intervention regions. Its central analytical feature is
therefore the simultaneous treatment of the investment and disinvestment
boundaries generated by a two-sided singular control problem.

\section{Model}\label{sec:Model}

Fix a finite horizon \(T>0\). We work on a probability space \((\Omega, {\cal F}, \mathbb{F}, \mathbb{P})\)
with a filtration \(\mathbb{F}=({\cal F}_t)_{t\ge0}\) satisfying the usual
conditions of right-continuity and completeness. The space supports a standard
\(n\)-dimensional \(\mathbb F\)-Brownian motion \(B=(B_t)_{t\ge0}\).

Throughout, \(\rho>0\) denotes a constant discount rate, and
\(K^+>K^->0\) denote, respectively, the per-unit cost of upward capacity
adjustment and the per-unit salvage value from downward adjustment.

For a given initial capacity level \(y>0\), let \(\Xi(y)\) denote the set of all
pairs \((\xi^+,\xi^-)\) of \(\mathbb{F}\)-adapted, non-decreasing, left
continuous with right limits (LCRL) processes on \([0,T]\) with
\(\xi^+_T<\infty\) and \(\xi^-_T<\infty\) almost surely, such that
\begin{itemize}
    \item[(i)] \(\xi^+_0=\xi^-_0=0\) and, writing
    \(Y_t=y+\xi^+_t-\xi^-_t\), one has
    \(Y_t>0\) for every \(t\in[0,T]\) and
    \(Y_{t+}>0\) for every \(t\in[0,T)\), almost surely;
    \item[(ii)] the associated measures \(d\xi^+\) and \(d\xi^-\) are mutually
    singular almost surely, so that upward and downward adjustments never occur
    simultaneously.
\end{itemize}
The process \(Y=(Y_t)_{t\in[0,T]}\), with \(Y_0=y\), is the firm's capacity,
and \(\xi^+\) and \(\xi^-\) represent, respectively, cumulative investment and
disinvestment.

Controls are LCRL rather than right continuous because the correspondence
between singular control and optimal switching of \citet{Guo2008}, on which
Section~\ref{sec:switching} rests, is stated for that class.  Accordingly, for
a non-decreasing LCRL process \(A\) we write \(dA\) for the measure with
\(dA([s,t))=A_t-A_s\), so that \(dA(\{t\})=\Delta^+A_t:=A_{t+}-A_t\)
for \(t<T\); in
particular, integration over \([0,T)\) charges a right-jump at time zero and
excludes an intervention at \(T\).

The economic environment is described by an \(n\)-dimensional geometric
Brownian motion \(X=(X_t)_{t\in[0,T]}\). Starting from \(X_0=x\in(0,\infty)^n\),
its dynamics are given by
\begin{equation}\label{eq:GBM}
\begin{aligned}
dX_t
&=
\operatorname{diag}(X_t)\bigl(\mu\,dt+\Sigma\,dB_t\bigr),
\qquad 0\le t\le T,\qquad
X_0=x,
\end{aligned}
\end{equation}
where \(\mu\in\mathbb R^n\) is the drift vector and
\(\Sigma\in\mathbb R^{n\times n}\) is the volatility matrix.
We assume that \(a:=\Sigma\Sigma^\top\) is positive definite and denote
its entries by \(a_{ij}\).
We write \(\mathbb E_{t,x}\) for expectation under the law
of \(X\) started from \(x\) at time \(t\), and abbreviate
\(\mathbb E_x:=\mathbb E_{0,x}\).

The operating revenue function \(\Pi(x,y)\) depends on the economic state
\(x=(x_1,\dots,x_n)\) and the firm's capacity \(y\).
The firm is risk-neutral and maximizes expected discounted profit over the
finite horizon \([0,T]\):

\begin{pr}\label{pr:main}
For given initial conditions \(x\in(0,\infty)^n\) and \(y>0\), the
optimization problem is
\begin{equation}\label{eq:value-function}
\mathfrak{U}(x,y)
:=
\sup_{(\xi^+,\xi^-)\in\Xi(y)}
{\cal J}(x,y;\xi^+,\xi^-).
\end{equation}
Here, the objective functional associated with
\((\xi^+,\xi^-)\in\Xi(y)\) is given by
\begin{equation}\label{eq:main-problem}
\begin{aligned}
{\cal J}(x,y;\xi^+,\xi^-)
&{:=}
\mathbb E_x\bigg[
\int_0^T e^{-\rho t}\Pi(X_t,Y_t)\,dt
-K^+\int_{[0,T)} e^{-\rho t}\,d\xi_t^+
+K^-\int_{[0,T)} e^{-\rho t}\,d\xi_t^-
+e^{-\rho T}K^-Y_T
\bigg].
\end{aligned}
\end{equation}
Under Assumption~\ref{as:global} below, the objective is well defined and
\(\mathfrak{U}\) is finite, as shown in
Section~\ref{sec:integrability}.
\end{pr}

Assumption~\ref{as:global} below imposes two types of requirements on
\(\Pi\): regularity in the capacity variable, because the analysis runs throughout
on the marginal revenue rather than on \(\Pi\) itself, and integrability conditions,
without which operating revenue can outrun the cost of the capacity that produces
it and the firm's value may be infinite. The latter are divided into a revenue-envelope
condition and a level-wise integrability condition. What we require is that revenue, net of a carrying
cost \(\rho K^-\) per unit of capacity, stay integrable.  This suffices for the
reduction of Section~\ref{sec:switching} and the solution theory of
Section~\ref{sec:pdo}; locating the investment and disinvestment boundaries asks
more of marginal revenue, and those conditions are collected in
Assumption~\ref{as:free_boundary}, where they are first used.
Remark~\ref{rem:Pi} describes several economically important specifications
satisfying both.

\begin{as}\label{as:global}
    We assume that the operating revenue function \(\Pi\) satisfies the following conditions.
\begin{enumerate}
    \item[\textnormal{(i)}] \emph{Regularity.} The function \(\Pi\) is jointly Borel measurable. For every
    \(x\in(0,\infty)^n\), the function \(y\mapsto\Pi(x,y)\) is non-negative,
    non-decreasing and differentiable, and its derivative
    \[\pi(x,y):=\frac{\partial\Pi}{\partial y}(x,y)\] is continuous on \((0,\infty)^n\times(0,\infty)\).

    \item[\textnormal{(ii)}] \emph{Revenue envelope.}  The function
    \begin{equation}\label{eq:revenue-envelope-Psi}
        \Psi(x):=\sup_{u>0}\bigl(\Pi(x,u)-\rho K^-u\bigr)
    \end{equation}
    satisfies
    \[
        \mathbb E_x\left[\int_0^T e^{-\rho t}\Psi(X_t)\,dt\right]<\infty,
        \qquad x\in(0,\infty)^n.
    \]

    \item[\textnormal{(iii)}] \emph{Level-wise integrability.}  For every
    \(z>0\),
    \[
        \mathbb E_x\left[\int_0^T e^{-\rho t}
            \pi(X_t,z)\,dt\right]<\infty,
        \qquad x\in(0,\infty)^n.
    \]
\end{enumerate}
\end{as}

\begin{rem}\label{rem:assumptions}
The supremum in \eqref{eq:revenue-envelope-Psi} is non-negative, because
\(\Pi\ge0\) and \(u\downarrow0\) is admissible.  By construction,
\begin{equation}\label{eq:revenue-envelope-bound}
    0\le \Pi(x,u)\le \rho K^-u+\Psi(x),
    \qquad (x,u)\in(0,\infty)^n\times(0,\infty),
\end{equation}
so every fixed capacity level has integrable discounted revenue,
\begin{equation}\label{eq:baseline-revenue-integrability}
    \mathbb E_x\left[\int_0^T e^{-\rho t}\Pi(X_t,u)\,dt\right]<\infty,
    \qquad u>0 .
\end{equation}
Clause~\textnormal{(iii)} is implied by \textnormal{(i)} and
\textnormal{(ii)} when \(\Pi(x,\cdot)\) is concave, since then
\(\pi(x,z)\le z^{-1}\int_0^z\pi(x,s)\,ds\le\rho K^-+z^{-1}\Psi(x)\); it is
stated separately because concavity is imposed only from
Section~\ref{sec:fb} on.  Clause~\textnormal{(i)} is what
Section~\ref{sec:pdo} uses, clauses~\textnormal{(i)}--\textnormal{(iii)} what
Section~\ref{sec:switching} uses, and Sections~\ref{sec:fb}
and~\ref{sec:verification} require in addition
Assumption~\ref{as:free_boundary}.
\end{rem}

\begin{rem} \label{rem:Pi}
    The following specifications illustrate economically relevant
    functional forms.  Remark~\ref{rem:examples-global} verifies
    Assumption~\ref{as:global} for them, and
    Remark~\ref{rem:free-boundary-examples} verifies
    Assumption~\ref{as:free_boundary}.
    We give three examples, with $n=2$ for simplicity:
    \begin{enumerate}
        \item Constant elasticity of substitution (CES):
            \[
                \Pi(x_1, x_2, y)
                =
                P \left(
                    \beta_1 x_1^\theta
                    + \beta_2 x_2^\theta
                    + \beta_y y^\theta
                \right)^{\gamma/\theta},
                \qquad
                \gamma > \theta \neq 0,
                \quad
                \gamma \in (0,1),
            \]
            where $P$ is the price, $x_i$ are state variables that represent market conditions, and $\frac{1}{1-\theta}$ is the elasticity of substitution.
            Here $P,\beta_1,\beta_2,$ and $\beta_y$ are positive.

            The broader CES specification has the following familiar formal limits:
            \begin{enumerate}
                \item $\theta \to 0$ yields the Cobb--Douglas production function, if \(\beta_1+\beta_2+\beta_y=1\).
                \item $\theta \to -\infty$ yields the Leontief production function.
                \item $\theta \to 1$ yields perfect substitutes---a formal limit
                outside the maintained restriction \(\theta<\gamma<1\).
            \end{enumerate}
            Finite $\theta$ thus interpolates between these familiar
            technologies.
        \item CES structure in non-capital factors:
            \[
                \Pi(x_1, x_2, y)
                =
                P \left(
                    \beta_1 x_1^\theta
                    + \beta_2 x_2^\theta
                \right)^{\gamma/\theta} y^\delta,
                \qquad
                \gamma > 0,
                \quad
                \theta \neq 0,
                \quad
                0 < \delta < 1.
            \]
            Here $P,\beta_1,$ and $\beta_2$ are positive.
            The $x$-factors are CES aggregated and $y$ has its own elasticity.
            Thus the substitutability of the $x$-factors is maintained, allowing for flexible modeling, while the $y$-factor enters complementarily.
            As above, this includes limits such as perfect complements:
            $\min\{x_1, x_2\}^\gamma y^\delta$.
            Here, the returns to scale are given by $\gamma + \delta$.

        \item Production switching:
            \[
                \Pi(x_1, x_2, y)
                =
                \max\{x_1, x_2\} \cdot y^\delta,
                \qquad
                0 < \delta < 1,
            \]
            which represents production capital $y$ capable of producing two goods with prices $x_1$ and $x_2$, with frictionless switching between them.
    \end{enumerate}
\end{rem}

\begin{rem}\label{rem:examples-global}
Each specification of Remark~\ref{rem:Pi} satisfies
Assumption~\ref{as:global}. Clause~\textnormal{(i)} is immediate, and,
by Remark~\ref{rem:assumptions},
clause~\textnormal{(iii)} follows from clause~\textnormal{(ii)} and concavity
in \(y\). Thus it remains to verify the envelope condition~\textnormal{(ii)}.
Writing \(c=\rho K^-\), maximizing \(A(x)u^\delta-cu\) over \(u>0\) gives,
for \(0<\delta<1\),
\begin{equation}\label{eq:power-revenue-envelope}
    \sup_{u>0}\bigl(A(x)u^\delta-cu\bigr)
    =(1-\delta)\,\delta^{\delta/(1-\delta)}
      c^{-\delta/(1-\delta)}A(x)^{1/(1-\delta)}.
\end{equation}
For the second and third specifications, the corresponding
functions \(A\) satisfy, respectively,
\[
A(x)\le C(x_1^\gamma+x_2^\gamma)
\qquad\text{and}\qquad
A(x)\le x_1+x_2.
\]
Hence \eqref{eq:power-revenue-envelope} shows that \(\Psi\) has at most
polynomial growth.
Since a geometric Brownian motion has moments of every positive order on
\([0,T]\), this verifies condition~\textnormal{(ii)} for these two
specifications.

For the first specification, when \(\theta>0\), a power inequality gives
\[
\Pi\le C_0(x_1^\gamma+x_2^\gamma)+C_1u^\gamma,
\]
to whose second term \eqref{eq:power-revenue-envelope}, with
\(\delta=\gamma\), applies. When \(\theta<0\),
\[
\Pi\le
P(\beta_1x_1^\theta+\beta_2x_2^\theta)^{\gamma/\theta}
\le C(x_1^\gamma+x_2^\gamma)
\]
uniformly in \(u\).
Thus \(\Psi(x)\le
C(1+x_1^\gamma+x_2^\gamma)\) in either case, completing the verification.
\end{rem}

\section{Reduction to a family of optimal switching problems}
\label{sec:switching}

This section decomposes the payoff of every finite-valued control in
Problem~\ref{pr:main} into a family of two-regime
optimal switching problems, one for each capacity level: the level is either
active or inactive, and raising or lowering capacity past it is a switch.  We
first record the integrability conditions required by the decomposition, and show that every
control of finite payoff automatically satisfies them.  We then set up the switching
problems in a convention matching the LCRL controls of Problem~\ref{pr:main},
carry out the decomposition, and reduce the resulting HJB system to a scalar
parabolic double-obstacle problem.
Section~\ref{sec:verification} later constructs a consistent
family that attains the resulting switching-value upper bound.

\subsection{Integrability of the singular payoff}
\label{sec:integrability}

The decomposition of Section~\ref{sec:switching} splits each term of
\eqref{eq:main-problem} into an integral over capacity levels, and so requires
the four terms to be separately integrable. The present
subsection identifies precisely which controls satisfy this requirement.
The tool is a pathwise rearrangement of the payoff, obtained by integrating
the adjustment costs by parts, in which the two costs collapse into a single
term of known sign.

For \(\xi=(\xi^+,\xi^-)\in\Xi(y)\), set
\[
    A_\rho^\pm(\xi)
    :=
    \int_{[0,T)}e^{-\rho t}\,d\xi_t^\pm,
    \qquad
    Y_t^\xi:=y+\xi_t^+-\xi_t^-.
\]
Discounted integration by parts gives, pathwise,
\begin{equation}\label{eq:discounted-capacity-ibp}
\begin{aligned}
e^{-\rho T}Y_T^\xi-y
&=
A_\rho^+(\xi)-A_\rho^-(\xi)
-\rho\int_0^T e^{-\rho t}Y_t^\xi\,dt,\\
-K^+A_\rho^+(\xi)+K^-A_\rho^-(\xi)
+e^{-\rho T}K^-Y_T^\xi
&=
K^-y-(K^+-K^-)A_\rho^+(\xi)
-\rho K^-\int_0^T e^{-\rho t}Y_t^\xi\,dt.
\end{aligned}
\end{equation}

By \eqref{eq:discounted-capacity-ibp}, the random variable inside the
expectation in \eqref{eq:main-problem} equals, pathwise,
\begin{equation}\label{eq:J-singular-extended}
    Z(\xi)
    :=
    K^-y+\int_0^T e^{-\rho t}
    \bigl(\Pi(X_t,Y_t^\xi)-\rho K^-Y_t^\xi\bigr)\,dt
    -(K^+-K^-)A_\rho^+(\xi),
\end{equation}
so that \({\cal J}(x,y;\xi^+,\xi^-)=\mathbb E_x[Z(\xi)]\).  On the right-hand
side the potentially positive and negative contributions are both controlled: the adjustment costs have collapsed into
\(-(K^+-K^-)A_\rho^+(\xi)\le0\), so that only upward adjustment is charged and
only the bid--ask spread is charged for it, while the running integrand is
operating revenue net of the carrying cost \(\rho K^-\) of a unit of capacity,
whose positive part is at most \(\Psi(X_t)\) by \eqref{eq:revenue-envelope-Psi}
and \(Y^\xi>0\).  Hence
\begin{equation}\label{eq:Z-upper-bound}
    Z(\xi)\le K^-y+\int_0^T e^{-\rho t}\Psi(X_t)\,dt ,
\end{equation}
whose right-hand side is integrable by
Assumption~\ref{as:global}\textnormal{(ii)} and therefore almost surely finite.
Moreover, since \(\Pi\ge0\), \(A_\rho^+(\xi)<\infty\)
and \(Y_t^\xi\le y+\xi_T^+<\infty\) pathwise,
\[
Z(\xi)\ge K^-y-\rho K^-\int_0^T e^{-\rho t}Y_t^\xi\,dt
-(K^+-K^-)A_\rho^+(\xi)>-\infty
\qquad\text{a.s.}
\]
So \(Z(\xi)\) is almost surely finite with
\(\mathbb E_x[Z(\xi)^+]<\infty\): for every \(\xi\in\Xi(y)\) the objective is
well defined in \([-\infty,\infty)\), no admissibility restriction on
\(\Xi(y)\) is needed, and taking suprema in \eqref{eq:Z-upper-bound} bounds
\(\mathfrak U(x,y)\) from above.
Finiteness of \(\mathbb E_x[A_\rho^+(\xi)]\) is a separate matter, and is not
automatic: \(A_\rho^+(\xi)\le\xi_T^+<\infty\) pathwise, but the expectation
need not be finite.  The following lemma identifies that condition as the exact
dividing line between controls of finite and of infinite payoff, and shows that
on the finite side it supplies all the integrability used by the level-crossing
decomposition below.

\begin{lem}\label{lem:finite-payoff-integrability}
Let \(\xi\in\Xi(y)\).  Then the following hold.
\begin{enumerate}
    \item[\textnormal{(i)}] \({\cal J}(x,y;\xi^+,\xi^-)>-\infty\) if and only
    if \(\mathbb E_x[A_\rho^+(\xi)]<\infty\).

    \item[\textnormal{(ii)}] In that case,
    \begin{equation}\label{eq:finite-payoff-control-integrability}
    \begin{aligned}
    &\mathbb E_x[A_\rho^+(\xi)]
    +\mathbb E_x[A_\rho^-(\xi)]<\infty,\\
    &\mathbb E_x[e^{-\rho T}Y_T^\xi]<\infty,
    \qquad
    \mathbb E_x\left[
    \int_0^T e^{-\rho t}\Pi(X_t,Y_t^\xi)\,dt
    \right]<\infty,
    \end{aligned}
    \end{equation}
    and each of the four terms in \eqref{eq:main-problem} is separately
    integrable.

    \item[\textnormal{(iii)}] The value function is finite:
    \begin{equation}\label{eq:U-finite}
        -\infty<\mathfrak U(x,y)
        \le K^-y+\mathbb E_x\left[\int_0^Te^{-\rho t}\Psi(X_t)\,dt\right]
        <\infty .
    \end{equation}
\end{enumerate}
In particular, controls with \({\cal J}=-\infty\) entail no loss in the
maximization problem.
\end{lem}

\begin{proof}
Write
\[
I_\Psi:=\int_0^T e^{-\rho t}\Psi(X_t)\,dt
\]
and \(\Delta K:=K^+-K^->0\), and split \(Z(\xi)=K^-y+U^\xi-V^\xi\) with
\[
    F_t^\xi:=\Pi(X_t,Y_t^\xi)-\rho K^-Y_t^\xi,
    \qquad
    U^\xi:=\int_0^T e^{-\rho t}(F_t^\xi)^+\,dt,
    \qquad
    V^\xi:=\int_0^T e^{-\rho t}(F_t^\xi)^-\,dt+\Delta K A_\rho^+(\xi).
\]
As in \eqref{eq:Z-upper-bound}, \(0\le U^\xi\le I_\Psi\), so
\(\mathbb E_x[U^\xi]<\infty\) by
Assumption~\ref{as:global}\textnormal{(ii)}, and the condition
\({\cal J}>-\infty\) is equivalent to \(\mathbb E_x[V^\xi]<\infty\).

We first prove \textnormal{(i)}.  Since
\(0\le \Delta K A_\rho^+(\xi)\le V^\xi\), finiteness of
\(\mathbb E_x[V^\xi]\) forces \(\mathbb E_x[A_\rho^+(\xi)]<\infty\).
Conversely, \(\Pi\ge0\) gives \((F_t^\xi)^-\le\rho K^-Y_t^\xi\), while
\(A_\rho^+(\xi)\ge e^{-\rho T}\xi_T^+\) gives the pathwise bound
\begin{equation}\label{eq:capacity-pathwise-bound}
    0<Y_t^\xi\le y+\xi_T^+\le y+e^{\rho T}A_\rho^+(\xi),
    \qquad t\in[0,T].
\end{equation}
Hence
\begin{equation}\label{eq:V-upper-bound}
    V^\xi
    \le K^-\bigl(1-e^{-\rho T}\bigr)
        \bigl(y+e^{\rho T}A_\rho^+(\xi)\bigr)
      +\Delta K A_\rho^+(\xi),
\end{equation}
so that \(\mathbb E_x[A_\rho^+(\xi)]<\infty\) implies
\(\mathbb E_x[V^\xi]<\infty\).  This proves \textnormal{(i)}, and
\eqref{eq:capacity-pathwise-bound} also gives
\begin{equation}\label{eq:capacity-time-integrability}
    \mathbb E_x\left[\int_0^T e^{-\rho t}Y_t^\xi\,dt\right]<\infty .
\end{equation}

For \textnormal{(ii)}, the first identity in
\eqref{eq:discounted-capacity-ibp}, together with \(Y^\xi>0\), yields the
pathwise bounds
\begin{equation}\label{eq:downward-terminal-from-upward}
    A_\rho^-(\xi)\le y+A_\rho^+(\xi),
    \qquad
    e^{-\rho T}Y_T^\xi\le y+A_\rho^+(\xi),
\end{equation}
which give the first two assertions in
\eqref{eq:finite-payoff-control-integrability}.  The third follows from the
envelope bound \eqref{eq:revenue-envelope-bound} together with
\eqref{eq:capacity-time-integrability}.  Each of the four terms in
\eqref{eq:main-problem} is then separately integrable.

For \textnormal{(iii)}, the no-adjustment control has finite payoff by
\eqref{eq:baseline-revenue-integrability}, so \(\mathfrak U(x,y)>-\infty\),
while for every \(\xi\in\Xi(y)\),
\[
    {\cal J}(x,y;\xi^+,\xi^-)
    \le
    K^-y+\mathbb E_x[U^\xi]
    \le K^-y+\mathbb E_x[I_\Psi]<\infty .
\]
\end{proof}

\subsection{Switching controls and level-wise payoffs}

A switching control is a sequence
\(\eta=(\tau_n,\kappa_n)_{n\ge0}\) with the following properties.  Each
\(\tau_n\) is an \(\mathbb F\)-stopping time taking values in \([0,T]\).
The variable \(\kappa_0\) is the initial regime before a possible time-zero
adjustment, whereas \(\kappa_n\), \(n\ge1\), is the regime immediately after
the switch at \(\tau_n\).  We use \(T\) as a cemetery time and impose
\[
    0=\tau_0\le\tau_1\le T,
    \qquad
    \tau_{n+1}>\tau_n
    \quad\text{on }\{\tau_n<T\},\quad n\ge1,
\]
and \(\tau_{n+1}=T\) on \(\{\tau_n=T\}\).  Moreover,
\(\kappa_n\in\{0,1\}\) is \({\cal F}_{\tau_n}\)-measurable and
\[
    \kappa_n=1-\kappa_{n-1}
    \quad\text{on }\{\tau_n<T\},
    \qquad n\ge1.
\]
Thus \(\tau_1=\tau_0=0\) is allowed, but after this possible initial switch
all switching times before \(T\) are strictly increasing.  No switch is made
at \(T\).

The admissible switching class is
\begin{equation}\label{eq:Gamma set condition}
\Gamma
:=
\left\{
\eta=(\tau_n,\kappa_n)_{n\ge0}:
\mathbb E_x\left[
\sum_{n\ge1}e^{-\rho\tau_n}{\bf 1}_{\{\tau_n<T\}}
\right]<\infty
\right\},
\end{equation}
where the structural conditions above are understood.
Although the admissible class depends on the initial state
\(x\) through \(\mathbb E_x\), we simply write \(\Gamma\) rather than
\(\Gamma(x)\).  Every
\(\eta\in\Gamma\) has only finitely many switches before \(T\), almost surely,
because
\[
e^{-\rho T}\,
\mathbb E_x\left[
\sum_{n\ge1}{\bf 1}_{\{\tau_n<T\}}
\right]
\le
\mathbb E_x\left[
\sum_{n\ge1}e^{-\rho\tau_n}{\bf 1}_{\{\tau_n<T\}}
\right]<\infty.
\]

Given \(\eta\in\Gamma\), its LCRL regime indicator is
\begin{equation}\label{eq:regime_indcator}
\zeta_0:=\kappa_0,
\qquad
\zeta_t
:=
\sum_{n=0}^{\infty}
\kappa_n{\bf 1}_{\{\tau_n<t\le\tau_{n+1}\}},
\qquad 0<t\le T.
\end{equation}
At a positive switching time \(\tau_n<T\),
\(\zeta_{\tau_n}=\kappa_{n-1}\) and
\(\zeta_{\tau_n+}=\kappa_n\).  If \(\tau_1=0\), then
\(\zeta_0=\kappa_0\) and \(\zeta_{0+}=\kappa_1\).  Thus the switching cost is
charged at \(\tau_n\), while the new regime applies immediately to its right.

For a fixed level \(z>0\) and \(\eta\in\Gamma\), define
\begin{equation}\label{eq:P-plus-minus}
\begin{aligned}
{\cal P}^+(x,z;\eta)
:=
\mathbb E_x\bigg[
&\int_0^T e^{-\rho t}\pi(X_t,z)\zeta_t\,dt
+e^{-\rho T}K^-\zeta_T\\
&+\sum_{n=1}^{\infty}e^{-\rho\tau_n}
\Bigl(
-K^+{\bf 1}_{\{\kappa_n=1\}}
+K^-{\bf 1}_{\{\kappa_n=0\}}
\Bigr){\bf 1}_{\{\tau_n<T\}}
\bigg],
\end{aligned}
\end{equation}
and
\begin{equation}\label{eq:P-minus}
\begin{aligned}
{\cal P}^-(x,z;\eta)
:=
\mathbb E_x\bigg[
&\int_0^T e^{-\rho t}\pi(X_t,z)(\zeta_t-1)\,dt
+e^{-\rho T}K^-(\zeta_T-1)\\
&+\sum_{n=1}^{\infty}e^{-\rho\tau_n}
\Bigl(
-K^+{\bf 1}_{\{\kappa_n=1\}}
+K^-{\bf 1}_{\{\kappa_n=0\}}
\Bigr){\bf 1}_{\{\tau_n<T\}}
\bigg].
\end{aligned}
\end{equation}
Regime \(1\) means that level \(z\) is active, while regime \(0\) means that
it is inactive.  The functional \({\cal P}^+\) is measured from the inactive
benchmark and \({\cal P}^-\) from the active benchmark.  These functionals are
finite for every \(\eta\in\Gamma\), by Assumption~\ref{as:global} and the
definition of \(\Gamma\).

\begin{pr}\label{pr:OSP}
For \(z>0\) and \(j\in\{0,1\}\), define
\[
{\bf P}_j^+(x,z)
:=
\sup_{\eta\in\Gamma,\,\kappa_0=j}{\cal P}^+(x,z;\eta),
\qquad
{\bf P}_j^-(x,z)
:=
\sup_{\eta\in\Gamma,\,\kappa_0=j}{\cal P}^-(x,z;\eta).
\]
\end{pr}

The two switching problems differ by an additive benchmark independent of
the switching rule.  Define the always-active, no-switching benchmark
\begin{equation}\label{eq:R0-switching-benchmark}
{\cal R}_0(t,x,z)
:=
\mathbb E_{t,x}\bigg[
\int_t^T e^{-\rho(s-t)}\pi(X_s,z)\,ds
+e^{-\rho(T-t)}K^-
\bigg].
\end{equation}
Then, for every \(\eta\in\Gamma\),
\begin{equation}\label{eq:plus-minus-benchmark}
{\cal P}^+(x,z;\eta)-{\cal P}^-(x,z;\eta)
={\cal R}_0(0,x,z).
\end{equation}
Consequently, the \(+\)- and \(-\)-problems have the same optimizers and
\begin{equation}\label{eq:Pplus-Pminus-values}
{\bf P}_j^-(x,z)
={\bf P}_j^+(x,z)-{\cal R}_0(0,x,z),
\qquad j=0,1.
\end{equation}

An elementary pairing of consecutive switches shows, pathwise, that the
switching-cost sum together with the terminal term is bounded above by
\(K^-\kappa_0\). Consequently,
\[
0\le {\bf P}_0^+(x,z)
\le {\cal R}_0(0,x,z)-e^{-\rho T}K^-,
\qquad
0\le {\bf P}_1^-(x,z)
\le K^-\bigl(1-e^{-\rho T}\bigr),
\]
so the switching values are finite. Their Borel measurability in \(z\)
follows from the standard finite-switch approximation.

\subsection{Level-crossing decomposition of singular controls}

Fix \(\xi\in\Xi(y)\) with
\({\cal J}(x,y;\xi^+,\xi^-)>-\infty\).  The level-set
change-of-variables theorem for finite-variation paths applies pathwise for
almost every level.  Tonelli's theorem therefore provides a deterministic
Lebesgue-null set \(N_\xi\subset(0,\infty)\) such that, for every
\(z\notin N_\xi\), the process
\begin{equation}\label{eq:singular-level-indicator}
\begin{aligned}
\zeta_0^\xi(z)&:={\bf 1}_{\{y\ge z\}},\qquad
\zeta_t^\xi(z)&:=\lim_{s\uparrow t}{\bf 1}_{\{Y_s^\xi>z\}},
&&0<t\le T,
\end{aligned}
\end{equation}
is an LCRL finite-variation \(\{0,1\}\)-valued process almost surely; see
\citet[Proposition~2.9]{Guo2008}.  Its right-jump times and successive regimes
define a switching control \(\eta^\xi(z)\).  If
\(\zeta_{0+}^\xi(z)\ne\zeta_0^\xi(z)\), its first switch is recorded as
\(\tau_1^\xi(z)=0\).  No switch is introduced at \(T\).  On \(N_\xi\), use
the no-switch rule with initial regime \({\bf 1}_{\{y\ge z\}}\).

Applying the change-of-variables formula of
\citet[Lemma~2.13]{Guo2008}
to the constant extensions of \(\xi^+\) and \(\xi^-\)
beyond \(T\), with the c\`adl\`ag integrand
\[
    g(t)=e^{-\rho t}{\bf 1}_{\{t<T\}},
    \qquad t\in[0,\infty),
    \qquad g(\infty)=0,
\]
gives
\begin{equation}\label{eq:change-variable}
\begin{aligned}
A_\rho^+(\xi)
&=
\int_0^\infty
\sum_{\substack{n\ge1\\ \kappa_n^\xi(z)=1}}
e^{-\rho\tau_n^\xi(z)}
{\bf 1}_{\{\tau_n^\xi(z)<T\}}\,dz,\qquad
A_\rho^-(\xi)
&=
\int_0^\infty
\sum_{\substack{n\ge1\\ \kappa_n^\xi(z)=0}}
e^{-\rho\tau_n^\xi(z)}
{\bf 1}_{\{\tau_n^\xi(z)<T\}}\,dz.
\end{aligned}
\end{equation}
Adding the two identities and applying Tonelli's theorem yields
\begin{equation}\label{eq:aggregate-switch-count}
\begin{aligned}
&\int_0^\infty
\mathbb E_x\bigg[
\sum_{n\ge1}e^{-\rho\tau_n^\xi(z)}
{\bf 1}_{\{\tau_n^\xi(z)<T\}}
\bigg]dz
=
\mathbb E_x\bigl[A_\rho^+(\xi)+A_\rho^-(\xi)\bigr]
<\infty,
\end{aligned}
\end{equation}
where the last inequality follows from
Lemma~\ref{lem:finite-payoff-integrability}. Consequently,
\(\eta^\xi(z)\in\Gamma\) for Lebesgue-almost every \(z>0\).

Since \(Y^\xi\) is left-continuous, the terminal indicator in
\eqref{eq:singular-level-indicator} agrees with
\({\bf 1}_{\{Y_T^\xi>z\}}\) for Lebesgue-almost every \(z>0\).  The
layer-cake identities therefore give
\begin{equation}\label{eq:zeta_terminal}
\begin{aligned}
e^{-\rho T}K^-Y_T^\xi
&=e^{-\rho T}K^-y
+\int_0^y e^{-\rho T}K^-
\bigl(\zeta_T^\xi(z)-1\bigr)\,dz\\
&\quad
+\int_y^\infty e^{-\rho T}K^-\zeta_T^\xi(z)\,dz,
\end{aligned}
\end{equation}
and, for \(dt\otimes d\mathbb P\)-almost every \((t,\omega)\),
\begin{equation}\label{eq:zeta-t}
\begin{aligned}
e^{-\rho t}\Pi(X_t,Y_t^\xi)
&=e^{-\rho t}\Pi(X_t,y)
+\int_0^y e^{-\rho t}\pi(X_t,z)
\bigl(\zeta_t^\xi(z)-1\bigr)\,dz\\
&\quad
+\int_y^\infty e^{-\rho t}\pi(X_t,z)
\zeta_t^\xi(z)\,dz.
\end{aligned}
\end{equation}
The required absolute integrability follows from
\begin{align*}
&\int_y^\infty \pi(X_t,z)\zeta_t^\xi(z)\,dz
+\int_0^y \pi(X_t,z)
\bigl(1-\zeta_t^\xi(z)\bigr)\,dz
=\bigl|\Pi(X_t,Y_t^\xi)-\Pi(X_t,y)\bigr|,\\
&\int_y^\infty\zeta_T^\xi(z)\,dz
+\int_0^y\bigl(1-\zeta_T^\xi(z)\bigr)\,dz
=|Y_T^\xi-y|,
\end{align*}
Lemma~\ref{lem:finite-payoff-integrability}, and
\eqref{eq:aggregate-switch-count}.  Hence Tonelli--Fubini applies, and
\eqref{eq:change-variable}--\eqref{eq:zeta-t} yield
\begin{equation}\label{eq:Guo-decomposition}
\begin{aligned}
{\cal J}(x,y;\xi^+,\xi^-)
&=
\mathbb E_x\bigg[
\int_0^T e^{-\rho t}\Pi(X_t,y)\,dt
+e^{-\rho T}K^-y
\bigg]\\
&\quad
+\int_y^\infty
{\cal P}^+(x,z;\eta^\xi(z))\,dz
+\int_0^y
{\cal P}^-(x,z;\eta^\xi(z))\,dz.
\end{aligned}
\end{equation}
In particular,
\begin{equation}\label{eq:singular-switching-upper-bound}
\begin{aligned}
{\cal J}(x,y;\xi^+,\xi^-)
&\le
\mathbb E_x\bigg[
\int_0^T e^{-\rho t}\Pi(X_t,y)\,dt
+e^{-\rho T}K^-y
\bigg]\\
&\quad
+\int_y^\infty{\bf P}_0^+(x,z)\,dz
+\int_0^y{\bf P}_1^-(x,z)\,dz.
\end{aligned}
\end{equation}
Controls with payoff \(-\infty\) satisfy this upper bound trivially and are
irrelevant by Lemma~\ref{lem:finite-payoff-integrability}.  If a control in
\(\Xi(y)\) has finite payoff and its consistent level-crossing family attains
the two switching values for almost every level, then equality holds in
\eqref{eq:singular-switching-upper-bound}.  The construction and verification
of such a control are carried out in Section~\ref{sec:verification}.

\subsection{Reduction to a double-obstacle problem}

For \(0\le t\le T\), let
\({\cal Q}_j^+(t,x,z)\) denote the time-\(t\) version of
\({\bf P}_j^+(x,z)\), with switching restricted to \([t,T]\) and
discounting relative to \(t\). Thus
\({\cal Q}_j^+(0,x,z)={\bf P}_j^+(x,z)\). The HJB reduction below is used
at this stage to identify the scalar candidate problem; the value-function
identification is verified in Section~\ref{sec:verification}.

For fixed \(z>0\), the HJB system associated with the \(+\)-switching problem
is stated in terms of the discounted infinitesimal generator of \(X\),
\begin{equation}\label{eq:discounted-generator}
{\cal L} f(x)
=
\frac12
\sum_{i,j=1}^n
a_{ij}x_i x_j\,\partial_{x_i x_j}f(x)
+
\sum_{i=1}^n
\mu_i x_i\,\partial_{x_i}f(x)
-\rho f(x),
\end{equation}
whose zeroth-order term already contains \(-\rho\).  The system reads
\begin{equation}\label{eq:system-VI}
\begin{aligned}
&\max\Bigl\{
\partial_t{\cal Q}_0^+ +{\cal L}{\cal Q}_0^+,
{\cal Q}_1^+-{\cal Q}_0^+-K^+
\Bigr\}=0,\\[1mm]
&\max\Bigl\{
\partial_t{\cal Q}_1^+ +{\cal L}{\cal Q}_1^++\pi(x,z),
{\cal Q}_0^+-{\cal Q}_1^++K^-
\Bigr\}=0,
\end{aligned}
\end{equation}
on \({\cal D}_T=[0,T)\times(0,\infty)^n\), where the value functions are
evaluated at \((t,x,z)\), with terminal conditions
\begin{equation}\label{eq:switching-terminal-conditions}
{\cal Q}_0^+(T,x,z)=0,
\qquad
{\cal Q}_1^+(T,x,z)=K^-.
\end{equation}

Suppose a sufficiently regular pair
\((\mathcal Q_0^+,\mathcal Q_1^+)\) satisfies
\eqref{eq:system-VI}--\eqref{eq:switching-terminal-conditions}, and define
\[
    {\cal Q}:={\cal Q}_1^+-{\cal Q}_0^+.
\]
Then the obstacle inequalities and complementarity conditions imply
\(K^-\le{\cal Q}\le K^+\) and yield the scalar double-obstacle problem
\begin{equation}\label{eq:double-obs-n-dim}
\begin{cases}
\partial_t{\cal Q}+{\cal L}{\cal Q}+\pi(x,z)=0,
&K^-<{\cal Q}<K^+,\\[1mm]
\partial_t{\cal Q}+{\cal L}{\cal Q}+\pi(x,z)\ge0,
&{\cal Q}=K^+,\\[1mm]
\partial_t{\cal Q}+{\cal L}{\cal Q}+\pi(x,z)\le0,
&{\cal Q}=K^-,\\[1mm]
{\cal Q}(T,x,z)=K^-.
\end{cases}
\end{equation}
The upper obstacle corresponds to the cost of switching from regime \(0\) to
regime \(1\), and the lower obstacle to the salvage value received when
switching from regime \(1\) to regime \(0\).

The implication above is one-way: the difference \({\cal Q}\) does not by
itself determine the common additive component of
\((\mathcal Q_0^+,\mathcal Q_1^+)\).  Section~\ref{sec:verification} recovers
the pair from the strong solution of \eqref{eq:double-obs-n-dim} by
Feynman--Kac potentials and then verifies both the switching problems and the
original singular control problem.

\section{Parabolic Double Obstacle Problem}\label{sec:pdo}

We first fix $z>0$ and treat it as a parameter.  When no confusion can
arise, we write ${\cal Q}(t,x)$ and $\pi(x)$ in place of
${\cal Q}(t,x,z)$ and $\pi(x,z)$.

Let $(\Omega_N)_{N\geq1}$ be a nested sequence of smooth bounded
domains such that $\overline{\Omega_N}\Subset(0,\infty)^n$ and
$\bigcup_{N\geq1}\Omega_N=(0,\infty)^n$, and set
${\cal D}_T^N=[0,T)\times\Omega_N$. We consider the bounded-domain problem
\begin{equation}\label{eq:double-bdd}
\begin{cases}
 \partial_t{\cal Q}^N+{\cal L}{\cal Q}^N+\pi=0
     &\text{on }\{K^-<{\cal Q}^N<K^+\},\\
 \partial_t{\cal Q}^N+{\cal L}{\cal Q}^N+\pi\geq0
     &\text{on }\{{\cal Q}^N=K^+\},\\
 \partial_t{\cal Q}^N+{\cal L}{\cal Q}^N+\pi\leq0
     &\text{on }\{{\cal Q}^N=K^-\},\\
 {\cal Q}^N=K^-&\text{on }([0,T)\times\partial\Omega_N)
                         \cup(\{T\}\times\Omega_N).
\end{cases}
\end{equation}
The artificial lateral condition is used only to construct an interior limit;
no boundary condition is imposed on the state space $(0,\infty)^n$.

Set
\[
 C_1=\rho K^-,\qquad C_{2,N}=\|\pi\|_{L^\infty(\Omega_N)}.
\]
Choose smooth penalty functions $\beta_1$ and $\beta_2$ satisfying
\begin{equation*}
\begin{aligned}
 &\beta_1\leq0,\quad \beta_1'\geq0,\quad \beta_1''\leq0,
 \quad \beta_1(0)=-C_1,\quad \beta_1(\xi)=0\quad(\xi>1),\\
 &\beta_2\geq0,\quad \beta_2'\geq0,\quad \beta_2''\geq0,
 \quad \beta_2(0)=C_{2,N},\quad \beta_2(\xi)=0\quad(\xi<-1).
\end{aligned}
\end{equation*}
For $0<\varepsilon<K^+-K^-$, let
$\beta_{i,\varepsilon}(\xi)=\beta_i(\xi/\varepsilon)$ and consider
\begin{equation}\label{eq:penal}
\begin{cases}
 \partial_t{\cal Q}^{N,\varepsilon}+{\cal L}{\cal Q}^{N,\varepsilon}
       +\pi
 =\beta_{1,\varepsilon}({\cal Q}^{N,\varepsilon}-K^-)
  +\beta_{2,\varepsilon}({\cal Q}^{N,\varepsilon}-K^+)
       &\text{in }{\cal D}_T^N,\\
 {\cal Q}^{N,\varepsilon}=K^-&\text{on }
 ([0,T)\times\partial\Omega_N)\cup(\{T\}\times\Omega_N).
\end{cases}
\end{equation}

\begin{lem}\label{lem:QNe}
For every $N\geq1$, $0<\varepsilon<K^+-K^-$, and $1<p<\infty$,
problem~\eqref{eq:penal} has a unique solution
${\cal Q}^{N,\varepsilon}\in W^{1,2}_p({\cal D}_T^N)$.  Moreover,
\[
 \partial_t{\cal Q}^{N,\varepsilon}\leq0,\qquad
 K^-\leq{\cal Q}^{N,\varepsilon}\leq K^+
 \quad\text{a.e. in }{\cal D}_T^N,
\]
and
\begin{equation}\label{eq:W12-QNe}
 \|{\cal Q}^{N,\varepsilon}\|_{W^{1,2}_p({\cal D}_T^N)}
 \leq C_{N,p}\bigl(C_{2,N}+K^+\bigr),
\end{equation}
where $C_{N,p}$ is independent of $\varepsilon$.
\end{lem}
\begin{proof}
Since $\overline{\Omega_N}\Subset(0,\infty)^n$ and $a$ is positive
definite, ${\cal L}$ is uniformly elliptic on $\Omega_N$.
Hence~\eqref{eq:penal} is a uniformly parabolic semilinear Dirichlet problem
on the smooth bounded cylinder ${\cal D}_T^N$. Truncate the two
penalties at levels $\pm M$, with $M>\max\{C_1,C_{2,N}\}$.  Standard
semilinear parabolic theory gives a unique $W^{1,2}_p$ solution of the
truncated problem.

The parabolic maximum principle gives bounds independent of $M$.
Indeed, an interior minimum below $K^-$ would give
$-\rho K^-<\partial_t{\cal Q}^{N,\varepsilon,M}
+{\cal L}{\cal Q}^{N,\varepsilon,M}+\pi
\leq\beta_{1,\varepsilon,M}(0)=-\rho K^-$, a contradiction.
The corresponding maximum argument, using
$\pi\leq C_{2,N}$, rules out values above $K^+$.  Consequently the truncated
penalties never reach $\pm M$, so the solution solves~\eqref{eq:penal}.  On the
interval $[K^-,K^+]$ the two penalty terms are bounded respectively by $C_1$
and $C_{2,N}$.  The global $W^{1,2}_p$ estimate on the smooth cylinder now
gives~\eqref{eq:W12-QNe}; comparison gives uniqueness.

It remains to prove the time monotonicity without differentiating at the
terminal--lateral corner.  For $h\in(0,T)$ set
${\cal Q}_h(t,x)={\cal Q}^{N,\varepsilon}(t+h,x)$ on
$[0,T-h]\times\Omega_N$.  The equation is autonomous.
On the lateral boundary ${\cal Q}_h={\cal Q}^{N,\varepsilon}=K^-$, while at
$t=T-h$, ${\cal Q}_h(T-h,\cdot)=K^-\leq
{\cal Q}^{N,\varepsilon}(T-h,\cdot)$.
Since $\beta_{1,\varepsilon}+\beta_{2,\varepsilon}$ is nondecreasing, the
comparison principle yields ${\cal Q}_h\leq{\cal Q}^{N,\varepsilon}$.  Dividing
by $h$ and letting $h\downarrow0$ proves
$\partial_t{\cal Q}^{N,\varepsilon}\leq0$ a.e.

\end{proof}

\begin{pro}\label{lem:QN}
For every $N\geq1$ and $1<p<\infty$, problem~\eqref{eq:double-bdd} has a unique
solution ${\cal Q}^N\in W^{1,2}_p({\cal D}_T^N)$.  Moreover,
\begin{equation}\label{eq:mono-QN}
 \partial_t{\cal Q}^N\leq0,\qquad K^-\leq{\cal Q}^N\leq K^+
 \quad\text{a.e. in }{\cal D}_T^N.
\end{equation}
\end{pro}
\begin{proof}
Since estimate~\eqref{eq:W12-QNe} holds for every $p$, choosing $p$
sufficiently large and using parabolic compactness give, along a sequence
$\varepsilon\downarrow0$,
${\cal Q}^{N,\varepsilon}\rightharpoonup{\cal Q}^N$ in
$W^{1,2}_p({\cal D}_T^N)$ and
${\cal Q}^{N,\varepsilon}\longrightarrow{\cal Q}^N$ in
$C^{0,1}_{\alpha}(\overline{{\cal D}_T^N})$ for every $0<\alpha<1$.
The bounds and time monotonicity pass to the limit.

On every compact subset of $\{{\cal Q}^N>K^-\}$, the lower penalty vanishes for
all sufficiently small $\varepsilon$; the upper penalty is nonnegative.
Similarly, on compact subsets of $\{{\cal Q}^N<K^+\}$, the upper penalty
vanishes and the lower penalty is nonpositive.  Passing to the weak limit gives
\begin{equation}\label{eq:double-bdd-vi}
\begin{cases}
 \partial_t{\cal Q}^N+{\cal L}{\cal Q}^N+\pi\geq0
       &\text{a.e. on }\{{\cal Q}^N>K^-\},\\
 \partial_t{\cal Q}^N+{\cal L}{\cal Q}^N+\pi\leq0
       &\text{a.e. on }\{{\cal Q}^N<K^+\}.
\end{cases}
\end{equation}
Hence ${\cal Q}^N$ solves~\eqref{eq:double-bdd}.  The uniform bounds on the
penalty terms pass to the limit, and hence
$-\rho K^-\leq\partial_t{\cal Q}^N+{\cal L}{\cal Q}^N+\pi
\leq C_{2,N}$ a.e. in ${\cal D}_T^N$.
The parabolic $L^p$ estimate therefore yields
${\cal Q}^N\in W^{1,2}_p({\cal D}_T^N)$ for every $1<p<\infty$.
Uniqueness follows from the bounded-domain comparison principle.
\end{proof}

\begin{thm}\label{thm:Q}
There exists a unique bounded strong solution
\[
 {\cal Q}\in W^{1,2}_{p,\mathrm{loc}}({\cal D}_T)
       \cap C([0,T]\times(0,\infty)^n)
       \quad\text{for every }1<p<\infty
\]
of~\eqref{eq:double-obs-n-dim}.
Moreover,
\begin{equation}\label{eq:mono-Q}
 \partial_t{\cal Q}\leq0,\qquad K^-\leq{\cal Q}\leq K^+
 \quad\text{a.e. in }{\cal D}_T.
\end{equation}
\end{thm}
\begin{proof}
Let ${\cal Q}^N$ be the solution from Proposition~\ref{lem:QN}.
Since the derivatives of ${\cal Q}^N-K^\pm$ vanish a.e. on their
zero sets,~\eqref{eq:double-bdd-vi} gives
\[
 \partial_t{\cal Q}^N+{\cal L}{\cal Q}^N
 =-\rho K^-\chi_{\{{\cal Q}^N=K^-\}}
  -\pi\chi_{\{K^-<{\cal Q}^N<K^+\}}
  -\rho K^+\chi_{\{{\cal Q}^N=K^+\}}
 \quad\text{a.e.}
\]
Thus
$|\partial_t{\cal Q}^N+{\cal L}{\cal Q}^N|\leq\pi+\rho K^+$ a.e.
For every smooth bounded domain $G\Subset(0,\infty)^n$ and all sufficiently
large $N$, the local $W^{1,2}_p$ estimate up to the terminal boundary gives
\[
 \|{\cal Q}^N\|_{W^{1,2}_p((0,T)\times G)}\leq C_{G,p},
\]
where $C_{G,p}$ is independent of $N$.
Fix $p>n+2$. For each $G\Subset(0,\infty)^n$, the preceding estimate
and compactness yield a subsequence converging weakly in
$W^{1,2}_p((0,T)\times G)$ and uniformly on $[0,T]\times\overline G$. A
diagonal argument gives a function ${\cal Q}$ such that
${\cal Q}^N\rightharpoonup{\cal Q}$ in
$W^{1,2}_{p,\mathrm{loc}}({\cal D}_T)$ and
${\cal Q}^N\longrightarrow{\cal Q}$ locally uniformly in ${\cal D}_T$.
Since the estimate holds for every $1<p<\infty$, weak compactness
and the local uniform convergence also give
${\cal Q}\in W^{1,2}_{p,\mathrm{loc}}({\cal D}_T)$ for every such $p$.
The argument used for~\eqref{eq:double-bdd-vi} passes the two variational
inequalities to the limit, and~\eqref{eq:mono-QN} gives~\eqref{eq:mono-Q}.
Moreover, the convergence is locally uniform on
$[0,T]\times(0,\infty)^n$. Since ${\cal Q}^N(T,\cdot)=K^-$, we have
${\cal Q}(T,\cdot)=K^-$, and hence ${\cal Q}$ solves
\eqref{eq:double-obs-n-dim} with the asserted continuity up to $t=T$.

It remains to prove uniqueness on the unbounded state space. Let
${\cal Q}_1$ and ${\cal Q}_2$ be two bounded solutions and set
\begin{equation*}
 \psi(x)=\sum_{i=1}^n\bigl(x_i+x_i^{-1}\bigr),
 \qquad
 \varphi(t,x)=e^{\lambda(T-t)}\psi(x).
\end{equation*}
A direct computation gives
\begin{equation*}
 {\cal L}\psi
 =\sum_{i=1}^n\mu_i x_i
  +\sum_{i=1}^n(a_{ii}-\mu_i)x_i^{-1}-\rho\psi.
\end{equation*}
Hence ${\cal L}\psi\leq C\psi$ for some constant $C$. Fixing
$\lambda>C$, we obtain
$(\partial_t+{\cal L})\varphi
=e^{\lambda(T-t)}({\cal L}\psi-\lambda\psi)<0$ on
$[0,T)\times(0,\infty)^n$.
Since $\psi(x)$ tends to infinity whenever $x$ leaves every compact
subset of $(0,\infty)^n$, the positivity set of
${\cal Q}_1-{\cal Q}_2-\varepsilon\varphi$ is relatively compact in the
spatial variables and is empty at $t=T$. On this set,
${\cal Q}_1>K^-$ and ${\cal Q}_2<K^+$, and hence the variational inequalities
imply
$(\partial_t+{\cal L}){\cal Q}_1\geq-\pi
\geq(\partial_t+{\cal L}){\cal Q}_2$.
Consequently,
$(\partial_t+{\cal L})({\cal Q}_1-{\cal Q}_2-\varepsilon\varphi)>0$ on its
positivity set. The parabolic maximum principle makes this set empty. Letting
$\varepsilon\downarrow0$ gives ${\cal Q}_1\leq{\cal Q}_2$; interchanging the
two solutions gives equality.
\end{proof}

\section{Free Boundary Analysis}\label{sec:fb}

In Theorem \ref{thm:Q}, we obtain the unique solution ${\cal Q}$
to the double obstacle problem
\eqref{eq:double-obs-n-dim}. Now, we can decompose the domain
${\cal D}_T\times(0,\infty)$ into three parts as follows:
\begin{equation*}\left\{\begin{aligned}
	\mbox{Increasing Region :}&\;{\cal IR} = \{(t,x,z)\in{\cal D}_T\times(0,\infty) \mid {\cal Q}(t, x, z)= K^+\},\\
	\mbox{Decreasing Region :}&\;{\cal DR} = \{(t,x,z)\in{\cal D}_T\times(0,\infty) \mid {\cal Q}(t, x, z)= K^-\},\\
	\mbox{Waiting Region :}&\;{\cal WR} = \{(t,x,z)\in{\cal D}_T\times(0,\infty) \mid K^- < {\cal Q}(t, x, z) <  K^+\}.
\end{aligned}\right.\end{equation*}

To characterize the free boundaries
$\partial {\cal IR} \cap \bigl({\cal D}_T\times(0,\infty)\bigr)$
and
$\partial {\cal DR} \cap \bigl({\cal D}_T\times(0,\infty)\bigr)$,
we require certain properties of ${\cal Q}$, such as monotonicity. We therefore
introduce the following additional assumptions on
$\pi = \partial_z \Pi$. Throughout Section \ref{sec:fb}, we adopt
Assumption \ref{as:free_boundary}.

\begin{as}\label{as:free_boundary}
\begin{enumerate}
    \item[(i)] \emph{Regularity and monotonicity.}
    The function
    \[
        \pi:(0,\infty)^n\times(0,\infty)\longrightarrow(0,\infty)
    \]
    is locally Lipschitz continuous. For every $z>0$, the map
    $x\mapsto\pi(x,z)$ is non-decreasing in each coordinate, and for every
    $x\in(0,\infty)^n$, the map $z\mapsto\pi(x,z)$ is strictly decreasing.

    \item[(ii)] \emph{Spatial control.}
    There exists a constant $C_\pi>0$ such that
    \[
        0\le x_i\partial_{x_i}\pi(x,z)\le C_\pi\pi(x,z)
        \quad\text{a.e.},\qquad i=1,\ldots,n,
    \]
    Moreover, there exists a positive $C^2$ function
    $\widehat\pi$ such that, for every $z>0$,
    \[
        \pi\le\widehat\pi\le C_\pi\pi.
    \]
    and, for every $i,j=1,\ldots,n$,
    \[
        |x_i\partial_{x_i}\widehat\pi(x,z)|+
        |x_ix_j\partial_{x_ix_j}\widehat\pi(x,z)|
        \le C_\pi\widehat\pi(x,z).
    \]

    \item[(iii)] \emph{Capacity regularity.}
    There exists a fixed exponent $m>0$, independent of $I$, such that,
    for every compact interval $I\Subset(0,\infty)$, there is a constant
    $C_I>0$ satisfying
    \[
        |\partial_z\pi(x,z)|
        \le C_I\left(1+\sum_{i=1}^n x_i^m\right)
        \quad\text{for a.e. }(x,z)\in(0,\infty)^n\times I.
    \]

    \item[(iv)] \emph{Boundary behavior.}
    The limits
    \[
        \lim_{z\downarrow0}\pi(x,z)=\infty,
        \qquad
        \lim_{z\to\infty}\pi(x,z)=0
    \]
    hold locally uniformly in $x\in(0,\infty)^n$.
\end{enumerate}
\end{as}

The first estimate controls the spatial growth of $\pi$. The smooth majorant
$\widehat\pi$ is introduced for the comparison arguments used later in this
section.
Since the coefficients of ${\cal L}$ in
\eqref{eq:discounted-generator} are $\tfrac12a_{ij}x_ix_j$, $\mu_ix_i$ and
$-\rho$, the weighted derivative bounds in \textnormal{(ii)} imply, after
increasing $C_\pi$ if necessary,
\[
    {\cal L}\widehat\pi(\cdot,z)
    \le C_\pi\widehat\pi(\cdot,z),
    \qquad z>0.
\]

\begin{rem}\label{rem:free-boundary-examples}
The specifications of Remark~\ref{rem:Pi} satisfy
Assumption~\ref{as:free_boundary}.

\emph{Direct CES.} Put $S=\sum_i\beta_ix_i^\theta+\beta_yz^\theta$ and let
$s_i:=\beta_ix_i^\theta/S$ and $s_y:=\beta_yz^\theta/S$ denote the CES shares,
all lying in $(0,1)$ and summing to one. With
$\pi=P\gamma\beta_yz^{\theta-1}S^{\gamma/\theta-1}$, a direct computation gives
\begin{equation}\label{eq:ces-euler}
    x_i\partial_{x_i}\pi=(\gamma-\theta)s_i\pi,
    \qquad
    x_ix_j\partial_{x_ix_j}\pi=(\gamma-\theta)
        \bigl[(\gamma-2\theta)s_is_j+(\theta-1)s_i{\bf 1}_{\{i=j\}}\bigr]\pi .
\end{equation}
Both right-hand sides are
bounded in absolute value by constant multiples
of $\pi$, with constants depending only on $\gamma$ and $\theta$, uniformly
in $x$ \emph{and} in $z$, so $\widehat\pi=\pi$ satisfies
\textnormal{(ii)}.
The first identity in \eqref{eq:ces-euler} shows that
$x\mapsto\pi(x,z)$ is
strictly increasing in each coordinate under the
maintained condition $\gamma>\theta$, while
\[
    \frac{z\,\partial_z\pi}{\pi}
    =(\theta-1)+(\gamma-\theta)s_y<\gamma-1<0 ,
\]
so $z\mapsto\pi(x,z)$ is
strictly decreasing under the maintained condition
$\gamma<1$; this verifies \textnormal{(i)} under the
parameter restrictions of Remark~\ref{rem:Pi}.
For $\theta>0$ one has $\pi\sim C(x)z^{\theta-1}$ as
$z\downarrow0$ and $\pi\sim Cz^{\gamma-1}$ as $z\to\infty$; for $\theta<0$
the two expansions are exchanged. Since $\theta<\gamma<1$, both limits in
\textnormal{(iv)} follow locally uniformly in $x$.
Moreover, for every compact interval
$I\Subset(0,\infty)$,
\[
    |\partial_z\pi(x,z)|
    \le C_I\pi(x,z)
    \le C_I\left(1+\sum_{i=1}^n x_i^{\gamma-\theta}\right),
    \qquad z\in I,
\]
so \textnormal{(iii)} holds with $m=\gamma-\theta>0$.

\emph{CES in non-capital factors.} Here $\pi=\delta A(x)z^{\delta-1}$ with
$A=P(\sum_i\beta_ix_i^\theta)^{\gamma/\theta}$. The same computation, now
with $S=\sum_i\beta_ix_i^\theta$, gives
$x_i\partial_{x_i}\pi=\gamma s_i\pi$ and
$x_ix_j\partial_{x_ix_j}\pi
=\gamma[(\gamma-\theta)s_is_j+(\theta-1)s_i{\bf 1}_{\{i=j\}}]\pi$, so
$\widehat\pi=\pi$ again satisfies \textnormal{(ii)}. Clauses
\textnormal{(i)}, \textnormal{(iii)} and \textnormal{(iv)} follow from
$0<\delta<1$ and $\gamma>0$.
More explicitly, \textnormal{(iii)} holds with
$m=\gamma$, since, for every $I\Subset(0,\infty)$,
\[
    |\partial_z\pi(x,z)|
    \le C_I A(x)
    \le C_I\left(1+\sum_{i=1}^n x_i^\gamma\right),
    \qquad z\in I.
\]
The perfect-complements limit
$P\min\{x_1,x_2\}^\gamma y^\delta$ is also covered.
\[
    \pi=P\delta\min\{x_1,x_2\}^\gamma z^{\delta-1},\qquad
    0\le x_i\partial_{x_i}\pi\le\gamma\pi
    \quad\text{a.e.},\quad i=1,2.
\]
As above, \textnormal{(i)}, \textnormal{(iii)} and \textnormal{(iv)} hold
with $m=\gamma$. For \textnormal{(ii)}, fix $q>0$ and set
\[
    \widehat\pi(x,z)
    =P\delta 2^{\gamma/q}
      \bigl(x_1^{-q}+x_2^{-q}\bigr)^{-\gamma/q}z^{\delta-1}.
\]
Then
\[
    \pi\le\widehat\pi\le2^{\gamma/q}\pi,
    \qquad
    |x_i\partial_{x_i}\widehat\pi|
    +|x_ix_j\partial_{x_ix_j}\widehat\pi|
    \le C_{\gamma,q}\widehat\pi,
    \qquad i,j=1,2,
\]
so \textnormal{(ii)} holds.

\emph{Production switching.} Here
$\pi=\delta\max\{x_1,x_2\}z^{\delta-1}$ is locally Lipschitz,
non-decreasing in each coordinate and strictly decreasing in $z$, so
\textnormal{(i)} holds, and \textnormal{(iii)}--\textnormal{(iv)} follow
as above with $m=1$.
Moreover,
$0\le x_i\partial_{x_i}\pi\le\pi$ a.e., for $i=1,2$.
Since $\max\{x_1,x_2\}$ is not differentiable on the diagonal, take the
majorant
\[
    \widehat\pi(x,z)=\delta\bigl(x_1^q+x_2^q\bigr)^{1/q}z^{\delta-1},
    \qquad q>1 ,
\]
which satisfies $\pi\le\widehat\pi\le2^{1/q}\pi$, is smooth on
$(0,\infty)^2$, and homogeneous of degree one in $x$.
A direct computation gives
\[
    0\le x_i\partial_{x_i}\widehat\pi\le\widehat\pi,
    \qquad
    |x_ix_j\partial_{x_ix_j}\widehat\pi|
    \le(q-1)\widehat\pi,
    \qquad i,j=1,2,
\]
and therefore
${\cal L}\widehat\pi\le C_\pi\widehat\pi$.
\end{rem}

\subsection{Monotonicity and local Lipschitz estimates}

\begin{pro}\label{pro:Q-monotonicity-lipschitz}
    Let ${\cal Q} \in W^{1,2}_{p, {\rm loc}}({\cal D}_T)$ be the unique solution to \eqref{eq:double-obs-n-dim}. Then, for some constant $C>0$, we have
    \begin{equation}\label{eq:mono-Qi}
         0\leq x_i \partial_{x_i} {\cal Q} \le C\pi(x, z), \quad\text{and}\quad
        -C \pi(x, z)\le \partial_t {\cal Q} \le 0
    \end{equation}
    in ${\cal D}_T$ for each $i=1, \dots, n$.
    For $h\ne0$ and $z\in(\max\{0,-h\},\infty)$, define
    \[
        D_z^h{\cal Q}(t,x,z)
        :=\frac{{\cal Q}(t,x,z+h)-{\cal Q}(t,x,z)}{h}.
    \]
    Then
    \begin{equation}\label{eq:Dqy<0}
        D_z^h{\cal Q}\le0\qquad\text{in }{\cal D}_T.
    \end{equation}
    Moreover, ${\cal Q}$ is locally Lipschitz continuous in $z$: for every
    compact ${\cal K}_0\subset{\cal D}_T$ and compact interval
    $I\subset(0,\infty)$, there is a constant $C_{{\cal K}_0,I}>0$ such that
    \[
        |{\cal Q}(t,x,z)-{\cal Q}(t,x,z')|
        \le C_{{\cal K}_0,I}|z-z'|
    \]
    for all $(t,x)\in{\cal K}_0$ and $z,z'\in I$.
\end{pro}
\begin{proof}
Fix \(N,\varepsilon\) and let
\(q={\cal Q}^{N,\varepsilon}\). Define
\[
 b_\varepsilon(r)
 :=\beta_{1,\varepsilon}(r-K^-)+\beta_{2,\varepsilon}(r-K^+).
\]
Since \(b_\varepsilon\) is nondecreasing,
\(b_\varepsilon(K^-)=-\rho K^-\), and
\({\cal L}\widehat\pi\leq C_\pi\widehat\pi\), the function
\[
 \overline q(t,x)
 :=K^-+(T-t)e^{C_\pi(T-t)}\widehat\pi(x,z)
\]
is a supersolution of the penalized equation and dominates its boundary
data. Comparison gives
\begin{equation}\label{eq:penal-terminal-bound}
 0\leq q(t,x)-K^-
 \leq(T-t)e^{C_\pi(T-t)}\widehat\pi(x,z).
\end{equation}
For \(0<h<T\), set
\[
 D_t^h q(t,x):=\frac{q(t+h,x)-q(t,x)}h,
 \qquad 0\leq t\leq T-h.
\]
By Lemma~\ref{lem:QNe}, \(D_t^h q\leq0\).
Subtracting the two penalized
equations gives
\((\partial_t+{\cal L}-c_h)D_t^h q=0\), where \(c_h\geq0\) is the
divided difference of \(b_\varepsilon\). The lateral boundary value of
 \(D_t^h q\) is zero. The bound~\eqref{eq:penal-terminal-bound} gives
\(D_t^h q(T-h,x)\geq-e^{C_\pi h}\widehat\pi(x,z)\).
Comparison with \(-e^{C_\pi(T-t)}\widehat\pi\) therefore yields
\[
 -e^{C_\pi(T-t)}\widehat\pi(x,z)\leq D_t^h q(t,x)\leq0.
\]
Letting \(h\downarrow0\), we obtain
\(-e^{C_\pi(T-t)}\widehat\pi\leq\partial_t q\leq0\)
a.e. in \({\cal D}_T^N\).
Passing first to the limit $\varepsilon\to0$ and then to $N\to\infty$,
as in Theorem~\ref{thm:Q}, and using
$\widehat\pi\leq C_\pi\pi$, yields
\begin{equation}\label{eq:Qt-bound-full}
    -C_\pi e^{C_\pi(T-t)}\pi(x,z)
    \leq\partial_t{\cal Q}(t,x,z)\leq0.
\end{equation}

We next prove the spatial estimate directly for the solution ${\cal Q}$ on
the whole state space. Fix $i\in\{1,\ldots,n\}$ and, for $h>0$, set
\[
    S_i^h x=(x_1,\ldots,e^h x_i,\ldots,x_n),
    \qquad
    {\cal Q}_h(t,x,z)={\cal Q}(t,S_i^h x,z),
    \qquad
    \pi_h(x,z)=\pi(S_i^h x,z).
\]
By the scaling invariance of ${\cal L}$, ${\cal Q}_h$ solves the
same double obstacle problem as ${\cal Q}$, with $\pi_h$ in place of $\pi$.
Moreover,
Assumption~\ref{as:free_boundary}\textnormal{(ii)} and Gronwall's inequality give
\begin{equation}\label{eq:pi-scaling-bound}
    0\le \pi_h(x,z)-\pi(x,z)
    \le \bigl(e^{C_\pi h}-1\bigr)\pi(x,z).
\end{equation}

Let $\lambda=C_\pi+1$ and define
$\Psi_h(t,x,z):=(e^{C_\pi h}-1)e^{\lambda(T-t)}\widehat\pi(x,z)$.
Since ${\cal L}\widehat\pi\le C_\pi\widehat\pi$, we have
\begin{equation}\label{eq:Psi-supersolution}
\begin{aligned}
    (\partial_t+{\cal L})\Psi_h
    &=\bigl(e^{C_\pi h}-1\bigr)e^{\lambda(T-t)}
      \bigl(-\lambda\widehat\pi+{\cal L}\widehat\pi\bigr)\\
    &\le-\bigl(e^{C_\pi h}-1\bigr)\widehat\pi
    \le-\bigl(e^{C_\pi h}-1\bigr)\pi
    \le-\bigl(\pi_h-\pi\bigr).
\end{aligned}
\end{equation}

Let \(\varphi\) be the coercive function introduced in the proof of
Theorem~\ref{thm:Q}. Applying the same unbounded-domain comparison argument
to \({\cal Q}_h-{\cal Q}-\Psi_h-\delta\varphi\) and then letting
\(\delta\downarrow0\) gives
${\cal Q}_h-{\cal Q}\leq\Psi_h$.
On the other hand, \eqref{eq:pi-scaling-bound} and the comparison principle
give ${\cal Q}_h\ge{\cal Q}$. Thus
\[
    0\le
    \frac{{\cal Q}(t,S_i^h x,z)-{\cal Q}(t,x,z)}{h}
    \le
    \frac{e^{C_\pi h}-1}{h}e^{\lambda(T-t)}\widehat\pi(x,z)
    \le C_\pi\frac{e^{C_\pi h}-1}{h}e^{\lambda(T-t)}\pi(x,z).
\]
Letting $h\to0^+$ gives
\[
    0\le x_i\partial_{x_i}{\cal Q}(t,x,z)
    \le C_\pi^2 e^{(C_\pi+1)(T-t)}\pi(x,z)
    \le C_\pi^2 e^{(C_\pi+1)T}\pi(x,z).
\]
Combining this inequality with \eqref{eq:Qt-bound-full} and increasing the
constant if necessary proves \eqref{eq:mono-Qi}.

If $h>0$, define
${\cal U}_1:=\{(t,x):{\cal Q}(t,x,z+h)>K^-
\text{ and }{\cal Q}(t,x,z)<K^+\}$.
The variational inequalities imply
$(\partial_t+{\cal L})D_z^h{\cal Q}\geq-D_z^h\pi>0$ in ${\cal U}_1$.
On ${\cal U}_1^c$, the obstacle bounds give $D_z^h{\cal Q}\leq0$.
The same comparison argument as above therefore gives
$D_z^h{\cal Q}\leq0$ in ${\cal D}_T$. If $h<0$, apply the preceding result
at $z+h$ with increment $-h>0$, since
$D_z^h{\cal Q}=D_{z+h}^{-h}{\cal Q}$. Thus \eqref{eq:Dqy<0} follows.

    We also prove the local Lipschitz estimate needed below. Fix a compact
    interval $I\subset(0,\infty)$ and let $z<z'$ belong to $I$. By
    Assumption~\ref{as:free_boundary}\textnormal{(iii)},
    there are constants $C_I>0$ and $m>0$ such that, for a.e.~$x$,
    \[
        0\le \pi(x,z)-\pi(x,z')
        \le C_I(z'-z)\Bigl(1+\sum_{i=1}^n x_i^m\Bigr).
    \]
    Set
    \[
        \Phi(t,x)=C_Ie^{\lambda(T-t)}
        \Bigl(1+\sum_{i=1}^n x_i^m\Bigr).
    \]
    For sufficiently large $\lambda$,
    \[
        (\partial_t+{\cal L})\Phi
        \le -C_I\Bigl(1+\sum_{i=1}^n x_i^m\Bigr).
    \]
    Applying the preceding unbounded-domain comparison to
    ${\cal Q}(\cdot,z)-{\cal Q}(\cdot,z')-(z'-z)\Phi-\varepsilon\varphi$
    and then letting $\varepsilon\downarrow0$, we obtain
    \[
        0\le {\cal Q}(t,x,z)-{\cal Q}(t,x,z')
        \le (z'-z)\Phi(t,x).
    \]
    Since $\Phi$ is bounded on every compact subset of ${\cal D}_T$, the
    asserted local Lipschitz continuity follows.
\end{proof}

\subsection{Existence and basic properties of the free boundaries}

\begin{lem}\label{lem:Zexist} There exist two functions ${\cal Z}_{\cal I} : {\cal D}_T \to [0, \infty)$ and ${\cal Z}_{\cal D} : {\cal D}_T  \to (0, \infty]$ such that
\begin{equation}\label{eq:3region-z}\left\{\begin{aligned}
		{\cal IR} &= \{ (t,x,z)\in{\cal D}_T\times(0,\infty)
		\mid z \leq {\cal Z}_{\cal I}(t,x)\},\\
		{\cal DR} &= \{ (t,x,z)\in{\cal D}_T\times(0,\infty)
		\mid z \geq {\cal Z}_{\cal D}(t,x)\},\\
		{\cal WR} &= \{ (t,x,z)\in{\cal D}_T\times(0,\infty)
		\mid {\cal Z}_{\cal I}(t,x)<z<{\cal Z}_{\cal D}(t,x)\}.
\end{aligned}\right.\end{equation}
Moreover, both ${\cal Z}_{\cal I}$ and ${\cal Z}_{\cal D}$ are
non-increasing in time $t$ and non-decreasing in each state coordinate
$x_i$, with
\begin{equation}\label{ineq:zdoup}
	{\cal Z}_{\cal I}(t, x) \leq
    \pi^{-1}_z(x, \rho K^+) <
    \pi^{-1}_z(x, \rho K^-) \leq
    {\cal Z}_{\cal D}(t, x)
\end{equation}
for each $(t, x) \in {\cal D}_T$, where $\pi^{-1}_z(x, \cdot)$ is an inverse function of $\pi$ in the $z$ variable with fixed $x$.
\end{lem}
\begin{proof}
	Define two functions ${\cal Z}_{\cal I} : {\cal D}_T \to [0, \infty)$ and ${\cal Z}_{\cal D} : {\cal D}_T \to (0, \infty]$ as
	\begin{equation*}
		 {\cal Z}_{\cal I}( t, x) := \sup\{ z \in (0, \infty) \mid {\cal Q}(t, x, z) = K^+\}
    \end{equation*}
         and
\begin{equation*}
         {\cal Z}_{\cal D}( t, x) := \inf\{ z \in (0, \infty) \mid {\cal Q}(t, x, z) = K^-\}.
		\end{equation*}
		Here and below, we use the conventions $\sup\varnothing=0$ and
        $\inf\varnothing=\infty$. By Proposition~\ref{pro:Q-monotonicity-lipschitz}, the map
        $z\mapsto{\cal Q}(t,x,z)$ is continuous and nonincreasing. Hence its
        upper and lower coincidence sets are closed intervals in the
        $z$-variable, and \eqref{eq:3region-z} follows. Moreover, from the monotonicity $\partial_{t} {\cal Q} \leq 0$ and $\partial_{x_{i}} {\cal Q} \geq 0$, we also obtain the monotonicity of ${\cal Z}_{\cal I}$ and ${\cal Z}_{\cal D}$ for $t$ and $x_i$.

    It remains to prove~\eqref{ineq:zdoup}. The standard contact-point
    argument gives
    \[
        \pi(x,z)\geq \rho K^+ \quad\text{on }\{{\cal Q}=K^+\},
        \qquad
        \pi(x,z)\leq \rho K^- \quad\text{on }\{{\cal Q}=K^-\}.
    \]
    Since $z\mapsto\pi(x,z)$ is strictly decreasing and the
    coincidence sets are intervals in the $z$-variable, this gives
    \eqref{ineq:zdoup}.
\end{proof}

\begin{lem}\label{lem:free-boundary-range-limits}
For every $(t,x)\in{\cal D}_T$, the free-boundary functions satisfy
\begin{equation}\label{eq:range-Z}
 0<{\cal Z}_{\cal I}(t,x)\leq\pi^{-1}_z(x,\rho K^+),
 \qquad
 \pi^{-1}_z(x,\rho K^-)\leq{\cal Z}_{\cal D}(t,x)<\infty.
\end{equation}
Moreover,
\begin{equation}\label{eq:terminal-Z}
 \lim_{t\to T^-}{\cal Z}_{\cal I}(t,x)=0,
 \qquad
 \lim_{t\to T^-}{\cal Z}_{\cal D}(t,x)
 =\pi^{-1}_z(x,\rho K^-).
\end{equation}
\end{lem}
\begin{proof}
    We begin with ${\cal Z}_{\cal I}$ and record a barrier comparison. Fix $(\hat t, \hat x) \in {\cal D}_T$ and set
    $h(x) := \sum_{i = 1}^n \big(\log(x_i/\hat x_i)\big)^2$,
    $R := \sqrt{K^+ - K^-}$, $B := \{x \mid h(x) < R^2\}$, and
    $g(t) := (K^+ - K^-)\frac{t - \hat t}{T - \hat t}$. Since $x_i\partial_{x_i}h = 2\log(x_i/\hat x_i)$ and
    $x_i^2\partial_{x_i x_i}h = 2 - 2\log(x_i/\hat x_i)$,
    \begin{equation*}
        {\cal L}h = \sum_{i = 1}^n a_{ii} + 2\sum_{i = 1}^n\left(\mu_i - \tfrac{1}{2}a_{ii}\right)\log(x_i/\hat x_i) - \r h
        \leq M_0 := \sum_{i = 1}^n a_{ii} + \frac{1}{\r}\sum_{i = 1}^n\left(\mu_i - \tfrac{1}{2}a_{ii}\right)^2 .
    \end{equation*}
    Suppose that $z^*>0$ satisfies
    $\pi(\cdot,z^*)\geq g'+\rho K^++M_0$ on $\overline B$, and set
    $\Phi:=K^+-h-g$. Then
    \[
        (\partial_t+{\cal L})\Phi+\pi(\cdot,z^*)
        \geq-g'-\rho K^+-M_0+\pi(\cdot,z^*)\geq0
        \quad\text{in }B.
    \]
    Moreover, $\Phi\leq{\cal Q}(\cdot,\cdot,z^*)$ on the parabolic
    boundary of $[\hat t,T)\times B$. Comparison therefore gives
    ${\cal Q}(\hat t,\hat x,z^*)\geq\Phi(\hat t,\hat x)=K^+$, and hence
    ${\cal Z}_{\cal I}(\hat t,\hat x)\geq z^*$.

    To prove the ${\cal Z}_{\cal I}$ bound in \eqref{eq:range-Z},
    fix $(t,x)\in{\cal D}_T$ and center $h$ at $x$.
    Since $\pi(x, z) \to \infty$ as $z \to 0^+$ uniformly on the compact set $\overline{B}$,
    we may choose $z^* = \underline z > 0$ so that the preceding condition holds. The above comparison gives
    ${\cal Z}_{\cal I}(t, x) \geq \underline z > 0$. Together with the upper bound
    ${\cal Z}_{\cal I}(t,x)\leq\pi^{-1}_z(x,\r K^+)$ from
    \eqref{ineq:zdoup}, this proves the ${\cal Z}_{\cal I}$ bound in
    \eqref{eq:range-Z}.

    For the ${\cal Z}_{\cal I}$ limit in \eqref{eq:terminal-Z}, since
    ${\cal Q}$ is continuous up to $t=T$ with
    ${\cal Q}(T,x,z)=K^-$, for each fixed
    $z > 0$ we have ${\cal Q}(t, x, z) < K^+$, hence ${\cal Z}_{\cal I}(t, x) < z$, once $t$ is close to $T$. As
    $z > 0$ is arbitrary, the limit follows.

    For the ${\cal Z}_{\cal D}$ bound in \eqref{eq:range-Z}, center $h$
    at $x$ and set
    $\varepsilon:=\rho K^-/(2M_0)$, $R^2:=(K^+-K^-)/\varepsilon$, and
    $B:=\{y:h(y)<R^2\}$.
    Choose $\bar z$ so large that
    $\pi(\cdot,\bar z)\leq\rho K^-/2$ on $\overline B$. Then
    $\Psi:=K^-+\varepsilon h$ satisfies
    \[
        (\partial_t+{\cal L})\Psi+\pi(\cdot,\bar z)
        \leq-\rho K^-+\varepsilon M_0+\frac{\rho K^-}{2}=0
    \]
    in $B$, and $\Psi\geq{\cal Q}(\cdot,\cdot,\bar z)$ on the
    parabolic boundary of $[t,T)\times B$. Thus comparison gives
    ${\cal Q}(t,x,\bar z)=K^-$, so
    ${\cal Z}_{\cal D}(t,x)\leq\bar z<\infty$. The remaining lower
    bound for ${\cal Z}_{\cal D}$ in \eqref{eq:range-Z} follows from
    \eqref{ineq:zdoup}.

    For the ${\cal Z}_{\cal D}$ limit in \eqref{eq:terminal-Z},
    monotonicity and the corresponding range bound in
    \eqref{eq:range-Z} give
    $\zeta(x):=\lim_{t\to T^-}{\cal Z}_{\cal D}(t,x)
    \geq\pi^{-1}_z(x,\rho K^-)$.
    Fix $z>\pi^{-1}_z(x,\rho K^-)$. Choose $r,\delta>0$ such that
    \(
        \pi(\cdot,z)\leq\rho K^--2\delta
    \)
    on $\overline{B_r(x)}$, where
    $h_x(y):=\sum_{i=1}^n\bigl(\log(y_i/x_i)\bigr)^2$ and
    $B_r(x):=\{y:h_x(y)<r^2\}$,
    and ${\cal L}h_x\leq M_0$. Set
    $\eta(t_0):=\sup_{[t_0,T]\times\partial B_r(x)}
    ({\cal Q}(s,y,z)-K^-)$.
    By \eqref{eq:Qt-bound-full}, $\eta(t_0)\to0$ as $t_0\to T^-$. Choose
    $t_0$ close enough to $T$ that $\eta(t_0)M_0/r^2\leq\delta$. Then
    $\Psi_{t_0}(y):=K^-+\eta(t_0)h_x(y)/r^2$
    dominates ${\cal Q}(\cdot,\cdot,z)$ on the parabolic boundary and
    satisfies
    \[
        (\partial_s+{\cal L})\Psi_{t_0}+\pi(\cdot,z)
        \leq-\rho K^-+\frac{\eta(t_0)}{r^2}M_0
             +\pi(\cdot,z)<0.
    \]
    Comparison gives
    ${\cal Q}(t_0,x,z)\leq\Psi_{t_0}(x)=K^-$. Since ${\cal Q}\geq K^-$,
    equality holds, and hence ${\cal Z}_{\cal D}(t_0,x)\leq z$. Letting
    $z\downarrow\pi^{-1}_z(x,\rho K^-)$ proves the limit.
\end{proof}

\subsection{Strict ordering of coincidence sets and free-boundary continuity}

\begin{lem}\label{lem:local-strict-comparison}
Fix $p_0=(t_0,x_0)\in{\cal D}_T$. Let ${\cal U}$ be a neighborhood of
$p_0$ in ${\cal D}_T$, and let
$u,v\in W^{1,2}_{p,\mathrm{loc}}({\cal U})$, $p>n+2$, satisfy
\[
 \begin{cases}
 (\partial_t+{\cal L})u=f_u(x)\chi_{\{u>0\}},\\
 (\partial_t+{\cal L})v=f_v(x)\chi_{\{v>0\}},
 \end{cases}
 \qquad\text{a.e. in }{\cal U},
 \qquad 0\leq u\leq v,
\]
where $f_u$ and $f_v$ are locally Lipschitz. If
$f_u(x_0)>f_v(x_0)>0$,
then $p_0$ cannot be a free-boundary point of both $u$ and $v$.
\end{lem}

\begin{proof}
Suppose that $p_0$ is a common free-boundary point and set
$c_u=f_u(x_0)>c_v=f_v(x_0)>0$. With logarithms and exponentials taken
componentwise, let
$\xi_0:=\log x_0$, $A:=\frac12(a_{ij})_{i,j=1}^n$, and
$\beta:=A^{-1/2}(\mu_i-\frac12a_{ii})_{i=1}^n$,
and set
$\tau:=t_0-t$ and $y:=A^{-1/2}(\log x-\xi_0)+\beta\tau$.
For $w=u,v$, define
\[
 W_w(y,\tau)
 =e^{\rho\tau}w
   \bigl(t_0-\tau,e^{\xi_0+A^{1/2}(y-\beta\tau)}\bigr),
 \qquad
 F_w(y,\tau)
 =e^{\rho\tau}f_w
   \bigl(e^{\xi_0+A^{1/2}(y-\beta\tau)}\bigr).
\]
The point $p_0$ corresponds to $(y,\tau)=(0,0)$, and a direct calculation
gives $0\leq W_u\leq W_v$ and
$(\Delta_y-\partial_\tau)W_w
=F_w(y,\tau)\chi_{\{W_w>0\}}$ for $w=u,v$.
Since $F_u(0,0)=c_u>c_v=F_v(0,0)>0$, both sources are bounded above and away
from zero near the origin.
The standard local
quadratic-growth and nondegeneracy estimates for the parabolic obstacle
problem give
\begin{equation}\label{eq:local-quadratic-growth-comparison}
 cr^2\leq\sup_{Q_r^-}w\leq\sup_{Q_r}w\leq Cr^2,
 \qquad 0<r<r_0,\quad w=W_u,W_v.
\end{equation}
Here $Q_r=B_r\times(-r^2,r^2)$ and
$Q_r^-=B_r\times(-r^2,0]$; cf.\
\cite[Proposition~2.4]{LindgrenMonneau2015}
and the nondegeneracy argument in \citet[Section~5.1]{CPS04}.

The quadratic-growth bound and interior
$W^{1,2}_p$ estimates give a common sequence of rescalings converging locally
uniformly to global limits $0\leq U\leq V$ on
$\mathbb R^n\times\mathbb R^-$ satisfying
$(\Delta_y-\partial_\tau)U=c_u\chi_{\{U>0\}}$ and
$(\Delta_y-\partial_\tau)V=c_v\chi_{\{V>0\}}$, where $c_u>c_v>0$.
Here the rescaled sources converge locally uniformly to $c_u,c_v$,
so the variable-source compactness argument in
\citet[Section~3.2]{LindgrenMonneau2015} applies; nondegeneracy excludes zero limits.
The normalized blow-up limits \(U/c_u\) and \(V/c_v\) satisfy
the obstacle equation with source coefficient \(1\). We apply the results
below to these functions.
Applying the Weiss
monotonicity formula \cite[Theorem~3.4]{CPS04} along a further common blow-up
and then \cite[Lemma~6.3]{CPS04}, we may take $U$ and $V$ to be homogeneous,
with
\[
\begin{aligned}
 U&=\frac{c_u}{2}(y\cdot e_u)_+^2
 \quad\text{or}\quad
 U=\lambda_u\tau+P_u(y)
 \quad\text{with}\quad \Delta P_u-\lambda_u=c_u,\\
 V&=\frac{c_v}{2}(y\cdot e_v)_+^2
 \quad\text{or}\quad
 V=\lambda_v\tau+P_v(y)
 \quad\text{with}\quad \Delta P_v-\lambda_v=c_v.
\end{aligned}
\]
Here $e_u,e_v$ are unit vectors, $\lambda_u,\lambda_v\leq0$ are the
coefficients of the time-linear terms, and $P_u,P_v$ are homogeneous
quadratic polynomials.

None of the four possible pairs is compatible with $0\leq U\leq V$ and
$c_u>c_v$. For two half-space solutions, the order forces the same
half-space and $c_u\leq c_v$. For two polynomial solutions,
$P_u\leq P_v$ and $\lambda_u\geq\lambda_v$, again giving $c_u\leq c_v$.
If only $V$ is a half-space solution, then $U$ vanishes on an open
half-space and hence is zero. In the remaining case, write
$U=\frac{c_u}{2}(y\cdot e)_+^2$ and
$V=\lambda_v\tau+y^\top By$. The order gives
$B\geq0$, $\lambda_v\leq0$, and $e^\top Be\geq c_u/2$, so
$c_v=2\operatorname{tr}B-\lambda_v\geq2e^\top Be\geq c_u$,
which is also impossible.
\end{proof}

\begin{lem}\label{lem:strict-nesting-z}
For $0<z_1<z_2$, set
${\cal E}^+_{z_j}:=\{(t,x)\in{\cal D}_T:{\cal Q}(t,x,z_j)=K^+\}$ and
${\cal E}^-_{z_j}:=\{(t,x)\in{\cal D}_T:{\cal Q}(t,x,z_j)=K^-\}$.
Then
\[
 {\cal E}^+_{z_2}\subset
 \operatorname{int}_{{\cal D}_T}{\cal E}^+_{z_1},
 \qquad
 {\cal E}^-_{z_1}\subset
 \operatorname{int}_{{\cal D}_T}{\cal E}^-_{z_2}.
\]
\end{lem}
\begin{proof}
Write ${\cal Q}_j={\cal Q}(\cdot,\cdot,z_j)$. By comparison,
${\cal Q}_1\geq{\cal Q}_2$, and hence
${\cal E}^+_{z_2}\subset{\cal E}^+_{z_1}$ and
${\cal E}^-_{z_1}\subset{\cal E}^-_{z_2}$. We show that these
inclusions are strict at the free boundaries.

\emph{Upper coincidence sets.}
Suppose that
$p_0=(t_0,x_0)\in{\cal E}^+_{z_2}$ does not belong to
$\operatorname{int}_{{\cal D}_T}{\cal E}^+_{z_1}$, and choose
$\bar z\in(z_1,z_2)$. Then $p_0$ is a common free boundary point for the
upper coincidence sets corresponding to $z_1$ and $\bar z$. Indeed,
${\cal E}^+_{z_2}\subset{\cal E}^+_{\bar z}\subset{\cal E}^+_{z_1}$;
hence $p_0$ belongs to the latter two sets but is not an interior point of
either one. In a neighborhood of $p_0$
the lower obstacle is inactive. Set
$u:=K^+-{\cal Q}(\cdot,\cdot,z_1)$ and
$v:=K^+-{\cal Q}(\cdot,\cdot,\bar z)$.
Then $0\leq u\leq v$. Since $p_0\in{\cal E}^+_{z_2}$,
Lemma~\ref{lem:Zexist} and the strict decrease of $\pi$ in $z$ give
$\pi(x_0,z_1)-\rho K^+>\pi(x_0,\bar z)-\rho K^+>0$.
Lemma~\ref{lem:local-strict-comparison} rules out this common
free-boundary point. Hence
${\cal E}^+_{z_2}\subset
\operatorname{int}_{{\cal D}_T}{\cal E}^+_{z_1}$.

\emph{Lower coincidence sets.}
The lower inclusion follows in the same way. Indeed, for
$p_0=(t_0,x_0)\in{\cal E}^-_{z_1}\setminus
\operatorname{int}_{{\cal D}_T}{\cal E}^-_{z_2}$ and
$\bar z\in(z_1,z_2)$, apply Lemma~\ref{lem:local-strict-comparison} to
$u:={\cal Q}(\cdot,\cdot,z_2)-K^-$ and
$v:={\cal Q}(\cdot,\cdot,\bar z)-K^-$, for which
$\rho K^--\pi(x_0,z_2)>\rho K^--\pi(x_0,\bar z)>0$.
This gives a contradiction, and hence
${\cal E}^-_{z_1}\subset
\operatorname{int}_{{\cal D}_T}{\cal E}^-_{z_2}$.
\end{proof}

\begin{thm}
\label{thm:free-boundary-continuity}
    The solution ${\cal Q}$ to \eqref{eq:double-obs-n-dim} is locally Lipschitz continuous on ${\cal D}_T\times(0,\infty)$. Moreover, the free boundary functions ${\cal Z}_{\cal I}$ and ${\cal Z}_{\cal D}$ are continuous on $[0, T) \times (0, \infty)^n$.
\end{thm}
\begin{proof}
Since $\pi$ is locally bounded, \eqref{eq:mono-Qi} gives locally bounded
$\partial_t{\cal Q}$ and $\partial_{x_i}{\cal Q}$, while
Proposition~\ref{pro:Q-monotonicity-lipschitz} gives the locally uniform
Lipschitz estimate in $z$.
Hence ${\cal Q}$ is locally Lipschitz continuous in $(t,x,z)$.
It remains to prove the continuity of the two free boundaries.

For $z>0$, set
${\cal E}^+_z:=\{(t,x)\in{\cal D}_T:{\cal Q}(t,x,z)=K^+\}$ and
${\cal E}^-_z:=\{(t,x)\in{\cal D}_T:{\cal Q}(t,x,z)=K^-\}$.
We use the strict-nesting property from
Lemma~\ref{lem:strict-nesting-z}: whenever $0<z_1<z_2$,
\begin{equation}\label{eq:strict-nesting-z}
 {\cal E}^+_{z_2}\subset
 \operatorname{int}_{{\cal D}_T}{\cal E}^+_{z_1},
 \qquad
 {\cal E}^-_{z_1}\subset
 \operatorname{int}_{{\cal D}_T}{\cal E}^-_{z_2}.
\end{equation}

We first consider ${\cal Z}_{\cal I}$. Let
$p_k=(t_k,x_k)\to p=(t,x)$ in ${\cal D}_T$ and write
$a_k={\cal Z}_{\cal I}(p_k)$ and $a={\cal Z}_{\cal I}(p)$.
By the strict decrease and continuity of $\pi$ in $z$, together
with the locally uniform limits in Assumption~\ref{as:free_boundary}(iv),
the maps $x\mapsto\pi_z^{-1}(x,\rho K^\pm)$ are continuous.
The upper bound for ${\cal Z}_{\cal I}$ in \eqref{eq:range-Z} shows that
$(a_k)$ is locally bounded. If $\bar a=\limsup_k a_k>0$, choose a
subsequence along which
$a_k\to\bar a$. Since ${\cal Q}(p_k,a_k)=K^+$, the continuity of
${\cal Q}$ gives ${\cal Q}(p,\bar a)=K^+$, and hence $\bar a\le a$.
Thus ${\cal Z}_{\cal I}$ is upper semicontinuous.

Suppose that $\underline a:=\liminf_k a_k<a$. Choose
$\underline a<z_1<z_2<a$.
Then $p\in{\cal E}^+_{z_2}$, so~\eqref{eq:strict-nesting-z} gives a
neighborhood $U$ of $p$ in ${\cal D}_T$ such that
$U\subset{\cal E}^+_{z_1}$. Consequently $a_k\ge z_1$ for all sufficiently
large $k$, contradicting the definition of $\underline a$. Hence
$\liminf_k a_k\ge a$, and ${\cal Z}_{\cal I}$ is continuous.

The continuity of ${\cal Z}_{\cal D}$ follows by the same argument, with the
inequalities reversed and using the second inclusion in
\eqref{eq:strict-nesting-z}. The lower bound for ${\cal Z}_{\cal D}$ in
\eqref{eq:range-Z}
excludes a zero finite subsequential limit; the cases involving $+\infty$
are immediate.
\end{proof}

\subsection{Higher regularity of the free boundaries}

The next theorem allows us to differentiate the free boundaries with
respect to time and the state variables.
We note that the verification in
Section~\ref{sec:verification} uses only continuity and strict separation of
the boundaries.
For a given state process $X$, the functions
${\cal Z}_{\cal I}(t,X_t)$ and ${\cal Z}_{\cal D}(t,X_t)$ are the lower and
upper boundaries of the time-dependent inaction interval: capacity is
increased at the lower boundary and reduced at the upper boundary.

The additional hypotheses $(D_1)$, $(D_\infty)$, and
$(C_{\mathrm{eq}})$ are stated in
Appendix~\ref{app:capacity-boundary-smoothness}.
Under $(D_1)$, $\pi$ has strictly positive state derivatives and a strictly
negative capacity derivative, with $C^1$ dependence on capacity as a map
into $C_{\mathrm{loc}}^{2,\alpha_0}$ for some
$\alpha_0\in(0,1)$.
The stronger condition $(D_\infty)$ adds smoothness in all variables,
while $(C_{\mathrm{eq}})$ requires
$-z\partial_z\pi=\lambda(z)x_{i_\ast}\partial_{x_{i_\ast}}\pi$
for some state direction $i_\ast$ and a smooth positive function $\lambda$
with the bounds specified there.

\begin{samepage}
\begin{thm}
\label{thm:capacity-boundary-regularity}
If $(D_1)$ holds, then both free-boundary functions are
jointly continuously differentiable in time and all state variables.  Under
the stronger assumptions $(D_\infty)$ and $(C_{\mathrm{eq}})$, both functions
are jointly smooth.  More precisely,
\[
\begin{aligned}
 (D_1)
 &\quad\Longrightarrow\quad
 {\cal Z}_{\cal I},{\cal Z}_{\cal D}
 \in C^1({\cal D}_T),\\
 (D_\infty)+(C_{\mathrm{eq}})
 &\quad\Longrightarrow\quad
 {\cal Z}_{\cal I},{\cal Z}_{\cal D}
 \in C^\infty({\cal D}_T).
\end{aligned}
\]
\end{thm}
\end{samepage}

The proof of Theorem~\ref{thm:capacity-boundary-regularity} is given in
Appendix~\ref{app:capacity-boundary-smoothness}.

\begin{rem}
The two smooth CES specifications in Remark~\ref{rem:Pi} satisfy
\((D_1)\) and \((D_\infty)\). This follows from the derivative formulas in
Remark~\ref{rem:free-boundary-examples}. For finite \(\theta\ne0\) and
positive weights, however, the ratio
\(-z\partial_z\pi/(x_i\partial_{x_i}\pi)\) depends on \(x\) for each
\(i=1,2\), so \((C_{\mathrm{eq}})\) does not hold.
Theorem~\ref{thm:capacity-boundary-regularity} therefore gives \(C^1\)
regularity for these examples, including the specification in
Section~\ref{sec:Numeric}, but does not yield \(C^\infty\) regularity.
For Cobb--Douglas revenue
\(\Pi=A x_1^{a_1}x_2^{a_2}z^\delta\), with \(A,a_1,a_2>0\) and
\(0<\delta<1\), one has
\(-z\partial_z\pi=(1-\delta)\pi\) and
\(x_i\partial_{x_i}\pi=a_i\pi\). Thus all three conditions hold with
\(\lambda=(1-\delta)/a_i\), and both boundaries are smooth.
The production-switching specification and the perfect-complements revenue
\(P\min\{x_1,x_2\}^{\gamma}z^\delta\) do not satisfy \((D_1)\).
Their marginal revenues are not differentiable on the switching set,
and the derivative in an inactive state coordinate is zero.
Their boundaries still satisfy the conclusions of
Theorem~\ref{thm:free-boundary-continuity}, and the verification in
Section~\ref{sec:verification} applies.
\end{rem}

\section{Verification}
\label{sec:verification}

We verify the connection between the strong solution of the double-obstacle
problem and the reversible-investment problem of
Problem~\ref{pr:main}.  The
argument has four steps.  We first derive, from the free-boundary results of
Section~\ref{sec:fb}, the pathwise regularity needed for reflection.  We then
recover the two switching value functions from the marginal value by means of
Feynman--Kac potentials.  Next, we verify the switching problems, separating
the fixed-level optimal switching rule from the level-crossing rule generated
by the reflected capacity.  Finally, we aggregate the level-wise identities
and verify the singular control.

Throughout the section, Assumptions~\ref{as:global}
and~\ref{as:free_boundary} are in force.
For fixed \(z>0\),
Assumption~\ref{as:free_boundary}\textnormal{(ii)} implies that there are
constants \(C_z,q>0\) such that
\begin{equation}\label{eq:pi-polynomial-growth-verification}
    0<\pi(x,z)\le C_z(1+|x|^q),
    \qquad x\in(0,\infty)^n.
\end{equation}
Indeed, one first replaces every coordinate smaller than one by one, using
the coordinatewise monotonicity of \(\pi\), and then integrates
\(x_i\partial_{x_i}\pi\le C_\pi\pi\) along each remaining coordinate.  Since
the geometric Brownian motion has finite moments of every order on a finite
horizon, all fixed-\(z\) Feynman--Kac quantities used below are integrable.

\subsection{Terminal extension of the free boundaries}

In view of the terminal limits in
Lemma~\ref{lem:free-boundary-range-limits}, extend the free-boundary functions
to \(t=T\) by setting
\begin{equation}\label{eq:terminal-extension-free-boundaries}
    {\cal Z}_{\cal I}(T,x):=0,
    \qquad
    {\cal Z}_{\cal D}(T,x)
    :=\pi^{-1}_z(x,\rho K^-),
    \qquad x\in(0,\infty)^n.
\end{equation}
These are the left-limit extensions of the
preterminal free boundaries, rather than the boundaries of the terminal
coincidence sets. They are used to close the reflection interval
continuously, and no switch or intervention is introduced at \(T\).

\begin{lem}\label{lem:verification-boundary-regularity-separation}
The terminal extensions in
\eqref{eq:terminal-extension-free-boundaries} are continuous on
\([0,T]\times(0,\infty)^n\).  Moreover,
\begin{equation}\label{eq:strict-boundary-gap-pointwise}
    0\le {\cal Z}_{\cal I}(t,x)
    <{\cal Z}_{\cal D}(t,x)<\infty,
    \qquad (t,x)\in[0,T]\times(0,\infty)^n.
\end{equation}
For \(X^x\) solving \eqref{eq:GBM}, define
\begin{equation}\label{eq:random-reflection-boundaries}
    \ell_t:={\cal Z}_{\cal I}(t,X_t^x),
    \qquad
    r_t:={\cal Z}_{\cal D}(t,X_t^x),
    \qquad 0\le t\le T.
\end{equation}
Then, outside a \(\mathbb P\)-null set,
\begin{equation}\label{eq:pathwise-boundary-properties}
    \ell,r\in C([0,T]),
    \qquad
    \inf_{0\le t\le T}(r_t-\ell_t)>0.
\end{equation}
\end{lem}

\begin{proof}
Set
\[
    a(x):=\pi^{-1}_z(x,\rho K^+),
    \qquad
    b(x):=\pi^{-1}_z(x,\rho K^-).
\]
The continuity and strict decrease of \(z\mapsto\pi(x,z)\), together with
the locally uniform endpoint limits in
Assumption~\ref{as:free_boundary}\textnormal{(iv)}, imply that \(a\) and \(b\)
are continuous.  Since \(K^+>K^-\),
\begin{equation}\label{eq:inverse-gap}
    0<a(x)<b(x)<\infty.
\end{equation}
The inverse bounds \eqref{ineq:zdoup} give, for \(t<T\),
\begin{equation}\label{eq:boundary-gap-lower-bound}
    {\cal Z}_{\cal D}(t,x)-{\cal Z}_{\cal I}(t,x)
    \ge b(x)-a(x)>0.
\end{equation}
Together with \eqref{eq:terminal-extension-free-boundaries}, this proves
\eqref{eq:strict-boundary-gap-pointwise}.

By Theorem~\ref{thm:free-boundary-continuity}, both free boundaries are
continuous on \([0,T)\times(0,\infty)^n\). It remains to check continuity at
\(T\).  Let
\(K\Subset(0,\infty)^n\) and choose any sequence \(t_k\uparrow T\).  The
functions \(x\mapsto{\cal Z}_{\cal I}(t_k,x)\) are continuous on \(K\), are
nonincreasing in \(k\), and converge pointwise to zero by the first limit in
\eqref{eq:terminal-Z}.  Dini's theorem therefore yields uniform convergence on
\(K\).  Similarly, the continuous functions
\(x\mapsto{\cal Z}_{\cal D}(t_k,x)\) decrease pointwise to the continuous
function \(b(x)\) by the second limit in \eqref{eq:terminal-Z}; another
application of Dini's theorem gives uniform convergence on \(K\).
Monotonicity in time then upgrades the
sequential conclusion to joint continuity at \(T\).

For almost every \(\omega\), the set
\[
    K_\omega:=\{X_t^x(\omega):0\le t\le T\}
\]
is a compact subset of \((0,\infty)^n\).  The asserted continuity of
\(\ell\) and \(r\) follows by composition.  Finally,
\eqref{eq:boundary-gap-lower-bound} and the terminal definitions give
\[
    \inf_{0\le t\le T}(r_t-\ell_t)
    \ge
    \min_{u\in K_\omega}\bigl(b(u)-a(u)\bigr)>0,
\]
which proves \eqref{eq:pathwise-boundary-properties}.
\end{proof}

\subsection{Recovery of the switching value pair}

Fix \(z>0\) and define
\begin{equation}\label{eq:R0-and-H}
\begin{aligned}
{\cal R}_0(t,x,z)
&:=
\mathbb E_{t,x}\bigg[
\int_t^T e^{-\rho(s-t)}\pi(X_s,z)\,ds
+e^{-\rho(T-t)}K^-
\bigg],\\
H(t,x,z)
&:=
\mathbb E_{t,x}\bigg[
\int_t^T e^{-\rho(s-t)}\pi(X_s,z)\,ds
\bigg].
\end{aligned}
\end{equation}
Thus
\({\cal R}_0=H+e^{-\rho(T-t)}K^-\).

Let
\begin{equation}\label{eq:G-definition-verification}
    G(t,x,z)
    :=(\partial_t+{\cal L}){\cal Q}(t,x,z)+\pi(x,z),
\end{equation}
where \(G\) is understood through an almost-everywhere representative.  The
strong double-obstacle problem implies
\begin{equation}\label{eq:G-contact-decomposition}
\begin{aligned}
G^+(t,x,z)
&:=\max\{G(t,x,z),0\}
=\bigl(\pi(x,z)-\rho K^+\bigr)
  {\bf 1}_{\{{\cal Q}=K^+\}}(t,x,z),\\
G^-(t,x,z)
&:=\max\{-G(t,x,z),0\}
=\bigl(\rho K^--\pi(x,z)\bigr)
  {\bf 1}_{\{{\cal Q}=K^-\}}(t,x,z)
\end{aligned}
\end{equation}
almost everywhere.  In particular,
\begin{equation}\label{eq:G-bounds}
    0\le G^+\le\pi(\cdot,z),
    \qquad
    0\le G^-\le\rho K^-.
\end{equation}
In what follows, we take \(G^+\) and \(G^-\) to be the pointwise versions
defined by the right-hand sides of \eqref{eq:G-contact-decomposition}.
These versions vanish outside their respective contact sets.  Since \(X\)
admits transition densities on \((0,\infty)^n\), their values on any
\(dt\otimes dx\)-null set do not affect the corresponding Feynman--Kac
potentials.

\begin{pro}\label{pro:Qhat-existence}
For \(j=0,1\), define the Feynman--Kac potentials
\begin{equation}\label{eq:Qhat-FK}
\begin{aligned}
\widehat{\cal Q}_0(t,x,z)
&:=\mathbb E_{t,x}\bigg[
  \int_t^T e^{-\rho(s-t)}G^+(s,X_s,z)\,ds
\bigg],\\
\widehat{\cal Q}_1(t,x,z)
&:=\mathbb E_{t,x}\bigg[
  \int_t^T e^{-\rho(s-t)}G^-(s,X_s,z)\,ds
\bigg].
\end{aligned}
\end{equation}
Moreover,
\begin{equation}\label{eq:Qhat-bounds}
    0\le\widehat{\cal Q}_0\le H,
    \qquad
    0\le\widehat{\cal Q}_1
    \le K^-\bigl(1-e^{-\rho(T-t)}\bigr).
\end{equation}
Define
\begin{equation}\label{eq:Qplus-from-potentials}
    {\cal Q}_0^+:=\widehat{\cal Q}_0,
    \qquad
    {\cal Q}_1^+:={\cal R}_0+\widehat{\cal Q}_1.
\end{equation}
Then
\begin{equation}\label{eq:Q-difference-recovered}
    {\cal Q}_1^+-{\cal Q}_0^+={\cal Q},
\end{equation}
and
\begin{equation}\label{eq:Qplus-value-bounds}
    0\le {\cal Q}_0^+\le H,
    \qquad
    0\le {\cal Q}_1^+\le H+K^-.
\end{equation}
Moreover,
\begin{equation}\label{eq:Qplus-terminal}
    {\cal Q}_0^+(T,x,z)=0,
    \qquad
    {\cal Q}_1^+(T,x,z)=K^-.
\end{equation}
\end{pro}

\begin{proof}
On \(\{{\cal Q}=K^\pm\}\), the positive-order weak
derivatives of \({\cal Q}-K^\pm\) entering
\(\partial_t+{\cal L}\) vanish almost everywhere.  Hence
\[
    (\partial_t+{\cal L}){\cal Q}=-\rho K^\pm
    \quad\text{a.e. on }\{{\cal Q}=K^\pm\},
\]
and the two variational inequalities give
\eqref{eq:G-contact-decomposition}. The growth estimate
\eqref{eq:pi-polynomial-growth-verification} and
\eqref{eq:G-bounds} make the potentials in \eqref{eq:Qhat-FK} finite.  The estimates
\eqref{eq:Qhat-bounds} follow immediately from \eqref{eq:G-bounds}.

Since
\[
    (\partial_t+{\cal L}){\cal Q}=G-\pi(\cdot,z),
    \qquad
    {\cal Q}(T,\cdot,z)=K^-,
\]
the localized Feynman--Kac formula, followed by the moment bounds for \(X\),
gives
\begin{equation}\label{eq:Q-FK-with-G}
{\cal Q}(t,x,z)
={\cal R}_0(t,x,z)
-\mathbb E_{t,x}\bigg[
\int_t^T e^{-\rho(s-t)}G(s,X_s,z)\,ds
\bigg].
\end{equation}
Since \(G=G^+-G^-\), equations \eqref{eq:Qhat-FK} and
\eqref{eq:Q-FK-with-G} imply \eqref{eq:Q-difference-recovered}.  Finally,
\eqref{eq:Qplus-value-bounds} follows from
\eqref{eq:R0-and-H}, \eqref{eq:Qhat-bounds}, and
\eqref{eq:Qplus-from-potentials}.
The terminal conditions \eqref{eq:Qplus-terminal}
follow directly from \eqref{eq:R0-and-H}, \eqref{eq:Qhat-FK}, and
\eqref{eq:Qplus-from-potentials}.
\end{proof}

The Feynman--Kac construction gives a convenient verification device that
does not require any global estimate for the gradients of the switching value
functions.  For the process \(X^x\) starting at time zero, set
\begin{equation}\label{eq:FK-martingales}
\begin{aligned}
M_t^0
&:=e^{-\rho t}{\cal Q}_0^+(t,X_t,z)
  +\int_0^t e^{-\rho s}G^+(s,X_s,z)\,ds,\\
M_t^1
&:=e^{-\rho t}{\cal Q}_1^+(t,X_t,z)
  +\int_0^t e^{-\rho s}
    \bigl(\pi(X_s,z)+G^-(s,X_s,z)\bigr)\,ds.
\end{aligned}
\end{equation}
Indeed, set
\[
\begin{aligned}
Z^0
&:=\int_0^T e^{-\rho s}G^+(s,X_s,z)\,ds,\\
Z^1
&:=\int_0^T e^{-\rho s}
\bigl(\pi(X_s,z)+G^-(s,X_s,z)\bigr)\,ds
+e^{-\rho T}K^-.
\end{aligned}
\]
Both random variables are integrable by \eqref{eq:G-bounds} and the
moment bounds for \(X\).  By \eqref{eq:R0-and-H},
\eqref{eq:Qhat-FK}, and the strong Markov property, for every stopping
time \(\tau\le T\),
\[
M_\tau^0=\mathbb E_x[Z^0\mid{\cal F}_\tau],
\qquad
M_\tau^1=\mathbb E_x[Z^1\mid{\cal F}_\tau].
\]
Thus \(M^0\) and \(M^1\) are martingales closed by integrable terminal
random variables and are therefore uniformly integrable.  Hence optional
sampling applies to arbitrary stopping times bounded by \(T\).

\subsection{The fixed-level switching problem}

We use the switching class \(\Gamma\) and the LCRL regime \(\zeta\) of
Section~\ref{sec:switching}.  For optional-sampling arguments it is convenient
to use in addition the post-switch regime: for \(0\le t<T\), set
\begin{equation}\label{eq:post-switch-regime}
N_t^\eta
:=\max\{n\ge0:\tau_n\le t,\ \tau_n<T\},
\qquad
\iota_t:=\kappa_{N_t^\eta},
\qquad
\iota_T:=\zeta_T.
\end{equation}
The processes \(\zeta\) and \(\iota\) differ only at switching times and
therefore give the same Lebesgue-time running payoff.  As in
Section~\ref{sec:switching}, a single switch at time zero is allowed and its
cost is included in the switching-cost sum.

For \(z>0\), \(j\in\{0,1\}\), and \(X_0=x\), define a switching rule
\(\eta^{z,j,*}\) recursively.  Put
\(\tau_0^*=0\), \(\kappa_0^*=j\), and, while \(\tau_n^*<T\), set
\begin{equation}\label{eq:fixed-level-hitting-rule}
\tau_{n+1}^*
:=
\begin{cases}
\inf\{t\in[\tau_n^*,T):{\cal Q}(t,X_t,z)=K^+\},
&\kappa_n^*=0,\\
\inf\{t\in[\tau_n^*,T):{\cal Q}(t,X_t,z)=K^-\},
&\kappa_n^*=1,
\end{cases}
\end{equation}
with \(\inf\varnothing=T\), and put
\(\kappa_{n+1}^*=1-\kappa_n^*\) on
\(\{\tau_{n+1}^*<T\}\).  Thus, if the initial point already belongs to the
relevant contact set, the first switch occurs at time zero.  After this
possible initial switch, strict obstacle separation makes all subsequent
switching times strictly increasing.  No switch is made at \(T\).
Once \(\tau_{n+1}^*=T\), complete the sequence by setting
\(\tau_m^*=T\) and \(\kappa_m^*=\kappa_n^*\) for every \(m\ge n+1\).

The rule is well defined by continuity. Since the
contact sets are closed and \(t\mapsto{\cal Q}(t,X_t,z)\) is adapted and
continuous, each \(\tau_n^*\) is a stopping time and each
\(\kappa_n^*\) is \({\cal F}_{\tau_n^*}\)-measurable. It also has only finitely many
switches pathwise.  Indeed,
\[
    t\longmapsto {\cal Q}(t,X_t,z)
\]
is continuous on \([0,T]\), and two consecutive switches after the possible
time-zero switch require this function to move between the two distinct
levels \(K^+\) and \(K^-\).  Uniform continuity on \([0,T]\) therefore gives
a strictly positive path-dependent lower bound on the time separating two
consecutive switches.

\begin{pro}\label{pro:verification-switching-plus}
For every \(z>0\), \(x\in(0,\infty)^n\), and \(j\in\{0,1\}\),
\begin{equation}\label{eq:switching-value-verified}
    {\cal Q}_j^+(0,x,z)
    ={\bf P}_j^+(x,z)
    :=\sup_{\eta\in\Gamma,\,\kappa_0=j}
      {\cal P}^+(x,z;\eta).
\end{equation}
The fixed-level rule \(\eta^{z,j,*}\) belongs to \(\Gamma\) and attains the
supremum.
\end{pro}

\begin{proof}
Let \(\eta\in\Gamma\) start in regime \(j\), and use the post-switch regime on
each switching interval.   Applying optional sampling to
the uniformly integrable martingales in \eqref{eq:FK-martingales}
successively for the strategy stopped after its first \(m\) switches and then
letting \(m\to\infty\) gives the following dynamic identity for every stopping
time \(\theta\le T\).   {Writing}
\[
C_\theta(\eta)
:=\sum_{n\ge1}e^{-\rho\tau_n}
\left(
-K^+{\bf 1}_{\{\kappa_n=1\}}
+K^-{\bf 1}_{\{\kappa_n=0\}}
\right){\bf 1}_{\{\tau_n\le\theta\}}{\bf 1}_{\{\tau_n<T\}},
\]
then
\begin{equation}\label{eq:optional-sampling-switching-identity}
\begin{aligned}
&\mathbb E_x\bigg[
e^{-\rho\theta}{\cal Q}_{\iota_\theta}^+
    (\theta,X_\theta,z)
+\int_0^\theta e^{-\rho t}\pi(X_t,z)\iota_t\,dt
+C_\theta(\eta)
\bigg]\\
&={\cal Q}_j^+(0,x,z)
-\mathbb E_x\int_0^\theta e^{-\rho t}
\left[
G^+(t,X_t,z){\bf 1}_{\{\iota_t=0\}}
+G^-(t,X_t,z){\bf 1}_{\{\iota_t=1\}}
\right]dt\\
&\quad
+\mathbb E_x\sum_{n\ge1}e^{-\rho\tau_n}
\bigg[
{\cal Q}_{\kappa_n}^+(\tau_n,X_{\tau_n},z)
-{\cal Q}_{\kappa_{n-1}}^+(\tau_n,X_{\tau_n},z)\\
&\hspace{5.0cm}
-K^+{\bf 1}_{\{\kappa_n=1\}}
+K^-{\bf 1}_{\{\kappa_n=0\}}
\bigg]{\bf 1}_{\{\tau_n\le\theta\}}{\bf 1}_{\{\tau_n<T\}} .
\end{aligned}
\end{equation}
{ Taking \(\theta=T\) and using
\[
e^{-\rho T}{\cal Q}_{\iota_T}^+(T,X_T,z)
=e^{-\rho T}K^-\zeta_T,
\]}
dropping the nonnegative terms \(G^+\) in regime zero and \(G^-\) in regime
one, and using \eqref{eq:Q-difference-recovered} and
\(K^-\le {\cal Q}\le K^+\) to obtain
\[
    {\cal Q}_1^+-K^+\le {\cal Q}_0^+,
    \qquad
    {\cal Q}_0^++K^-\le {\cal Q}_1^+,
\]
at switching times, yields
\begin{equation}\label{eq:arbitrary-switching-upper-bound}
    {\cal P}^+(x,z;\eta)\le {\cal Q}_j^+(0,x,z).
\end{equation}
 {The passage \(m\to\infty\) is justified by
\eqref{eq:G-bounds}, the defining integrability condition of \(\Gamma\), and
the fact that the switching brackets are bounded because
\(K^-\le{\cal Q}\le K^+\).  The running and residual terms are integrable by
Assumption~\ref{as:global}, while the continuation-value family is of class
\(D\) by \eqref{eq:Qplus-value-bounds} and the tower property.}  Thus
\({\bf P}_j^+\le {\cal Q}_j^+(0,x,z)\).

We next consider \(\eta^{z,j,*}\), and suppress the stars on its switching
times and regimes for the remainder of the proof.  Up to each finite number of switches,
optional sampling gives equality rather than inequality: before a switch out
of regime zero, \({\cal Q}<K^+\) and hence \(G^+=0\); before a switch out of
regime one, \({\cal Q}>K^-\) and hence \(G^-=0\); and at every switch the
corresponding obstacle equality holds.

It remains to prove the integrability condition in \(\Gamma\), rather than
assume it.  Define the discounted switching cash flow through the first
\(N\) switches by
\[
    C_N:=
    \sum_{1\le k\le N}e^{-\rho\tau_k}
    \Bigl(
       -K^+{\bf 1}_{\{\kappa_k=1\}}
       +K^-{\bf 1}_{\{\kappa_k=0\}}
    \Bigr){\bf 1}_{\{\tau_k<T\}},
\]
and let
\[
    \Lambda_N
    :=\sum_{1\le k\le N}
       e^{-\rho\tau_k}{\bf 1}_{\{\tau_k<T\}}.
\]
Stop the strategy at the earlier of its \(N\)-th switch and \(T\), denoted by
\(\sigma_N\), and retain the switching value of the post-switch regime as the
continuation value.  Repeated optional sampling gives
\begin{equation}\label{eq:truncated-optimal-switching-identity}
\begin{aligned}
{\cal Q}_j^+(0,x,z)
=\mathbb E_x\bigg[&
\int_0^{\sigma_N}e^{-\rho t}\pi(X_t,z)\iota_t\,dt
+C_N\\
&+e^{-\rho\sigma_N}
 {\cal Q}_{\iota_{\sigma_N}}^+(\sigma_N,X_{\sigma_N},z)
\bigg].
\end{aligned}
\end{equation}
The first and third terms on the right are nonnegative.  Moreover,
\eqref{eq:Qplus-value-bounds} and the tower property give, uniformly over all
stopping times \(\sigma\le T\),
\begin{equation}\label{eq:continuation-value-uniform-bound}
\begin{aligned}
\mathbb E_x\bigl[
e^{-\rho\sigma}{\cal Q}_{\iota_\sigma}^+
(\sigma,X_\sigma,z)
\bigr]
&\le H(0,x,z)+K^-,\\
\mathbb E_x\bigg[
\int_0^\sigma e^{-\rho t}\pi(X_t,z)\iota_t\,dt
\bigg]
&\le H(0,x,z).
\end{aligned}
\end{equation}
{ Consequently,
\[
    \mathbb E_x[C_N]
    \ge {\cal Q}_j^+(0,x,z)-2H(0,x,z)-K^-,
\]
and hence \(\inf_N\mathbb E_x[C_N]>-\infty\).}

Because the regimes alternate and \(K^+>K^-\), the elementary pairing of
successive switches gives
\begin{equation}\label{eq:alternating-cost-bound}
    C_N
    \le
    \frac{K^++K^-}{2}
    -\frac{K^+-K^-}{2}\,\Lambda_N.
\end{equation}
Combining this inequality with the uniform lower bound on
\(\mathbb E_x[C_N]\) yields
\[
    \sup_N\mathbb E_x[\Lambda_N]<\infty.
\]
By monotone convergence, \(\eta^{z,j,*}\in\Gamma\).  Since the number of
switches is pathwise finite, the terms in the truncated identity converge
pathwise to their full-horizon counterparts.  The continuation-value family
is uniformly integrable: by \eqref{eq:Qplus-value-bounds}, for every stopping
time \(\sigma\le T\),
\[
e^{-\rho\sigma}{\cal Q}_{\iota_\sigma}^+
    (\sigma,X_\sigma,z)
\le
\mathbb E_x\left[
  \left.\int_0^T e^{-\rho s}\pi(X_s,z)\,ds\right|{\cal F}_\sigma
\right]+K^-,
\]
and the conditional-expectation family on the right is of class \(D\).
The running rewards are dominated by the same integrable full-horizon
integral, while the switching costs are absolutely integrable because the
rule belongs to \(\Gamma\).  Hence letting \(N\to\infty\) in
\eqref{eq:truncated-optimal-switching-identity}, {and using
\eqref{eq:Qplus-terminal} to identify the terminal continuation value,} gives
\[
    {\cal P}^+(x,z;\eta^{z,j,*})={\cal Q}_j^+(0,x,z).
\]
Together with \eqref{eq:arbitrary-switching-upper-bound}, this proves the
claim.
\end{proof}

\subsection{Construction of the reflected capacity process}

Fix \((x,y)\in(0,\infty)^n\times(0,\infty)\), and let \(\ell,r\) be given by
\eqref{eq:random-reflection-boundaries}.  Lemma~\ref{lem:verification-boundary-regularity-separation}
provides continuous, strictly separated, finite boundaries on the whole
horizon.

\begin{lem}\label{lem:skorokhod-capacity-existence}
There exists a unique adapted reflected capacity process and regulator pair
compatible with the LCRL convention of Problem~\ref{pr:main}.  More precisely,
let
\[
    (\overline Y,K^\ell,K^r)
\]
be the ordinary Skorokhod solution for the constant input \(y\) on the random
interval \([\ell_t,r_t]\).  Then
\begin{equation}\label{eq:standard-Skorokhod-solution}
\begin{gathered}
    \overline Y_t=y+K_t^\ell-K_t^r,
    \qquad
    \ell_t\le\overline Y_t\le r_t,\\
    \overline Y_0=(y\vee\ell_0)\wedge r_0,\\
    K_0^\ell=(\ell_0-y)^+,
    \qquad
    K_0^r=(y-r_0)^+.
\end{gathered}
\end{equation}
The processes \(K^\ell,K^r\) are nondecreasing and continuous after their
possible initial values, and their Stieltjes measures are carried by
\(\{\overline Y=\ell\}\) and \(\{\overline Y=r\}\), respectively.

Define the LCRL controls by
\begin{equation}\label{eq:LCRL-control-from-Skorokhod}
\begin{aligned}
&\widehat\xi_0^+=\widehat\xi_0^-=0,
\qquad
\widehat\xi_t^+=K_t^\ell,
\quad
\widehat\xi_t^-=K_t^r,
\qquad 0<t\le T,\\
&\widehat Y_0=y,
\qquad
\widehat Y_t=\overline Y_t,
\qquad 0<t\le T.
\end{aligned}
\end{equation}
Then
\begin{equation}\label{eq:initial-right-jump}
\begin{aligned}
\Delta^+\widehat\xi_0^+
&=(\ell_0-y)^+,
&\qquad
\Delta^+\widehat\xi_0^-
&=(y-r_0)^+,\\
\widehat Y_{0+}
&=(y\vee\ell_0)\wedge r_0.
\end{aligned}
\end{equation}
For \(t>0\), \(\ell_t\le\widehat Y_t\le r_t\), and the two regulator
measures are mutually singular.  In particular,
\((\widehat\xi^+,\widehat\xi^-)\in\Xi(y)\).
\end{lem}

\begin{proof}
By Lemma~\ref{lem:verification-boundary-regularity-separation}, the interval endpoints are
continuous and have a strictly positive pathwise gap.  The extended
Skorokhod problem therefore has a unique solution, and the strict gap makes it
the ordinary finite-variation Skorokhod solution; see
\citet[Theorem~2.6 and Corollary~2.4]{BurdzyKangRamanan2009}.  The Skorokhod
map is causal.  Since \(\ell\) and \(r\) are adapted, the solution and its
regulators are adapted.

The transformation \eqref{eq:LCRL-control-from-Skorokhod} merely moves the
initial values of the ordinary regulators to right jumps immediately after
zero.  Hence the control starts from
\(\widehat\xi_0^+=\widehat\xi_0^-=0\), as required in \(\Xi(y)\), while its
Stieltjes measures charge the initial adjustment in
\eqref{eq:initial-right-jump}.  For the support statement at time zero, the
post-adjustment value \(\widehat Y_{0+}\) is used; at positive times the
ordinary contact conditions apply.

Since \(\ell_t<r_t\), the lower and upper contact sets are disjoint.  Thus
the measures \(d\widehat\xi^+\) and \(d\widehat\xi^-\) have disjoint supports
and are mutually singular.  Positivity holds at time zero because \(y>0\),
and on \((0,T)\) because \(\widehat Y_t\ge\ell_t>0\).  Finally,
    \(r_t\ge\pi^{-1}_z(X_t,\rho K^-)>0\) on \([0,T]\), and continuity gives
\(\underline r:=\min_{0\le t\le T}r_t>0\).  Downward adjustment acts only
when the post-adjustment capacity equals \(r_t\), so it cannot drive the
capacity below \(\min\{y,\underline r\}>0\).  Hence
\(\widehat Y_T>0\) as well.  All the requirements in \(\Xi(y)\) are
satisfied.
\end{proof}

\subsection{The reflected level-crossing family}

For each \(z>0\), define the {level indicator}
\begin{equation}\label{eq:reflected-level-indicator}
\begin{aligned}
\widehat\zeta^{x,y}_0(z)
&:={\bf 1}_{\{y\ge z\}},\qquad
\widehat\zeta^{x,y}_t(z)
:= \limsup_{\substack{s\uparrow t}}
{\bf 1}_{\{\widehat Y_s>z\}},
\qquad 0<t\le T.
\end{aligned}
\end{equation}
The initial right jump in \eqref{eq:initial-right-jump} is recorded as a
possible switch at time zero, and no switch is imposed at \(T\).  Let
{\(\widehat\eta^{x,y}(z)\) be the switching control
canonically associated with \(\widehat\zeta^{x,y}(z)\) for every level at
which the latter is LCRL and of finite variation almost surely, and use the
no-switch rule with initial regime \({\bf 1}_{\{y\ge z\}}\) otherwise.  The
next lemma shows that the exceptional set of levels is Lebesgue-null.}

\begin{lem}\label{lem:level-crossing-consistency}
For Lebesgue-almost every \(z>0\),  {the limsup in
\eqref{eq:reflected-level-indicator} is a limit, and} the process
\(\widehat\zeta^{x,y}(z)\) has finite variation and only finitely many jumps
on \([0,T]\), almost surely.  The resulting family
\(\{\widehat\eta^{x,y}(z)\}_{z>0}\) is consistent in the sense of
\citet[Definition~2.6]{Guo2008}.

In addition, there is a Lebesgue-null set \(N_{x,y}\) such that, for every
\(z\notin N_{x,y}\),
\begin{equation}\label{eq:equality-level-null-time}
    (dt\otimes d\mathbb P)
    \bigl(\{(t,\omega):\widehat Y_t(\omega)=z\}\bigr)=0.
\end{equation}
\end{lem}

\begin{proof}
The reflected capacity is an LCRL finite-variation process.  The level-set
change-of-variables theorem for finite-variation functions, equivalently
\citet[Proposition~2.9]{Guo2008}, shows pathwise that,
for almost every level \(z\), the map
\(s\mapsto{\bf 1}_{\{\widehat Y_s>z\}}\) has finite variation.  Its left
limits therefore exist and define the process in
\eqref{eq:reflected-level-indicator}, with no larger variation.  A
\(\{0,1\}\)-valued finite-variation function has only finitely many jumps.
Tonelli's theorem changes the order of the pathwise and
level-wise exceptional sets, giving the assertion for Lebesgue-almost every
deterministic \(z\), almost surely.  Consistency then follows from the same
level-crossing theorem; see also \citet[Theorem~2.12]{Guo2008}.

For the final assertion, Tonelli's theorem gives
\[
\int_0^\infty
\mathbb E_x\bigg[
\int_0^T {\bf 1}_{\{\widehat Y_t=z\}}\,dt
\bigg]dz
=\mathbb E_x\bigg[
\int_0^T
\int_0^\infty {\bf 1}_{\{\widehat Y_t=z\}}\,dz\,dt
\bigg]=0.
\]
This proves \eqref{eq:equality-level-null-time} after enlarging the preceding
Lebesgue-null exceptional set if necessary.
\end{proof}

\begin{lem}\label{lem:reflected-level-switching-optimality}
For every fixed \((x,y)\), and for Lebesgue-almost every \(z>0\),
\(\widehat\eta^{x,y}(z)\in\Gamma\) and it is optimal for the switching
problem with initial regime
\begin{equation}\label{eq:reflected-initial-regime}
    \widehat\kappa_0(z)={\bf 1}_{\{y\ge z\}}.
\end{equation}
Consequently,
\begin{equation}\label{eq:reflected-level-value-equality}
    {\cal P}^+(x,z;\widehat\eta^{x,y}(z))
    ={\cal Q}_{\widehat\kappa_0(z)}^+(0,x,z).
\end{equation}
\end{lem}

\begin{proof}
Fix a level outside the exceptional set in
Lemma~\ref{lem:level-crossing-consistency}.  When the post-switch regime is
zero, \(\widehat Y_t\le z\) and
\(\widehat Y_t\ge {\cal Z}_{\cal I}(t,X_t)\).  By
\eqref{eq:equality-level-null-time}, the first inequality is strict for
\(dt\otimes d\mathbb P\)-almost every such point.  Therefore
\(z>{\cal Z}_{\cal I}(t,X_t)\), so \({\cal Q}<K^+\) and \(G^+=0\).  Similarly,
in regime one we have, almost everywhere,
\[
    z<\widehat Y_t\le {\cal Z}_{\cal D}(t,X_t),
\]
so \({\cal Q}>K^-\) and \(G^-=0\).

At an upward level crossing, the lower regulator acts.  The crossed level
belongs to the upper contact region
\(\{z\le {\cal Z}_{\cal I}\}\), including when a time-zero right jump crosses
an interval of levels.  Hence
\({\cal Q}_1^+-{\cal Q}_0^+-K^+=0\) at the switch.  At a downward crossing,
the upper regulator acts, the crossed level belongs to
\(\{z\ge {\cal Z}_{\cal D}\}\), and
\({\cal Q}_0^+-{\cal Q}_1^++K^-=0\).

The strategy has finitely many switches pathwise, but membership in
\(\Gamma\) still requires an expectation estimate.  Apply optional sampling
to \eqref{eq:FK-martingales} only through the first \(N\) crossings.  The
vanishing residual and obstacle terms just proved yield the truncated
identity \eqref{eq:truncated-optimal-switching-identity}.  The estimates
\eqref{eq:continuation-value-uniform-bound} and
\eqref{eq:alternating-cost-bound} therefore give
\[
    \sup_N\mathbb E_x\bigg[
      \sum_{1\le k\le N}
      e^{-\rho\widehat\tau_k(z)}
      {\bf 1}_{\{\widehat\tau_k(z)<T\}}
    \bigg]<\infty.
\]
Thus \(\widehat\eta^{x,y}(z)\in\Gamma\) a posteriori.  Letting
\(N\to\infty\) in the truncated identity is justified by the class
\(D\) continuation-value bound and the running-reward and switching-cost
domination used in the proof of
Proposition~\ref{pro:verification-switching-plus}.  It proves
\eqref{eq:reflected-level-value-equality}.  Optimality also follows from
Proposition~\ref{pro:verification-switching-plus}.
\end{proof}

\subsection{Verification for the singular control problem}
\label{subsec:verification-singular-control}

Define
\begin{equation}\label{eq:Q-minus-from-Q-plus-verification}
    {\cal Q}_j^-(t,x,z)
    :={\cal Q}_j^+(t,x,z)-{\cal R}_0(t,x,z),
    \qquad j=0,1.
\end{equation}
In particular,
\begin{equation}\label{eq:Q1-minus-potential}
    {\cal Q}_1^-=\widehat{\cal Q}_1\ge0.
\end{equation}

Applying the full level-wise decomposition to the reflected control requires
expected, rather than merely pathwise, finite variation.  The next lemma
derives this estimate from bounded level truncations.  Its proof uses only the
fixed-level switching verification and does not assume the full Guo
decomposition for the reflected control.

\begin{lem}\label{lem:reflected-candidate-integrability}
Let
\begin{equation}\label{eq:reflected-discounted-variations}
    \widehat A_\rho^\pm
    :=\int_{[0,T)}e^{-\rho t}\,d\widehat\xi_t^\pm .
\end{equation}
Then
\begin{equation}\label{eq:reflected-upward-moment}
    \mathbb E_x[\widehat A_\rho^+]<\infty.
\end{equation}
Consequently, Lemma~\ref{lem:finite-payoff-integrability}
\textnormal{(i)--(ii)} yields
\begin{equation}\label{eq:reflected-full-integrability}
\begin{aligned}
&\mathbb E_x[\widehat A_\rho^-]<\infty,
\qquad
\mathbb E_x[e^{-\rho T}\widehat Y_T]<\infty,\\
&\mathbb E_x\bigg[
    \int_0^T e^{-\rho t}\Pi(X_t,\widehat Y_t)\,dt
\bigg]<\infty.
\end{aligned}
\end{equation}
In particular,
\({\cal J}(x,y;\widehat\xi^+,\widehat\xi^-)\) is real-valued, and each term
of \eqref{eq:main-problem} is separately integrable for the reflected
control.
\end{lem}

\begin{proof}
Put
\[
    \delta_K:=K^+-K^->0,
    \qquad
    F(x,u):=\Pi(x,u)-\rho K^-u,
\]
and define the no-adjustment payoff
\begin{equation}\label{eq:no-adjustment-payoff-verification}
\begin{aligned}
    J^0(x,y)
    &:=
    \mathbb E_x\bigg[
      \int_0^T e^{-\rho t}\Pi(X_t,y)\,dt
      +e^{-\rho T}K^-y
    \bigg]\\
    &=K^-y+
      \mathbb E_x\bigg[
      \int_0^T e^{-\rho t}F(X_t,y)\,dt
      \bigg].
\end{aligned}
\end{equation}
Assumption~\ref{as:global}\textnormal{(ii)} gives
\begin{equation}\label{eq:F-envelope-verification}
    -\rho K^-u\le F(x,u)\le\Psi(x),
    \qquad u>0,
\end{equation}
so both expressions in \eqref{eq:no-adjustment-payoff-verification} are
finite.

For \(m\ge1\), set
\[
    a_m:=\frac{y}{m+1},
    \qquad
    b_m:=(m+1)y,
    \qquad
    I_m:=(a_m,b_m).
\]
Keep the reflected level indicator
\(\widehat\zeta^{x,y}(z)\) for \(z\in I_m\), keep all levels
\(z\le a_m\) permanently active, and keep all levels \(z\ge b_m\)
permanently inactive.  The capacity reconstructed from this consistent
family is
\begin{equation}\label{eq:clipped-reflected-capacity}
    \widehat Y_t^{\,m}
    =a_m+\int_{a_m}^{b_m}\widehat\zeta_t^{x,y}(z)\,dz
    =(\widehat Y_t\vee a_m)\wedge b_m.
\end{equation}
Let
\((\widehat\xi^{m,+},\widehat\xi^{m,-})\) be the Jordan regulator pair of
\(\widehat Y^{\,m}-y\), with the same time-zero convention as in
Lemma~\ref{lem:skorokhod-capacity-existence}, and write
\[
    \widehat A_{\rho,m}^\pm
    :=\int_{[0,T)}e^{-\rho t}\,d\widehat\xi_t^{m,\pm}.
\]
The clipped process is adapted, positive, LCRL, and of finite variation.
Its Jordan measures are mutually singular, and hence
\((\widehat\xi^{m,+},\widehat\xi^{m,-})\in\Xi(y)\).

We first verify integrability on this bounded level range.  For a level
outside the exceptional set in Lemma
\ref{lem:reflected-level-switching-optimality}, let
\begin{align*}
U(z)&:=\sum_{n\ge1}e^{-\rho\widehat\tau_n(z)}
      {\bf 1}_{\{\widehat\kappa_n(z)=1\}}
      {\bf 1}_{\{\widehat\tau_n(z)<T\}},\\
D(z)&:=\sum_{n\ge1}e^{-\rho\widehat\tau_n(z)}
      {\bf 1}_{\{\widehat\kappa_n(z)=0\}}
      {\bf 1}_{\{\widehat\tau_n(z)<T\}},\\
C(z)&:=-K^+U(z)+K^-D(z).
\end{align*}
{On the exceptional Lebesgue-null set, set
\(U(z)=D(z)=C(z)=0\); this does not affect any level integral.}
The switches alternate.  If \(z>y\), the initial regime is zero, and every
downward switch can be paired with an earlier upward switch.  If \(z\le y\),
there can be one additional initial downward switch.  Since discount factors
decrease with time, pathwise
\begin{equation}\label{eq:level-pairing-candidate}
    D(z)\le U(z)+{\bf 1}_{\{z\le y\}},
    \qquad
    \delta_K U(z)
    \le K^-{\bf 1}_{\{z\le y\}}-C(z).
\end{equation}
The level-wise equality
\eqref{eq:reflected-level-value-equality} gives
\begin{equation}\label{eq:candidate-level-cash-identity}
\begin{aligned}
\mathbb E_x[C(z)]
&={\cal Q}_{\widehat\kappa_0(z)}^+(0,x,z)-
\mathbb E_x\bigg[
 \int_0^T e^{-\rho t}\pi(X_t,z)
       \widehat\zeta_t^{x,y}(z)\,dt
 +e^{-\rho T}K^-\widehat\zeta_T^{x,y}(z)
\bigg].
\end{aligned}
\end{equation}
Because \({\cal Q}_0^+,{\cal Q}_1^+\ge0\),
\eqref{eq:level-pairing-candidate} and
\eqref{eq:candidate-level-cash-identity} imply
\begin{equation}\label{eq:level-upcrossing-bound}
\begin{aligned}
\delta_K\mathbb E_x[U(z)]
\le{}&K^-{\bf 1}_{\{z\le y\}}+\mathbb E_x\bigg[
 \int_0^T e^{-\rho t}\pi(X_t,z)
       \widehat\zeta_t^{x,y}(z)\,dt
 +e^{-\rho T}K^-\widehat\zeta_T^{x,y}(z)
\bigg].
\end{aligned}
\end{equation}
Tonelli's theorem is applicable to the nonnegative right-hand side.  Since
\(0\le\widehat\zeta\le1\), monotonicity of \(\Pi\) and
\eqref{eq:revenue-envelope-Psi} give
\begin{equation}\label{eq:bounded-level-running-bound}
\begin{aligned}
&\int_{a_m}^{b_m}\mathbb E_x\bigg[
 \int_0^T e^{-\rho t}\pi(X_t,z)
       \widehat\zeta_t^{x,y}(z)\,dt
 +e^{-\rho T}K^-\widehat\zeta_T^{x,y}(z)
\bigg]dz\\
&\quad\le
\mathbb E_x\bigg[
 \int_0^T e^{-\rho t}
   \bigl(\Pi(X_t,b_m)-\Pi(X_t,a_m)\bigr)\,dt
\bigg]+K^-b_m\\
&\quad\le
\mathbb E_x\bigg[
 \int_0^T e^{-\rho t}
   \bigl(\rho K^-b_m+\Psi(X_t)\bigr)\,dt
\bigg]+K^-b_m<\infty.
\end{aligned}
\end{equation}
It follows from \eqref{eq:level-upcrossing-bound} and then from the first
inequality in \eqref{eq:level-pairing-candidate} that
\begin{equation}\label{eq:bounded-clipped-variation}
    \int_{a_m}^{b_m}\mathbb E_x[U(z)+D(z)]\,dz<\infty.
\end{equation}
The level-set change-of-variables formula therefore yields, without any
signed nonintegrable interchange,
\begin{equation}\label{eq:clipped-coarea}
    \widehat A_{\rho,m}^+
       =\int_{a_m}^{b_m}U(z)\,dz,
    \qquad
    \widehat A_{\rho,m}^-
       =\int_{a_m}^{b_m}D(z)\,dz,
\end{equation}
and both variables on the left are integrable.

Since \(a_m\le\widehat Y^{\,m}\le b_m\), the operating revenue of the
clipped control is bounded by \(\Pi(X,b_m)\) and is integrable by
\eqref{eq:revenue-envelope-Psi}.  Thus the finite-range level decomposition
is now justified.  Using
\eqref{eq:plus-minus-benchmark},
\eqref{eq:reflected-level-value-equality}, and
\eqref{eq:Q1-minus-potential}, it gives
\begin{equation}\label{eq:clipped-candidate-value}
\begin{aligned}
J_m
&:={\cal J}(x,y;\widehat\xi^{m,+},\widehat\xi^{m,-})\\
&=J^0(x,y)
  +\int_y^{b_m}{\cal Q}_0^+(0,x,z)\,dz
  +\int_{a_m}^y{\cal Q}_1^-(0,x,z)\,dz
\ge J^0(x,y).
\end{aligned}
\end{equation}
Here the inequality follows from
\({\cal Q}_0^+\ge0\) and \({\cal Q}_1^-\ge0\).

Apply the discounted integration-by-parts identity
\eqref{eq:discounted-capacity-ibp} to the clipped control.  In the transformed
form \eqref{eq:J-singular-extended}, it reads
\[
J_m
=K^-y+\mathbb E_x\bigg[
 \int_0^T e^{-\rho t}F(X_t,\widehat Y_t^{\,m})\,dt
 -\delta_K\widehat A_{\rho,m}^+
\bigg].
\]
By \eqref{eq:F-envelope-verification} and
\eqref{eq:clipped-candidate-value},
\begin{equation}\label{eq:uniform-clipped-upward-bound}
\begin{aligned}
\delta_K\mathbb E_x[\widehat A_{\rho,m}^+]
&\le K^-y+
\mathbb E_x\bigg[\int_0^T e^{-\rho t}\Psi(X_t)\,dt\bigg]
-J_m\\
&\le K^-y+
\mathbb E_x\bigg[\int_0^T e^{-\rho t}\Psi(X_t)\,dt\bigg]
-J^0(x,y)<\infty.
\end{aligned}
\end{equation}
The upper bound is independent of \(m\).

The intervals \(I_m\) increase to \((0,\infty)\).  By
Lemma~\ref{lem:skorokhod-capacity-existence}, the two Skorokhod regulator
measures are mutually singular and therefore form the Jordan pair of
\(\widehat Y-y\).  The coarea identities for this reflected
finite-variation process and
\eqref{eq:clipped-coarea} give
\[
    \widehat A_{\rho,m}^+
    =\int_{I_m}U(z)\,dz
    \uparrow
    \int_0^\infty U(z)\,dz
    =\widehat A_\rho^+
    \qquad\text{a.s.}
\]
Monotone convergence and \eqref{eq:uniform-clipped-upward-bound} prove
\eqref{eq:reflected-upward-moment}.
By Lemma~\ref{lem:skorokhod-capacity-existence},
\((\widehat\xi^+,\widehat\xi^-)\in\Xi(y)\).  Therefore,
\eqref{eq:reflected-upward-moment} and
Lemma~\ref{lem:finite-payoff-integrability}\textnormal{(i)--(ii)} imply
\eqref{eq:reflected-full-integrability}, show that
\({\cal J}(x,y;\widehat\xi^+,\widehat\xi^-)\) is real-valued, and give the
separate integrability of every term in \eqref{eq:main-problem}.
\end{proof}

\begin{thm}\label{thm:verification-singular-control}
The reflected control
\((\widehat\xi^+,\widehat\xi^-)\) constructed in
Lemma~\ref{lem:skorokhod-capacity-existence} is optimal for
Problem~\ref{pr:main}.  Moreover,
\begin{equation}\label{eq:main-value-representation-verified}
\begin{aligned}
\mathfrak U(x,y)
&=\mathbb E_x\bigg[
\int_0^T e^{-\rho t}\Pi(X_t,y)\,dt
+e^{-\rho T}K^-y
\bigg]
+\int_y^\infty {\cal Q}_0^+(0,x,z)\,dz
+\int_0^y {\cal Q}_1^-(0,x,z)\,dz,
\end{aligned}
\end{equation}
and all terms in this identity are finite.
\end{thm}

\begin{proof}
{By \eqref{eq:Pplus-Pminus-values},
Proposition~\ref{pro:verification-switching-plus}, and
\eqref{eq:Q-minus-from-Q-plus-verification},} we have
\begin{equation}\label{eq:Pminus-value-verified}
    {\bf P}_j^-(x,z)
    ={\bf P}_j^+(x,z)-{\cal R}_0(0,x,z)
    ={\cal Q}_j^-(0,x,z),
    \qquad j=0,1.
\end{equation}

Lemma~\ref{lem:reflected-candidate-integrability} makes every term in the
full level-wise decomposition of the reflected control integrable.  Thus
\eqref{eq:Guo-decomposition} applies on \((0,\infty)\).  Lemma
\ref{lem:reflected-level-switching-optimality} and
\eqref{eq:plus-minus-benchmark} yield
\begin{equation}\label{eq:reflected-full-level-equalities}
\begin{aligned}
{\cal P}^+(x,z;\widehat\eta^{x,y}(z))
&={\cal Q}_0^+(0,x,z)
&&\text{for a.e. }z>y,\\
{\cal P}^-(x,z;\widehat\eta^{x,y}(z))
&={\cal Q}_1^-(0,x,z)
&&\text{for a.e. }0<z<y.
\end{aligned}
\end{equation}
Consequently,
\begin{equation}\label{eq:reflected-attains-aggregate-value}
\begin{aligned}
{\cal J}(x,y;\widehat\xi^+,\widehat\xi^-)
&=\mathbb E_x\bigg[
\int_0^T e^{-\rho t}\Pi(X_t,y)\,dt
+e^{-\rho T}K^-y
\bigg]\\
&\quad
+\int_y^\infty {\cal Q}_0^+(0,x,z)\,dz
+\int_0^y {\cal Q}_1^-(0,x,z)\,dz.
\end{aligned}
\end{equation}
The left-hand side is finite by Lemma
\ref{lem:reflected-candidate-integrability}; the baseline is finite by
\eqref{eq:baseline-revenue-integrability}; and the two integrands on the right are
nonnegative.  Hence both level integrals in
\eqref{eq:reflected-attains-aggregate-value} are finite.

Let \((\xi^+,\xi^-)\in\Xi(y)\) now be arbitrary.  If
\({\cal J}(x,y;\xi^+,\xi^-)=-\infty\), the desired upper bound is immediate.
Otherwise, Lemma~\ref{lem:finite-payoff-integrability} supplies precisely the
integrability needed for the level-crossing decomposition and for all
Tonelli--Fubini interchanges.  Therefore
\eqref{eq:singular-switching-upper-bound},
Proposition~\ref{pro:verification-switching-plus}, and
\eqref{eq:Pminus-value-verified} give
\begin{equation}\label{eq:singular-upper-bound-verification}
\begin{aligned}
{\cal J}(x,y;\xi^+,\xi^-)
&\le\mathbb E_x\bigg[
\int_0^T e^{-\rho t}\Pi(X_t,y)\,dt
+e^{-\rho T}K^-y
\bigg]\\
&\quad
+\int_y^\infty {\cal Q}_0^+(0,x,z)\,dz
+\int_0^y {\cal Q}_1^-(0,x,z)\,dz.
\end{aligned}
\end{equation}
Taking the supremum over \(\Xi(y)\) and comparing
\eqref{eq:singular-upper-bound-verification} with
\eqref{eq:reflected-attains-aggregate-value} proves optimality and
\eqref{eq:main-value-representation-verified}.
\end{proof}

\section{Numerical computations}\label{sec:Numeric}

We solve the double-obstacle problem numerically for a two-factor
($n=2$) specification with CES-in-non-capital-factors profit
$\Pi(x_1,x_2,y)=(x_1^\theta+x_2^\theta)^{\gamma/\theta}\,y^\delta$,
where $\theta=-1$, $\gamma=0.8$, and $\delta=0.4$.  Both factors
follow geometric Brownian motion with drift $\mu_i=0.02$ and
baseline variance $\sigma_i^2=0.04$; the factors are negatively
correlated ($\rho_{12}=-0.5$).
The horizon is $T=5$, the cost band is
$K^+=1.00$, $K^-=0.85$, and the discount rate is $\rho=0.12$.
The parameters are not calibrated to a particular industry, but each is
chosen to be economically plausible.  The elasticity of substitution
$1/(1-\theta)=1/2$ makes the two demand factors complements, and $\delta<1$
makes marginal revenue decreasing in capacity.  A drift of $2\%$ and a
volatility of $\sigma_i=0.2$ are of the order of long-run demand growth and
firm-level demand uncertainty, and the negative correlation describes two
partly offsetting sources of demand, such as distinct markets or customer
segments.  The salvage value $K^-=0.85K^+$ corresponds to a $15\%$ loss on
resale, a moderate degree of irreversibility; the discount rate $\rho=0.12$
lies within the range of firm-level costs of capital; and the horizon $T=5$
can be read as the planning horizon or the remaining life of the plant.  The
magnitudes reported below are nonetheless specific to this parameterisation.
Each solve discretises the state space on a grid of size
$400\times400\times N_y\times60$ in $(x_1,x_2,y,t)$ (with $N_y$ up
to $800$) on the domain
$x_i\in(0,100]$, $y\in[0.01,200]$, and applies an implicit
finite-difference scheme; Appendix~\ref{app:numerics} lists further details.
The figures are then cropped at the upper boundaries to reduce interference of boundary conditions.
In all figures, yellow denotes the investment region (IR),
teal the waiting region (WR), and purple the disinvestment region (DR).

The near-expiry behaviour in the figures below is governed by two
quantities, which we make explicit here.  Retaining a marginal unit of
capacity until~$T$ yields its discounted marginal revenue plus the
liquidation value $e^{-\rho(T-t)}K^-$, whereas selling it immediately
yields~$K^-$; retention is worthwhile when marginal revenue exceeds the
\emph{user cost of capital} $\rho K^-$: the flow forgone by leaving the
salvage value invested in the asset rather than realising it.
Purchasing a unit is governed by a second quantity: it costs~$K^+$,
earns approximately $\pi(x,z)(T-t)$ over the remaining horizon, and is
liquidated at~$K^-$, so it is worth buying only if
\begin{equation}\label{eq:inv_heuristic}
  \pi(x,z) \;\gtrsim\; \rho K^- + \frac{K^+ - K^-}{T-t},
\end{equation}
that is, only if marginal revenue covers the user cost \emph{plus} the
round-trip friction $K^+-K^-$ amortised over the time remaining.  The
two terminal limits in \eqref{eq:terminal-Z}
are the $t\to T$ statements of these two conditions.

\providecommand{\figwidth}{6.25cm}
\providecommand{\figheight}{6cm}

\begin{figure}[htbp]
\centering

\begin{tabular}{@{}c@{}c@{}c@{}}
\begin{tikzpicture}
  \begin{axis}[
    width=\figwidth,
    height=\figheight,
    xlabel={$y$},
    ylabel={$x_1 = x_2$},
    axis on top,
    enlargelimits=false,
    xmin=0.01, xmax=200.0,
    ymin=0.0, ymax=44.86215538847118,
    title = {$T-t = 5$}
  ]
    \addplot graphics [
      xmin=0.01, xmax=200.0,
      ymin=0.0, ymax=44.86215538847118
    ] {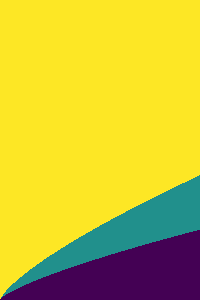};
  \end{axis}
\end{tikzpicture}&

\begin{tikzpicture}
  \begin{axis}[
    width=\figwidth,
    height=\figheight,
    xlabel={$y$},
    ytick=\empty,
    yticklabels={},
    axis on top,
    enlargelimits=false,
    xmin=0.01, xmax=200.0,
    ymin=0.0, ymax=44.86215538847118,
    title={$T-t \approx 0.85$}
  ]
    \addplot graphics [
      xmin=0.01, xmax=200.0,
      ymin=0.0, ymax=44.86215538847118
    ] {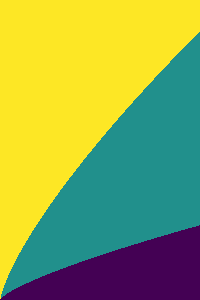};
  \end{axis}
\end{tikzpicture}
&
\begin{tikzpicture}
  \begin{axis}[
    width=\figwidth,
    height=\figheight,
    xlabel={$y$},
    ytick=\empty,
    yticklabels={},
    axis on top,
    enlargelimits=false,
    xmin=0.01, xmax=200.0,
    ymin=0.0, ymax=44.86215538847118,
    title={$T-t \approx 0.25$}
  ]
    \addplot graphics [
      xmin=0.01, xmax=200.0,
      ymin=0.0, ymax=44.86215538847118
    ] {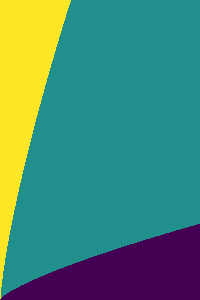};
  \end{axis}
\end{tikzpicture}
\\
\end{tabular}
\caption{Policy regions in the $(y,\, x_1{=}x_2)$ plane along the
symmetric diagonal for decreasing time-to-maturity $T - t$.  Each
panel shows the investment, waiting, and disinvestment regions as a
function of capacity~$y$ and the common factor level $x_1 = x_2$.
Far from expiry, the free boundaries $\bdrydisinv(x)$ and
$\bdryinv(x)$ are well-separated.  As $t \to T$, the
investment region shrinks and is absorbed by the waiting region:
the terminal liquidation at~$K^-$ makes upward adjustment expensive
near expiry.  The disinvestment boundary remains largely stationary, converging to the myopic threshold at which marginal revenue equals the user cost $\rho K^-$.}
\label{fig:boundary_x_v_y}
\end{figure}

Figure~\ref{fig:boundary_x_v_y} displays the policy regions in the
$(y,\,x_1{=}x_2)$ plane along the symmetric diagonal of the
factor space, at three values of the time-to-maturity.  Far from
expiry ($T-t=5$), the free boundaries $\bdrydisinv(x)$ and
$\bdryinv(x)$ are clearly separated and increase with the
factor level: higher demand calls for higher capacity, and the
boundaries demarcating optimal adjustment scale accordingly.  As $t\to T$,
the disinvestment boundary remains approximately stationary, converging
to the myopic threshold $\pi_z^{-1}(x,\rho K^-)$ given by the second limit
in \eqref{eq:terminal-Z}: near expiry the firm disinvests exactly when
marginal revenue falls below the user cost.  The investment region, by
contrast, shrinks markedly and the waiting region expands to absorb it.
The asymmetry is the amortised friction in \eqref{eq:inv_heuristic}.

\begin{figure}
\centering
\providecommand{\figwidth}{8cm}
\providecommand{\figheight}{6cm}
\begin{tikzpicture}
  \begin{axis}[
    width=\figwidth,
    height=\figheight,
    xlabel={$T-t$},
    ylabel={$x_1 = x_2$},
    axis on top,
    enlargelimits=false,
    xmin=0.0, xmax=5.0,
    ymin=0.0, ymax=83.27759197324414,
    x dir=reverse,
  ]
    \addplot graphics [
      xmin=0.0, xmax=5.0,
      ymin=0.0, ymax=83.27759197324414
    ] {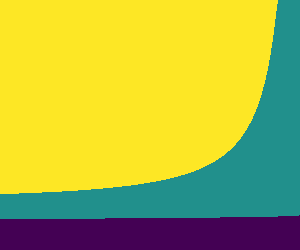};
  \end{axis}
\end{tikzpicture}
\caption{Policy regions along the symmetric diagonal $x_1 = x_2$ as a
function of time-to-maturity $T - t$, at fixed capacity $y \approx
100$.  The horizontal axis runs from the full horizon ($T - t = 5$)
on the left to expiry ($T - t = 0$) on the right.  As $t\to T$, the investment region shrinks markedly and the waiting region expands to absorb it: buying a unit near expiry must recover the round-trip friction $K^+-K^-$ over a vanishing remaining horizon, whereas retaining one is governed by the user cost $\rho K^-$ alone.  The disinvestment region is largely unchanged.}
\label{fig:boundary_x_v_t}
\end{figure}

\providecommand{\figwidth}{5.2cm}
\providecommand{\figheight}{5cm}

\begin{figure}[htbp]
\centering
\begin{tabular}{@{}c@{}c@{}c@{}}
\begin{tikzpicture}
  \begin{axis}[
    width=\figwidth, height=\figheight,
    xlabel={$x_1$}, ylabel={$x_2$},
    title={$T - t = 1$},
    axis on top, enlargelimits=false,
    xmin=0, xmax=100, ymin=0, ymax=100,
  ]
    \addplot graphics [xmin=0, xmax=100, ymin=0, ymax=100]
      {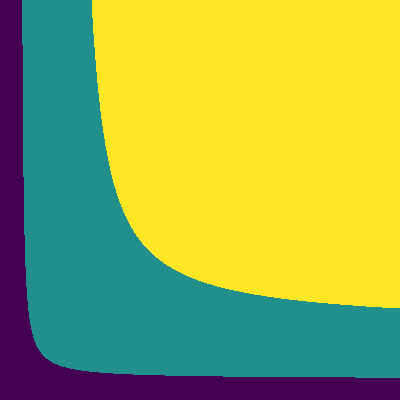};
  \end{axis}
\end{tikzpicture}
&
\begin{tikzpicture}
  \begin{axis}[
    width=\figwidth, height=\figheight,
    xlabel={$x_1$}, %
    ytick=\empty,
    yticklabels={},
    title={$T - t = 5$},
    axis on top, enlargelimits=false,
    xmin=0, xmax=100, ymin=0, ymax=100,
  ]
    \addplot graphics [xmin=0, xmax=100, ymin=0, ymax=100]
      {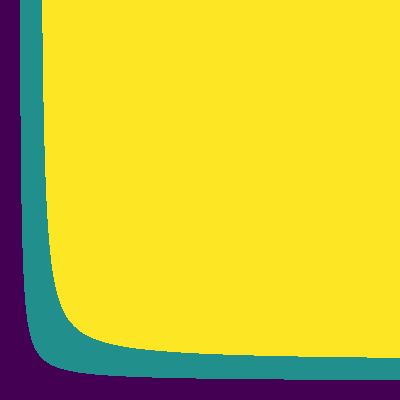};
  \end{axis}
\end{tikzpicture}
&
\begin{tikzpicture}
  \begin{axis}[
    width=\figwidth, height=\figheight,
    xlabel={$x_1$}, %
    ytick=\empty,
    yticklabels={},
    title={$T - t = 10$},
    axis on top, enlargelimits=false,
    xmin=0, xmax=100, ymin=0, ymax=100,
  ]
    \addplot graphics [xmin=0, xmax=100, ymin=0, ymax=100]
      {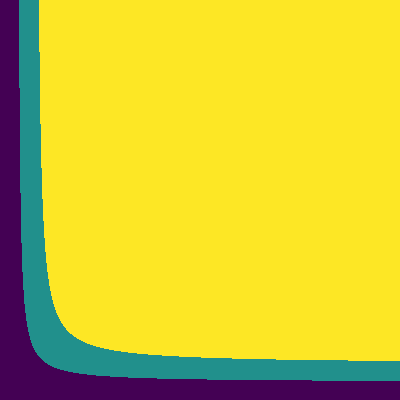};
  \end{axis}
\end{tikzpicture}
\\
\end{tabular}
\caption{Horizon comparative statics.  Each panel shows the policy
regions in the $(x_1, x_2)$ plane at fixed capacity $y = 100$, solved with $T\in\{1,5,10\}$, i.e., evaluated at $t=0$.
At $T-t=1$ the amortised friction $(K^+-K^-)/(T-t)$ is large, so investment
requires a much higher marginal revenue and the waiting region is
correspondingly thicker.  As $T$ increases, the terminal
condition becomes less relevant and the WR narrows toward its
long-run stationary shape.}
\label{fig:experiment_C}
\end{figure}

Figure~\ref{fig:boundary_x_v_t} traces the two boundaries along the diagonal
continuously in time-to-maturity, at fixed capacity $y\approx100$: the
investment region narrows towards expiry while the disinvestment boundary
stays nearly flat.
Figure~\ref{fig:experiment_C} reports a horizon comparative-statics
exercise: the policy regions in the $(x_1,x_2)$ plane at $y=100$ are compared across $T-t\in\{1,\,5,\,10\}$.
Shorter horizons shrink the investment region noticeably.
The comparison across panels is quantitative: by \eqref{eq:inv_heuristic}, the
amortised friction $(K^+-K^-)/(T-t)$ dominates the user cost $\rho K^-$ once
$T-t < (K^+-K^-)/(\rho K^-) \approx 1.5$.  The right-hand side of
\eqref{eq:inv_heuristic} accordingly falls from $0.252$ at $T-t=1$ to $0.132$ at
$T-t=5$ and $0.117$ at $T-t=10$, which is why the first panel differs sharply
from the other two while the last two resemble each other.
The channel runs entirely through the friction: were $K^+=K^-$ the amortised
term would vanish and the horizon would move neither boundary, so the
near-expiry collapse of the investment region is driven by costly
reversibility rather than by the finite horizon as such.

\providecommand{\figwidth}{6.5cm}
\providecommand{\figheight}{6cm}

\begin{figure}[htbp]
\centering
\begin{tabular}{@{}cc@{}}
\begin{tikzpicture}
  \begin{axis}[
    width=\figwidth, height=\figheight,
    xlabel={$x_1$}, ylabel={$x_2$},
    title={$\sigma_2^2 = 0.01$},
    axis on top, enlargelimits=false,
    xmin=0, xmax=100, ymin=0, ymax=100,
  ]
    \addplot graphics [xmin=0, xmax=100, ymin=0, ymax=100]
      {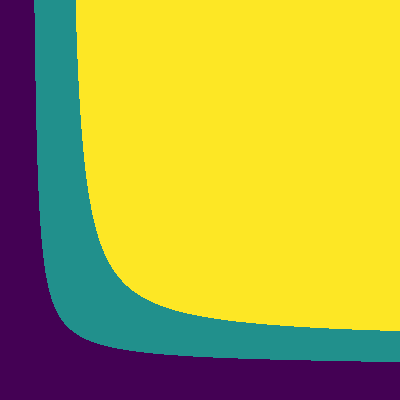};
  \end{axis}
\end{tikzpicture}
&
\begin{tikzpicture}
  \begin{axis}[
    width=\figwidth, height=\figheight,
    xlabel={$x_1$}, ylabel={$x_2$},
    title={$\sigma_2^2 = 0.16$},
    axis on top, enlargelimits=false,
    xmin=0, xmax=100, ymin=0, ymax=100,
  ]
    \addplot graphics [xmin=0, xmax=100, ymin=0, ymax=100]
      {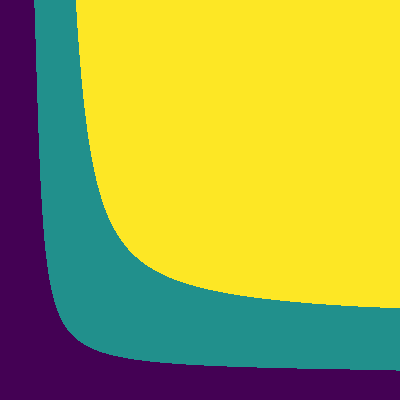};
  \end{axis}
\end{tikzpicture}
\\
\end{tabular}
\caption{Waiting region under asymmetric factor volatility.  Each panel
shows the investment region~(IR), waiting region~(WR), and
disinvestment region~(DR) in the $(x_1, x_2)$ plane at fixed capacity
$y = 100$, with $\sigma_1^2 = 0.04$, $\rho_{12} = -0.5$, and
$\sigma_2^2$ increasing across panels.  Raising $\sigma_2$ away from the least-uncertain configuration $\sigma_2=-\rho_{12}\sigma_1$ increases the volatility of the CES composite and thickens the waiting region, breaking the symmetry of the free boundary about the diagonal; the effect is concentrated where $x_2$ is the scarce factor.}
\label{fig:experiment_A}
\end{figure}

Figure~\ref{fig:experiment_A} displays the optimal policy regions in
the $(x_1,x_2)$ plane at fixed capacity $y=100$ for two values of the
second factor's volatility $\sigma_2\in\{0.1,\,0.4\}$, holding
$\sigma_1=0.2$ and the correlation $\rho_{12}=-0.5$ fixed; the
symmetric baseline $\sigma_2=\sigma_1=0.2$ is the configuration of
Figures~\ref{fig:boundary_x_v_y}--\ref{fig:experiment_C}, in which the
free boundaries are symmetric about the diagonal $x_1=x_2$.

Because the factors are negatively correlated, $\sigma_2$ does not act
monotonically on the uncertainty faced by the firm.  Profit depends on
the factors only through the CES composite
$A(x)=(x_1^\theta+x_2^\theta)^{\gamma/\theta}$, so the relevant
uncertainty is the volatility $\sigma_A$ of that composite rather than
the volatilities of the individual factors.  Let
$u_i:=\partial\ln A/\partial\ln x_i$ denote the elasticity of the
composite with respect to factor~$i$, and let $s:=x_1/(x_1+x_2)$ be the
share of the first factor.  For $\theta=-1$ one has $u_1=\gamma(1-s)$
and $u_2=\gamma s$, so that $u_2$ is large precisely when $x_1$
dominates the composite, and
\begin{equation}\label{eq:composite_vol}
  \sigma_A^2 \;=\; u_1^2\sigma_1^2
    + 2u_1u_2\rho_{12}\sigma_1\sigma_2 + u_2^2\sigma_2^2 .
\end{equation}
A factor thus transmits its volatility in proportion to how binding it
is: $u_2$ is largest where $x_2$ is scarce.  On the diagonal $s=1/2$,
so $\sigma_A^2=\gamma^2(\sigma_1^2+2\rho_{12}\sigma_1\sigma_2+\sigma_2^2)/4$,
which is minimised at $\sigma_2=-\rho_{12}\sigma_1=0.1$ along the diagonal.
The left panel is therefore the least-uncertain configuration available there
at these values of $\sigma_1$ and $\rho_{12}$, not merely a low-volatility one.
Raising $\sigma_2$ also alters the composite's drift, since $a_{22}$ and
$a_{12}$ enter the second-order terms of $\mathcal{L}A$.
Off the diagonal, the transmission is strongly asymmetric.
By \eqref{eq:composite_vol}, raising $\sigma_2$ from $0.1$ to $0.4$
changes $\sigma_A^2$ by $\gamma^2 s\,(0.21\,s-0.06)$, which is positive
where $x_2$ is the scarce factor and negative for $x_2>2.5\,x_1$.  The
reversal is a correlation effect: where $x_1$ binds, $x_2$ enters mainly
through the negatively correlated cross term, so a more volatile $x_2$
dampens the composite rather than amplifying it.  No single-factor model
admits this, since one variance moves the inaction band the same way
everywhere.  The two branches are very unequal in size: along the arm on
which $x_2$ binds the increase exceeds an order of magnitude and the band
thickens visibly, whereas over the range plotted the negative branch
removes at most $17\%$ of the local variance and leaves the other arm
unchanged.  Since $\sigma_A^2$ is instantaneous, evaluated at the current
state, it locates where a change in $\sigma_2$ is transmitted to the firm and
not how far the boundaries move.

\section{Implications}\label{sec:implications}

We quantify the interaction between the two demand factors through the
discrete cross-partial of each free boundary,
\[
  \Delta^2 f \;=\; f(x_0+\varepsilon e) - f(x_0+\varepsilon e_1)
              - f(x_0+\varepsilon e_2) + f(x_0),
  \qquad e = e_1 + e_2,
\]
evaluated at the baseline factor level $x_0=(10,\,10)$ and $T-t=5$.
Normalised by $\varepsilon^2$, this approximates
$\partial^2 f/\partial x_1\partial x_2$ and measures the departure from
additivity: a positive value means that a simultaneous improvement in
both factors shifts the boundary further than the sum of two
equal-sized individual improvements.  Table~\ref{tab:boundary_cross_conv}
reports $\Delta^2/\varepsilon^2$ for the disinvestment threshold
$\bdrydisinv$, the investment threshold $\bdryinv$, and the
waiting-region width $w = \bdrydisinv - \bdryinv$, for the two
one-sided ($\pm\varepsilon$) estimates at $\varepsilon\in\{2.5,\,1\}$.
The estimates tighten as $\varepsilon$ shrinks: at $\varepsilon=1$ the two
endpoints lie within $2.9$--$3.2\%$ of their average, against $8.1$--$8.7\%$ at
$\varepsilon=2.5$.  We base the discussion on the $\pm\varepsilon$ average at $\varepsilon=1$.
These one-sided estimates vary $\varepsilon$ at a fixed mesh, so they bound the
sensitivity to the perturbation size and not the discretisation error.

\begin{table}[htbp]
\centering
\caption{Sensitivity of the normalised cross-partial $\Delta^2/\varepsilon^2$ to the
perturbation size $\varepsilon$, for the disinvestment threshold
($\bdrydisinv$), investment threshold ($\bdryinv$), and waiting-region width
($w$) at $x_0=(10,10)$, with $N_y=800$.  The reference point is chosen to
capture both boundaries within the computation region, cf.\
Figure~\ref{fig:experiment_C}.  The headline estimate is the $\pm\varepsilon$
average at $\varepsilon=1$.}
\label{tab:boundary_cross_conv}
\begin{tabular}{l c c c c c c }
\toprule
  & \multicolumn{3}{c}{\(\varepsilon = 2.5\)} & \multicolumn{3}{c}{\(\varepsilon = 1.0\)} \\
\cmidrule(lr){2-4}
\cmidrule(lr){5-7}
  & \(\bdrydisinv\) & \(\bdryinv\) & \(w\) & \(\bdrydisinv\) & \(\bdryinv\) & \(w\) \\
\midrule
\(+\varepsilon\) shift & 0.6058 & 0.3263 & 0.2795 & 0.6437 & 0.3535 & 0.2902 \\
\(-\varepsilon\) shift & 0.7175 & 0.3886 & 0.3289 & 0.6839 & 0.3767 & 0.3072 \\
\bottomrule
\end{tabular}
\end{table}

The sign of these interactions is largely structural.  At maturity,
$\pi(x,z)=\delta A(x)z^{\delta-1}$ with $A(x)=(x_1^\theta+x_2^\theta)^{\gamma/\theta}$,
so the second limit in \eqref{eq:terminal-Z} gives
\begin{equation}\label{eq:ZD_terminal_closed}
  \bdrydisinv(T,x)=\Big(\tfrac{\delta}{\rho K^-}\Big)^{1/(1-\delta)}
    \big(x_1^\theta+x_2^\theta\big)^{q},
  \qquad q=\frac{\gamma}{\theta(1-\delta)} .
\end{equation}
Its mixed derivative has the sign of $q(q-1)$, which is positive if and only if
$\theta<0$ or $0<\theta<\gamma/(1-\delta)$; this includes the calibration
$\theta=-1$, for which $q=-4/3$.  In this range a positive cross-partial is a
feature of the CES power form rather than a signature of complementarity,
and we read the estimates below as measuring the size and asymmetry of the
interaction, not its sign.

At $\varepsilon=1$, both free boundaries display \emph{positive
super-additivity} in the demand factors at the evaluation point:
$\Delta^2 \bdrydisinv/\varepsilon^2 \approx 0.66$ and
$\Delta^2 \bdryinv/\varepsilon^2 \approx 0.37$.  The interaction is
positive at both thresholds, as \eqref{eq:ZD_terminal_closed} leads one to expect:
a simultaneous improvement
in both demand factors shifts each boundary further than the sum of
two individual improvements of the same size.  At this evaluation point the effect is about
$1.8$ times as strong at the disinvestment boundary as at the investment
boundary, indicating that joint factor improvements have an asymmetric
impact on the policy: they raise the firm's willingness to retain
existing capacity more than they raise its appetite for new capacity.

As a consequence, the waiting-region width
$w = \bdrydisinv - \bdryinv$ inherits an essentially
positive cross-partial of comparable magnitude,
$\Delta^2 w/\varepsilon^2 \approx 0.30 \approx
\Delta^2 \bdrydisinv/\varepsilon^2 -
\Delta^2 \bdryinv/\varepsilon^2$.  A demand shock that lifts both factors together therefore not only
translates the inaction band toward higher capacities (because both
thresholds rise super-additively), it also widens it: at the higher factor
level the firm tolerates a wider range of capacities before intervening in
either direction.
Since $\pi(cx,c^{4/3}z)=\pi(x,z)$ for the specification computed here, both
thresholds are homogeneous of degree $4/3$ in $x$; with
$x_0+\varepsilon(1,1)=1.1\,x_0$ at $\varepsilon=1$, the width at the
perturbed point is exactly $1.1^{4/3}-1\approx13.6\%$ larger.
The positive
cross-partial is a separate statement: it says the joint effect exceeds the sum
of the two one-at-a-time effects, not that the band is wider in levels.
This super-additive widening is a multi-factor effect that has no analogue in the single-factor model, and it arises with the volatility and discounting parameters held fixed.

Table~\ref{tab:corr_sweep} repeats the cross-partial measurement across the
factor correlation $\rho_{12}\in\{-0.9,-0.5,0,0.5,0.9\}$ at the same baseline
$x_0=(10,10)$ and $\varepsilon=1$ (here with $N_y=600$, against $N_y=800$ in
Table~\ref{tab:boundary_cross_conv}), holding all other parameters fixed.  Along the symmetric diagonal the sum-shock variance
$\sigma_{\Sigma}^2 = \sigma_1^2 + 2\rho_{12}\sigma_1\sigma_2
+ \sigma_2^2$ ranges from $0.008$ to $0.152$; the composite's own variance is
$\sigma_A^2 = (\gamma / 2)^2 \sigma_\Sigma^2$ at $x_1 = x_2$, i.e., 0.0013 to 0.0243.
The waiting-region width
grows monotonically with $\rho_{12}$, from $w\approx42.6$ to
$w\approx71.2$ --- a factor of about $1.7$.  Correlation enters the problem
both through the variance of the composite and, via the cross term in
$\mathcal{L}A$, through its expected level; we report the comparative static
without separating the two.

The two cross-partials behave very differently in $\rho_{12}$.  The
investment cross-partial $\Delta^2 \bdryinv/\varepsilon^2$
varies only weakly --- from $0.36$ at $\rho_{12}=-0.9$ to $0.32$ at
$\rho_{12}=+0.9$, a mild monotone decline --- so the super-additivity of
the investment threshold is, across this sweep, close to independent of how the factors co-move.
The disinvestment cross-partial
$\Delta^2 \bdrydisinv/\varepsilon^2$, by contrast, rises roughly
linearly with $\rho_{12}$ (slope $\approx 0.17$ per unit of correlation,
from $0.56$ at $\rho_{12}=-0.9$ to $0.87$ at $\rho_{12}=+0.9$):
it varies systematically with $\rho_{12}$.
The disinvestment cross-partial is the larger of the two at every correlation,
and their ratio rises from about $1.5$ at $\rho_{12}=-0.9$ to $2.7$ at
$\rho_{12}=0.9$.
The width cross-partial $\Delta^2 w/\varepsilon^2$ therefore
inherits essentially the entire $\rho_{12}$-dependence of
$\Delta^2 \bdrydisinv/\varepsilon^2$, growing from $0.20$ to $0.55$
across the sweep, so the super-additive widening of the inaction band is
present in every correlation regime and is \emph{strengthened} by
positive correlation.
Two comparisons bound the discretisation error.  The two tables overlap at
$\rho_{12}=-0.5$, $\varepsilon=1$ and differ only in the capacity grid
($N_y=800$ against $600$); in addition, we repeat the sweep on a grid coarsened
in both the factor and capacity directions ($300\times300\times400\times60$).
Across both comparisons the widths agree to within $0.4\%$, and all
cross-partials, of the two boundaries and of the width, to within $0.014$ in
capacity units, against a capacity mesh width of $0.25$--$0.50$; this is the
same order as the $\pm\varepsilon$ spread reported above.  All monotonicity
patterns are unchanged.

\begin{table}[htbp]
\centering
\caption{Waiting-region width and free-boundary cross-partials as a function of the factor correlation~\(\rho_{12}\).  All other parameters are held at their baseline values (\(\sigma_1^2 = \sigma_2^2 = 0.0400\)).  The column~\(\sigma_{\Sigma}^2 = \sigma_1^2 + 2\rho_{12}\sigma_1\sigma_2 + \sigma_2^2\) reports the sum-shock variance along the symmetric diagonal \(x_1=x_2\).  Cross-partials are normalised by~\(\varepsilon^2\) (\(\varepsilon = 1\)) and averaged over the \(\pm\varepsilon\) shifts.}
\label{tab:corr_sweep}
\begin{tabular}{r c c c c c}
\toprule
\(\rho_{12}\) & \(\sigma_{\Sigma}^2\) & Width \(w\) & \(\Delta^2 \bdrydisinv/\varepsilon^2\) & \(\Delta^2 \bdryinv/\varepsilon^2\) & \(\Delta^2 w/\varepsilon^2\) \\
\midrule
\(-0.9\) & 0.008000 & 42.5904 & 0.5616 & 0.3640 & 0.1977 \\
\(-0.5\) & 0.040000 & 51.1683 & 0.6520 & 0.3514 & 0.3006 \\
\(0.0\) & 0.080000 & 59.1699 & 0.7383 & 0.3366 & 0.4017 \\
\(0.5\) & 0.120000 & 66.0570 & 0.8142 & 0.3342 & 0.4801 \\
\(0.9\) & 0.152000 & 71.1544 & 0.8691 & 0.3236 & 0.5455 \\
\bottomrule
\end{tabular}
\end{table}

\section{Conclusion}\label{sec:conclusion}

We have studied a finite-horizon reversible investment problem in which
operating revenue depends on the full multi-factor state, so that no
one-dimensional reduction is available.  The correspondence between singular
control and optimal switching reduces the marginal capacity problem to a
parabolic double obstacle problem, for which we prove existence, uniqueness and
local Sobolev regularity.  Under monotonicity and growth conditions on marginal
revenue the state space decomposes into investment, waiting and disinvestment
regions separated by free boundaries, and the Skorokhod problem on the
resulting inaction band produces a reflected capacity process, which we verify
to be optimal.

Section~\ref{sec:implications} reports what this multi-factor structure does to
the policy.  Joint improvements in the demand factors move both thresholds
super-additively, 1.5--2.7 times as strongly at the disinvestment boundary
across the factor correlations we consider, so that the inaction band widens as well as translates; positive factor
correlation strengthens the widening.  Neither the asymmetry nor its dependence
on correlation has an analogue in a single-factor formulation.

Several extensions remain open. It would be natural to allow state-dependent
adjustment costs, more general factor dynamics such as mean-reverting or
jump-diffusion processes, and equilibrium interactions among firms. The
high-dimensional structure of the free-boundary problem also suggests further
work on efficient numerical and simulation-based methods for computing the
investment and disinvestment regions.

\appendix

\section{Regularity of the free boundaries in the state and capacity variables}
\label{app:spatial-boundary-smoothness}
\label{app:capacity-boundary-smoothness}

We first establish the fixed-capacity result used in the capacity-dependent
argument.

\subsection{Fixed-capacity regularity}
\label{app:fixed-capacity-regularity}

\begin{thm}
\label{thm:spatial-boundary-smoothness}
Fix $z>0$. Suppose that, for some
$i_\ast\in\{1,\ldots,n\}$, the function $x\mapsto\pi(x,z)$ is smooth and
\[
    \partial_{x_{i_\ast}}\pi(x,z)>0,
    \qquad x\in(0,\infty)^n.
\]
Then, near every point with $t<T$, each of the two
fixed-capacity free boundaries can be written as a smooth graph in
$\xi_{i_\ast}=\log x_{i_\ast}$, with smooth dependence on time and on the
remaining logarithmic state variables. The $i_\ast$-component of its inward
spatial normal is nonzero.
\end{thm}

For the proof, introduce the logarithmic coordinates and gap functions
\[
    \xi_j=\log x_j,\qquad
    \widetilde{\cal Q}(t,\xi)={\cal Q}(t,e^\xi,z),\qquad
    u_I=K^+-\widetilde{\cal Q},\qquad
    u_D=\widetilde{\cal Q}-K^-.
\]
Writing
\[
    A=\frac12(a_{jk})_{j,k=1}^n,\qquad
    b_j=\mu_j-\frac12a_{jj},\qquad
    {\cal H}=\partial_t+\operatorname{tr}(A D^2)+b\cdot D-\rho,
\]
we have, near the corresponding free boundaries,
\begin{equation}\label{eq:appendix-concise-gap-equations}
\begin{aligned}
    {\cal H}u_I
    &=\bigl(\widetilde\pi-\rho K^+\bigr)\chi_{\{u_I>0\}},\\
    {\cal H}u_D
    &=\bigl(\rho K^--\widetilde\pi\bigr)\chi_{\{u_D>0\}},
\end{aligned}
\end{equation}
where $\widetilde\pi(\xi)=\pi(e^\xi,z)$.
By Theorem~\ref{thm:Q}, $\widetilde{\cal Q}\in W^{1,2}_{p,\mathrm{loc}}$ for every
$p<\infty$; taking $p>n+2$, the parabolic Sobolev embedding makes
$D\widetilde{\cal Q}$ continuous. Hence $u_I$ and $u_D$ are nonnegative and
continuously differentiable in $\xi$, so each attains its minimum value $0$ on
its coincidence set and its spatial gradient vanishes there (spatial smooth
fit): $u=Du=0$ on the coincidence sets.

\begin{lem}\label{lem:appendix-concise-boundary-estimates}
At every point with $t<T$ on the upper and lower free boundaries, respectively,
\[
    \widetilde\pi>\rho K^+,
    \qquad
    \widetilde\pi<\rho K^-.
\]
Moreover, $\partial_t\widetilde{\cal Q}$ is continuous across both free
boundaries.
\end{lem}

\begin{proof}
The obstacle inequalities first give the corresponding weak inequalities.
Fix a free-boundary point $(t_0,\xi_0)$ and let $\Xi$ denote the
logarithmic diffusion satisfying
$d\Xi_s=b\,ds+\Sigma\,dB_s$, with $\Xi_0=\xi_0$, where $B$ and $\Sigma$ are
the Brownian motion and volatility matrix in \eqref{eq:GBM}.
Choose $r>0$ sufficiently small and set $U:=B_r(\xi_0)$.
For this process, the It\^o--Krylov formula gives,
for every bounded stopping time $\tau$ before $\Xi$ exits $U$,
\begin{equation}\label{eq:appendix-concise-stopping-bounds}
\begin{aligned}
    u_I(t_0,\xi_0)&\geq
      \mathbb E_{\xi_0}\!\int_0^\tau e^{-\rho s}
      \bigl(\rho K^+-\widetilde\pi(\Xi_s)\bigr)\,ds,\\
    u_D(t_0,\xi_0)&\geq
      \mathbb E_{\xi_0}\!\int_0^\tau e^{-\rho s}
      \bigl(\widetilde\pi(\Xi_s)-\rho K^-\bigr)\,ds.
\end{aligned}
\end{equation}
Assume that $\widetilde\pi(\xi_0)=\rho K^+$. Set
$g:=D\widetilde\pi(\xi_0)$ and $W_s:=g^\top\Sigma B_s$. Let $\tau_U$ be the
first exit time of $\Xi$ from $U$ and, for $h>0$, define
\[
 \tau_h^\pm
 :=\inf\{s\geq0:W_s=\pm\sqrt h\}\wedge h\wedge\tau_U.
\]
Taylor expansion and Brownian scaling give
\[
    \mathbb E_{\xi_0}\!\int_0^{\tau_h^+}e^{-\rho s}
      \bigl(\rho K^+-\widetilde\pi(\Xi_s)\bigr)\,ds
      =c h^{3/2}+O(h^2),
      \qquad c>0.
\]
For small $h$ this contradicts
\eqref{eq:appendix-concise-stopping-bounds} and
$u_I(t_0,\xi_0)=0$.
The same argument with $\tau_h^-$ gives the corresponding contradiction
at the lower free boundary.

For the continuity assertion, fix a free-boundary point
$P_0=(t_0,\xi_0)$ and a compact backward cylinder
${\cal C}=[t_b-R^2,t_b]\times B_R(\xi_0)$ containing it,
with $t_b>t_0$.
By continuity and the positivity
assumption, $\partial_{\xi_{i_\ast}}\widetilde\pi$ is bounded below by a
positive constant on ${\cal C}$.
Use the penalized solutions ${\cal Q}^{N,\varepsilon}$ of
\eqref{eq:penal}, expressed in logarithmic coordinates and restricted to
${\cal C}$. Taking $\varepsilon\downarrow0$ first at fixed $N$, and then
$N\to\infty$, local $W^{1,2}_p$ compactness yields a diagonal sequence
$\widetilde{\cal Q}_k$ such that
\[
 \widetilde{\cal Q}_k\to\widetilde{\cal Q},\qquad
 D_\xi\widetilde{\cal Q}_k\to D_\xi\widetilde{\cal Q}
\]
uniformly on ${\cal C}$. Write $\beta_{1,k}$ and $\beta_{2,k}$ for
the corresponding penalties, and set
\[
 p_k=-\partial_t\widetilde{\cal Q}_k,\qquad
 q_k=\partial_{\xi_{i_\ast}}\widetilde{\cal Q}_k,\qquad
 c_k=\beta'_{1,k}+\beta'_{2,k}\geq0.
\]
The penalized estimate used in the proof of
\eqref{eq:Qt-bound-full} gives
\[
 0\leq-\partial_t\widetilde{\cal Q}_k=p_k\leq M
 \qquad\text{on }{\cal C},
\]
with $M$ independent of $k$. The gradient convergence and
$\partial_{\xi_{i_\ast}}\widetilde{\cal Q}\geq0$ give
\mbox{$\eta_k:=\|(q_k)^-\|_{L^\infty({\cal C})}\to0$}. Moreover,
\[
 ({\cal H}-c_k)p_k=0,
 \qquad
 ({\cal H}-c_k)q_k=-\partial_{\xi_{i_\ast}}\widetilde\pi.
\]
With $h=t_b-t$ and $\psi(\xi)=|\xi-\xi_0|^2$, choose $K\geq M$ and then
$C$ sufficiently large. The function
\[
 W_k=hp_k-Cq_k-K\psi-C\eta_k
\]
satisfies $({\cal H}-c_k)W_k>0$ in ${\cal C}$. Since
$0\leq h\leq R^2$, $p_k\leq M$, and
$q_k\geq-\eta_k$, the choice $K\geq M$ gives
$W_k\leq0$ on the parabolic boundary.
The maximum principle and passage to the limit $k\to\infty$ give
\[
    (t_b-t)(-\partial_t\widetilde{\cal Q})
    \leq C\partial_{\xi_{i_\ast}}\widetilde{\cal Q}
       +K|\xi-\xi_0|^2.
\]
Letting $(t,\xi)\to P_0$ in this estimate from the waiting region and
using spatial smooth fit, $t_b>t_0$, and
$0\leq-\partial_t\widetilde{\cal Q}$, we obtain
$\partial_t\widetilde{\cal Q}\to0$ from the waiting side. On either
coincidence region, $\widetilde{\cal Q}=K^\pm$ and hence
$\partial_t\widetilde{\cal Q}=0$, proving continuity across the free boundary.
\end{proof}

\begin{proof}[Proof of Theorem~\ref{thm:spatial-boundary-smoothness}]
We first prove that all points with \(t<T\) on the two
fixed-capacity free boundaries are regular and that both boundaries are locally parabolic
\(C^1\) hypersurfaces. Let \(P_0=(t_0,\xi_0)\), with \(t_0<T\), be a point on
the upper free boundary. By
Lemma~\ref{lem:appendix-concise-boundary-estimates}, there exist \(\delta>0\)
and a neighborhood of \(P_0\) such that
\[
 \widetilde\pi-\rho K^+\geq2\delta,
 \qquad 0\leq-\partial_t\widetilde{\cal Q}\leq\delta.
\]
A fixed-time quadratic barrier gives the nondegeneracy estimate
\begin{equation*}
 \sup_{B_r(y)}u_I(t,\cdot)\geq cr^2
 \qquad\text{whenever }u_I(t,y)>0.
\end{equation*}
For the lower free boundary, the same estimate follows from
 \(\rho K^--\widetilde\pi\geq2\delta\) and
\(\partial_t\widetilde{\cal Q}\leq0\).

At a free-boundary point \((t,\xi_0)\), monotonicity in the state
variables implies that \(B_r(\xi_0)\) contains a ball of radius
\(r/(2\sqrt n)\) in the corresponding coincidence set.
For the upper free boundary, let \(B_{\alpha r}(\eta)\) be the coincidence ball above, where
\(\alpha=1/(2\sqrt n)\). If \(u_I(t+h,y)>0\) for some
\(y\in B_{\alpha r/2}(\eta)\) and \(h\leq c_0r^2\), the nondegeneracy
estimate and \(0\leq\partial_tu_I\leq\delta\) give
\(c\alpha^2r^2/4\leq\delta c_0r^2\). Choosing
\(c_0<c\alpha^2/(4\delta)\) gives a contradiction. Thus
\(B_{\alpha r/2}(\eta)\) remains in the coincidence set on
\([t,t+c_0r^2]\).
For the lower free boundary, the same property follows directly from time
monotonicity. Therefore, the coincidence set satisfies the uniform parabolic
thickness condition at every free-boundary point.

For \(u\in\{u_I,u_D\}\), write \(f\) for the corresponding
source in \eqref{eq:appendix-concise-gap-equations}. As in the proof of
Lemma~\ref{lem:local-strict-comparison}, set
\[
 \tau=t_0-t,\qquad
 U(y,\tau)=e^{\rho\tau}
 u\bigl(t_0-\tau,\xi_0+A^{1/2}y-b\tau\bigr).
\]
The moving coordinates remove the drift, and the exponential factor
removes the zeroth-order term. Thus
\[
\begin{aligned}
 (\Delta_y-\partial_\tau)U
 &=F(y,\tau)\chi_{\{U>0\}},\\
 F(y,\tau)
 &=e^{\rho\tau}f\bigl(\xi_0+A^{1/2}y-b\tau\bigr).
\end{aligned}
\]
Here \(F(0,0)>0\), and \(F\) is locally H\"older continuous.
By the classification in
\citet[Theorem~1.2]{LindgrenMonneau2015}, the nondegeneracy and thickness
estimates exclude the degenerate and singular alternatives.
Thus every free-boundary point is regular.
By \citet[Theorem~1.7]{LindgrenMonneau2015}, the free boundary is
locally parabolic \(C^1\).
For the finite-order conclusion, it suffices that
\(f\in C^{2,\alpha_0}_{\mathrm{loc}}\), with \(\alpha_0\in(0,1)\).
The regular-profile estimates give a linearly nondegenerate inward
directional derivative of \(U\), as in
\citet[Lemma~6.1]{Kukuljan2022}.
Apply \citet[Theorems~1.2--1.3]{Kukuljan2022} to the spatial and time
derivatives of \(U\), following the proof of Corollary~1.4 there.
\(D_yF\) and \(\partial_\tau F\) have the required H\"older bounds.
Thus the boundary graphs are locally \(C^{1,2}_{\beta}\) for some
\(\beta\in(0,\alpha_0)\).
When \(f\) is smooth, the same argument can be iterated to obtain
\(C^\infty\) regularity.

Finally, we prove the graph direction. Set
 \(q=\partial_{\xi_{i_\ast}}\widetilde{\cal Q}\geq0\). In
each adjacent
waiting region,
\[
 {\cal H}q=-\partial_{\xi_{i_\ast}}\widetilde\pi<0.
\]
By the strong maximum principle and the boundary-point lemma,
 \(\partial_\nu q\neq0\) on the free boundary.
At a regular point \(P_0=(t_0,\xi_0)\), the corresponding gap
function \(u\in\{u_I,u_D\}\) has the quadratic expansion
\[
 u(t,\xi)
 =\frac{c_0}{2}\bigl((\xi-\xi_0)\cdot\nu\bigr)_+^2
 +o\bigl(|\xi-\xi_0|^2+|t-t_0|\bigr),
 \qquad c_0>0.
\]
Since \(u_I=K^+-\widetilde{\cal Q}\) and
\(u_D=\widetilde{\cal Q}-K^-\), differentiation from the waiting side gives
\[
 \partial_\nu q=\mp c_0\nu_{i_\ast},
\]
where the signs correspond to the upper and lower free boundaries,
respectively.
Thus \(\nu_{i_\ast}\neq0\). By the implicit function theorem, the
free boundary is locally a graph in the \(\xi_{i_\ast}\)-direction.
This graph-direction argument remains valid for locally
\(C^{2,\alpha_0}\) sources with
\(\partial_{\xi_{i_\ast}}\widetilde\pi>0\).
\end{proof}
\subsection{Assumptions and logarithmic formulation}
\label{app:capacity-data}

We now prove Theorem~\ref{thm:capacity-boundary-regularity}. We keep
the logarithmic state variables and the operator ${\cal H}$ from
Subsection~\ref{app:fixed-capacity-regularity}, and set
\[
 s:=\log z,
 \qquad
 V(t,\xi,s):={\cal Q}(t,e^\xi,e^s),
 \qquad
 \widetilde\pi(\xi,s):=\pi(e^\xi,e^s).
\]
We first state the regularity assumptions in these variables.

\begin{as}[First-order data $(D_1)$]\label{as:capacity-first-order-data}
Assume that
\[
 \partial_{\xi_i}\widetilde\pi>0\quad(i=1,\ldots,n),
 \qquad
 \partial_s\widetilde\pi<0.
\]
For some \(\alpha_0\in(0,1)\), the map
\[
 s\longmapsto\widetilde\pi(\cdot,s)
\]
is of class \(C^1\) with values in \(C_{\mathrm{loc}}^{2,\alpha_0}
(\mathbb R^n)\).
\end{as}

\begin{as}[Smooth data $(D_\infty)$]\label{as:capacity-boundary-smoothness}
In addition, assume that
 \(\pi\in C^\infty((0,\infty)^n\times(0,\infty))\).
\end{as}

Fix \(i_\ast\in\{1,\ldots,n\}\). Under \((D_1)\), the ratio below
is positive and continuous. Hence, for every compact set
\(E\Subset\mathbb R^n\times\mathbb R\), there are constants
\(0<\lambda_-(E)\leq\lambda_+(E)<\infty\) such that
\begin{equation}\label{eq:capacity-log-compensation}
 \lambda_-(E)
 \leq
 \frac{-\partial_s\widetilde\pi(\xi,s)}
      {\partial_{\xi_{i_\ast}}\widetilde\pi(\xi,s)}
 =
 \frac{-z\partial_z\pi(e^\xi,e^s)}
      {x_{i_\ast}\partial_{x_{i_\ast}}\pi(e^\xi,e^s)}
 \leq\lambda_+(E),
 \qquad (\xi,s)\in E.
\end{equation}
Here \(x=e^\xi\) and \(z=e^s\). This is a consequence of
\((D_1)\), not an additional assumption.
The ratio in \eqref{eq:capacity-log-compensation} is the local change
in \(\log x_{i_\ast}\) required to offset a change in \(\log z\).

\begin{as}[Exact compensation $(C_{\mathrm{eq}})$]
\label{as:capacity-exact-compensation}
There are an index \(i_\ast\in\{1,\ldots,n\}\) and a function
\(\lambda\in C^\infty((0,\infty))\) such that
\begin{equation}\label{eq:capacity-exact-compensation}
\begin{gathered}
 0<\lambda_-\leq\lambda(z)\leq\lambda_+<\infty,\\
 -z\partial_z\pi(x,z)
   =\lambda(z)x_{i_\ast}\partial_{x_{i_\ast}}\pi(x,z),
 \qquad (x,z)\in(0,\infty)^n\times(0,\infty).
\end{gathered}
\end{equation}
\end{as}

\par\smallskip\noindent
Define
\[
 L(s):=\int_0^s\lambda(e^r)\,dr.
\]
Then, by integration along characteristics, we obtain
\begin{equation}\label{eq:capacity-exact-source-invariance}
 \widetilde\pi\bigl(\xi+(L(s_2)-L(s_1))e_{i_\ast},s_2\bigr)
   =\widetilde\pi(\xi,s_1),
 \qquad s_1,s_2\in\mathbb R.
\end{equation}

For each fixed \(s\), the fixed-capacity formulation
\eqref{eq:appendix-concise-gap-equations} applies with
\(\widetilde{\cal Q}=V(\cdot,\cdot,s)\) and
\(\widetilde\pi=\widetilde\pi(\cdot,s)\). Denote the corresponding sources by
\begin{equation}\label{eq:capacity-source-definitions}
    f_I(\xi,s):=\widetilde\pi(\xi,s)-\rho K^+,
    \qquad
    f_D(\xi,s):=\rho K^--\widetilde\pi(\xi,s).
\end{equation}
By Lemma~\ref{lem:appendix-concise-boundary-estimates},
\begin{equation}\label{eq:capacity-strict-sources}
\begin{aligned}
    f_I&>0 &&\quad\hbox{on the upper free boundary},\\
    f_D&>0 &&\quad\hbox{on the lower free boundary}.
\end{aligned}
\end{equation}

For \(P=(t_0,\xi_0)\), set
\[
 {\cal C}_r(P)
 :=(t_0-r^2,t_0+r^2)\times B_r(\xi_0).
\]

\subsection{Local first-order regularity}
\label{app:capacity-monotone-compensation}

Write \(\xi_{(i_\ast)}\) for the vector obtained from \(\xi\) by
removing its \(i_\ast\)-th component.
For fixed \((t,\xi_{(i_\ast)},s)\) and \(r\in\mathbb R\), let
\[
 \xi[r]
 :=(\xi_1,\ldots,\xi_{i_\ast-1},r,
       \xi_{i_\ast+1},\ldots,\xi_n).
\]
Near a point of the upper or lower free boundary, we write its local graph as
\[
 \xi_{i_\ast}=\gamma_\ell(t,\xi_{(i_\ast)},s),
 \qquad \ell\in\{I,D\}.
\]
More precisely, after shrinking the neighborhood if necessary,
\[
\begin{aligned}
 V(t,\xi[r],s)=K^+
 &\quad\Longleftrightarrow\quad
 r\geq\gamma_I(t,\xi_{(i_\ast)},s)
 &&\text{near the upper free boundary},\\
 V(t,\xi[r],s)=K^-
 &\quad\Longleftrightarrow\quad
 r\leq\gamma_D(t,\xi_{(i_\ast)},s)
 &&\text{near the lower free boundary}.
\end{aligned}
\]
These equivalences are understood only within the chosen neighborhood.

\begin{lem}\label{lem:capacity-logarithmic-boundary-estimates}
Assume~\ref{as:capacity-first-order-data}.
Fix \(\ell\in\{I,D\}\) and a point \((P_0,s_0)\), where
\(P_0=(t_0,\xi_0)\), \(0\leq t_0<T\), lies on the corresponding
fixed-\(s_0\) free boundary.
Then there are a cylinder \({\cal C}_r(P_0)\), an interval \(J\ni s_0\),
and \(\beta\in(0,1)\) such that, for every \(s\in J\), the corresponding
fixed-\(s\) free boundary is represented over a common
\((t,\xi_{(i_\ast)})\)-neighborhood by
\[
 \xi_{i_\ast}=\gamma_\ell(t,\xi_{(i_\ast)},s).
\]
On every smaller cylinder, these graph functions have uniform
\(C^{1,2}_{\beta}\) bounds and depend
continuously on \(s\), locally uniformly
in \((t,\xi_{(i_\ast)})\).
\end{lem}

\begin{proof}
By Theorem~\ref{thm:free-boundary-continuity},
\(V(\cdot,\cdot,s)\to V(\cdot,\cdot,s_0)\) locally uniformly as \(s\to s_0\).
Choose a smaller cylinder \({\cal C}_R(P_0)\) and interval
\(J\ni s_0\) so that the other obstacle is inactive. By
Lemma~\ref{lem:appendix-concise-boundary-estimates} and
Assumption~\ref{as:capacity-first-order-data}, there are constants
\(c_f,M>0\), independent of \(s\in J\), such that
\[
 f_\ell(\xi,s)\geq c_f,\qquad
 \|f_\ell(\cdot,s)\|_{C^{2,\alpha_0}(B_R(\xi_0))}\leq M
 \quad (\xi\in B_R(\xi_0),\ s\in J).
\]
Also, \(0\leq u_\ell\leq K^+-K^-\).
The finite-order part of the proof of
Theorem~\ref{thm:spatial-boundary-smoothness} therefore applies at each
fixed \(s\); its \(C^\infty\) conclusion is not used here.

The quadratic nondegeneracy constants can be chosen using
\(c_f\), \(M\), the fixed operator, and the cylinder, and hence
independently of \(s\in J\). Every reference boundary point at \(s=s_0\)
is regular.
On a compact subset of the limiting waiting region, the gap has
a positive minimum, so the subset remains in the waiting region for
\(s\) near \(s_0\). Suppose that positive-gap points remain in a compact subset of
the interior of the limiting contact set as \(s\to s_0\).
Nondegeneracy gives a positive lower bound for the supremum of the gap
on a fixed surrounding cylinder, independent of \(s\).
This contradicts uniform convergence to zero there. Since every limiting
boundary point is regular, it can be approached by interior contact
points. It follows that the contact sets converge locally in Hausdorff
distance as \(s\to s_0\).

Choose a fixed scale at \((P_0,s_0)\) at which the normalized gap is
sufficiently close to its regular half-space profile. At this scale, local
uniform convergence preserves the required closeness at nearby free-boundary
points for \(s\) in a smaller interval \(J\). Since the operator is fixed and
the source bounds above are uniform, the regular-point argument of
\citet[Section~4.1]{LindgrenMonneau2015} gives a common graph neighborhood
and a uniform parabolic \(C^1\) modulus.
These graph estimates, together with the Hausdorff convergence above,
imply that
\(\gamma_\ell(\cdot,\cdot,s)\to
\gamma_\ell(\cdot,\cdot,s_0)\) locally uniformly as \(s\to s_0\).
The same uniform profile estimates give a uniform linear
nondegeneracy bound for a fixed inward directional derivative of the gap,
as in \citet[Lemma~6.1]{Kukuljan2022}. Thus the constants in the boundary
Harnack argument below can also be chosen independently of \(s\).
Apply the finite-order boundary Harnack argument of
\citet[Theorems~1.2--1.3 and the proof of Corollary~1.4]{Kukuljan2022}
to the state and time derivatives of the gap.
Use the drift-removing and exponential transformations in the
proof of Theorem~\ref{thm:spatial-boundary-smoothness}, which are independent
of \(s\). The transformed sources and their first space and time
derivatives have uniform H\"older bounds. Together with the common graph
modulus and the uniform linear nondegeneracy bound, these give
\(C^{1,2}_{\beta}\) graph bounds on smaller cylinders, with a common
\(\beta\in(0,\alpha_0)\) and constants independent of \(s\in J\).

The graph-direction argument at the end of the proof of
Theorem~\ref{thm:spatial-boundary-smoothness} uses
\(\partial_{\xi_{i_\ast}}\widetilde\pi>0\) and gives
\(\nu_{i_\ast}\neq0\) at the reference point.
The uniform graph estimates and convergence preserve this property
on a smaller cylinder. This completes the proof.
\end{proof}

The next lemma gives the capacity derivative of the value function
and its behavior near the free boundary.

\begin{lem}\label{lem:capacity-local-stability}
Assume~\ref{as:capacity-first-order-data}. The derivative
\(w:=-\partial_sV\) exists and is continuous. It vanishes on the contact
sets and, in the waiting region, satisfies
\[
 {\cal H}w=\partial_s\widetilde\pi<0.
\]
In a smaller neighborhood as in
Lemma~\ref{lem:capacity-logarithmic-boundary-estimates}, set
\(q:=\partial_{\xi_{i_\ast}}V\).
The quotients
\[
 \frac{w}{q},\qquad
 -\frac{\partial_tV}{q},\qquad
 -\frac{\partial_{\xi_j}V}{q}\quad(j\neq i_\ast)
\]
extend continuously from the waiting region to the free boundary,
jointly in \((t,\xi,s)\). Moreover,
\begin{equation}\label{eq:capacity-local-quotient-bounds}
 0<c\leq\frac{w}{q}\leq C<\infty
\end{equation}
there, with constants depending on the chosen neighborhood.
\end{lem}

\begin{proof}
Let \(\Xi_r=\xi+br+\Sigma B_r\) be the logarithmic state process from
Subsection~\ref{app:fixed-capacity-regularity}.
For a capacity parameter \(a\), let \(\tau_I(a)\) and \(\tau_D(a)\)
be its first hitting times of the upper and lower contact sets:
\[
\begin{aligned}
 \tau_I(a)&:=\inf\{r\in[0,T-t):V(t+r,\Xi_r,a)=K^+\}\wedge(T-t),\\
 \tau_D(a)&:=\inf\{r\in[0,T-t):V(t+r,\Xi_r,a)=K^-\}\wedge(T-t).
\end{aligned}
\]
We use \(\inf\varnothing=\infty\).
For \(h>0\), set
\[
 d_h:=\frac{V(\cdot,\cdot,s)-V(\cdot,\cdot,s+h)}{h},
 \qquad
 g_h:=\frac{\widetilde\pi(\cdot,s)-\widetilde\pi(\cdot,s+h)}{h}.
\]
Both functions are nonnegative. By monotonicity in \(s\),
\(\tau_I(s)\leq\tau_I(s+h)\) and
\(\tau_D(s+h)\leq\tau_D(s)\).
The variational inequalities and the It\^o--Krylov formula give
\begin{equation}\label{eq:capacity-difference-quotient-squeeze}
\begin{aligned}
 \mathbb E_\xi\!\int_0^{\tau_I(s)\wedge\tau_D(s+h)}
       e^{-\rho r}g_h(\Xi_r)\,dr
 &\leq d_h(t,\xi)\\
 &\leq
 \mathbb E_\xi\!\int_0^{\tau_I(s+h)\wedge\tau_D(s)}
       e^{-\rho r}g_h(\Xi_r)\,dr .
\end{aligned}
\end{equation}
For the first inequality, both solutions are in the waiting region
before the stopping time, so \({\cal H}d_h=-g_h\); the inequality then
follows from \(d_h\geq0\) at the stopping time.
For the second, on
\(\{V(\cdot,\cdot,s)>K^-,\ V(\cdot,\cdot,s+h)<K^+\}\) we have
\({\cal H}d_h\geq-g_h\).
At its exit the two values agree with the same obstacle, and they also
agree at \(T\). For both inequalities, localization and
Assumption~\ref{as:free_boundary}\textnormal{(iii)} justify passage
to the unbounded state space.

At each fixed-capacity boundary point, the diffusion enters the interior
of the contact set immediately after hitting it. This follows from
uniform ellipticity and the local
\(C^{1,2}_{\beta}\) graph in
Lemma~\ref{lem:capacity-logarithmic-boundary-estimates}.
The local convergence of the contact sets then implies convergence of
these hitting times as the initial point and capacity vary.
Indeed, couple a convergent sequence of initial
points and capacity parameters using the same Brownian motion, and write
\(\tau_n\) and \(\tau\) for the corresponding hitting times.
The limiting paths stay a positive distance from the closed contact
set on compact intervals before the first hit. Immediately after the
hit they enter its interior, which persists for nearby parameters.
These two observations give, respectively,
\(\liminf \tau_n\geq \tau\) and \(\limsup \tau_n\leq \tau\), and hence
the convergence of the hitting times.
The same conclusion holds for the hitting times truncated at \(T-t\).

Let \(\tau(s):=\tau_I(s)\wedge\tau_D(s)\).
Since \(g_h\to-\partial_s\widetilde\pi(\cdot,s)\), dominated convergence
in \eqref{eq:capacity-difference-quotient-squeeze} yields
\begin{equation}\label{eq:capacity-derivative-representation}
 w(t,\xi,s)
 =\mathbb E_\xi\!\int_0^{\tau(s)}
       e^{-\rho r}\bigl(-\partial_s\widetilde\pi(\Xi_r,s)\bigr)\,dr .
\end{equation}
The bound in Assumption~\ref{as:free_boundary}\textnormal{(iii)}
and the finite moments of the geometric Brownian motion give an
integrable bound, locally uniform in the initial point and \(s\).
Applying the same argument with \(s-h\) gives the left derivative.
The hitting-time convergence also proves continuity of \(w\), including
\(w=0\) on both contact sets. Interior parabolic estimates applied to
the difference quotients yield
\({\cal H}w=\partial_s\widetilde\pi\) in the waiting region.

Fix the local graph neighborhood. The function \(q\) is nonnegative,
vanishes on the free boundary, and satisfies
\[
 {\cal H}q=-\partial_{\xi_{i_\ast}}\widetilde\pi<0
\]
in the waiting region.
The strong maximum principle, the boundary-point lemma, and boundary
Schauder estimates give, on a smaller neighborhood,
\[
 c\,d\leq q,w\leq C\,d,
\]
where \(d\) is the distance in the state variables to the corresponding
fixed-time boundary.
These estimates are uniform for \(s\) near \(s_0\): the graph bounds
are uniform, the two negative sources are bounded away from zero,
and \(q,w\) are locally bounded. This proves
\eqref{eq:capacity-local-quotient-bounds}.

The state and time derivatives of \(V\) vanish on the free boundary
by spatial smooth fit and
Lemma~\ref{lem:appendix-concise-boundary-estimates}. In the waiting region,
\[
 {\cal H}(\partial_{\xi_j}V)
   =-\partial_{\xi_j}\widetilde\pi,
 \qquad
 {\cal H}(\partial_tV)=0.
\]
Let \(v\) denote any of \(q,w,\partial_tV\), or
\(\partial_{\xi_j}V\) with \(j\neq i_\ast\), and let \(\nu_s(P)\) be the
inward spatial unit normal at a fixed-\(s\) boundary point \(P=(t,\xi)\).
The boundary estimates give
\[
 v(t,\xi+r\nu_s(P),s)
   =r\,\partial_{\nu_s}v(P,s)+o(r),
 \qquad
 \partial_{\nu_s}q(P,s)\geq c>0,
\]
as \(r\downarrow0\), with the remainder uniform on smaller boundary
patches and for \(s\) near \(s_0\). Hence \(v/q\) extends continuously
to the boundary with value
\(\partial_{\nu_s}v/\partial_{\nu_s}q\).
To verify continuity in \(s\), let \(s_k\to s_0\) and flatten the
boundaries using the graphs in
Lemma~\ref{lem:capacity-logarithmic-boundary-estimates}.
By interpolation, these coordinate maps converge locally together with
their spatial derivatives. Denote the pullback of \(v(\cdot,\cdot,s)\) to the common flat
half-cylinder by \(\widehat v_s(t,y)\). The uniform boundary Schauder
estimates make the spatial gradients of these pullbacks precompact up
to the flat boundary. Every subsequential limit agrees with
\(\nabla_y\widehat v_{s_0}\) in the interior, and therefore also on the
flat boundary by continuity. Thus
\[
 \nabla_y\widehat v_{s_k}\longrightarrow\nabla_y\widehat v_{s_0}
 \qquad\text{locally uniformly up to the flat boundary}.
\]
Transforming back gives locally uniform convergence of the one-sided
normal derivatives and hence of their ratios, since
\(\partial_{\nu_s}q\geq c>0\).
\end{proof}

\begin{proof}[Proof of
Theorem~\ref{thm:capacity-boundary-regularity} (the \(C^1\) case)]
Fix \(\ell\in\{I,D\}\) and a point on its free boundary.
Use the local graph neighborhood of
Lemma~\ref{lem:capacity-logarithmic-boundary-estimates}.
For \(\varepsilon>0\) small, define \(\gamma_\ell^\varepsilon\) by
\[
 V(t,\xi[\gamma_\ell^\varepsilon],s)
 =
 \begin{cases}
 K^+-\varepsilon,&\ell=I,\\
 K^-+\varepsilon,&\ell=D.
 \end{cases}
\]
These level sets lie in the waiting region, where \(q>0\).
The implicit function theorem and
Lemma~\ref{lem:capacity-local-stability} give
\[
 \partial_s\gamma_\ell^\varepsilon=\frac{w}{q},\qquad
 \partial_t\gamma_\ell^\varepsilon=-\frac{\partial_tV}{q},\qquad
 \partial_{\xi_j}\gamma_\ell^\varepsilon
   =-\frac{\partial_{\xi_j}V}{q}\quad(j\neq i_\ast),
\]
evaluated on the corresponding level set.
The quadratic growth and nondegeneracy estimates imply
\(\gamma_\ell^\varepsilon\to\gamma_\ell\) locally uniformly.
By the continuous boundary limits in
Lemma~\ref{lem:capacity-local-stability}, all displayed derivatives
also converge locally uniformly.
Integrating along coordinate segments shows that their limits are
the derivatives of \(\gamma_\ell\). Hence \(\gamma_\ell\) is jointly
\(C^1\), and
\begin{equation}\label{eq:capacity-graph-s-derivative-bounds}
 0<c\leq\partial_s\gamma_\ell\leq C<\infty
\end{equation}
locally. In particular, the corresponding local shift bounds follow
by integration in \(s\); their constants need not be uniform over the
whole free boundary.

Apply the \(C^1\) implicit function theorem once more to
\[
 \xi_{i_\ast}=\gamma_\ell(t,\xi_{(i_\ast)},s).
\]
Since \(\partial_s\gamma_\ell>0\), it gives a local \(C^1\) solution
\(s=s_\ell(t,\xi)\).
By the ordering of the contact sets,
\(e^{s_\ell(t,\xi)}={\cal Z}_\ell(t,e^\xi)\).
The logarithmic coordinate changes are smooth, and the boundary point
was arbitrary. This proves the \(C^1\) assertion.
\end{proof}

\subsection{Exact compensation and smoothness}
\label{app:capacity-exact-compensation-proof}

By $(C_{\mathrm{eq}})$, a translation in $\xi_{i_\ast}$ exactly
compensates for a change in capacity.

\begin{proof}[Proof of
Theorem~\ref{thm:capacity-boundary-regularity} (the \(C^\infty\) case)]
By \eqref{eq:capacity-exact-source-invariance},
 \(V(t,\xi+(L(s_2)-L(s_1))e_{i_\ast},s_2)\) and \(V(t,\xi,s_1)\)
have the same source. Since \({\cal H}\) is translation invariant and
the obstacles and terminal condition are constant, they
 solve the same double-obstacle problem.
By uniqueness in Theorem~\ref{thm:Q},
\begin{equation}\label{eq:capacity-exact-value-invariance}
 V\bigl(t,\xi+(L(s_2)-L(s_1))e_{i_\ast},s_2\bigr)
   =V(t,\xi,s_1).
\end{equation}

Fix \(\ell\in\{I,D\}\) and
\((t_0,x_0)\in{\cal D}_T\), and let
\(\xi_0=\log x_0\) and
\(s_0=\log{\cal Z}_{\ell}(t_0,x_0)\).
By Lemma~\ref{lem:free-boundary-range-limits}, \(s_0\) is finite.
Set \(s_1=s_0\) and \(s_2=0\) in
\eqref{eq:capacity-exact-value-invariance}. Then this boundary point is
mapped to \((t_0,\xi_0-L(s_0)e_{i_\ast},0)\).
 By Theorem~\ref{thm:spatial-boundary-smoothness}, the free boundary
at \(s=0\) has a smooth local graph
\[
 \xi_{i_\ast}
   =\gamma_{\ell,0}(t,\xi_{(i_\ast)})
\]
through that point.
By the same identity, for every nearby $s$,
\begin{equation}\label{eq:capacity-exact-graph-translation}
 \gamma_{\ell}(t,\xi_{(i_\ast)},s)
 =\gamma_{\ell,0}(t,\xi_{(i_\ast)})+L(s),
 \qquad \ell\in\{I,D\}.
\end{equation}

At a point on one of the two capacity boundaries, the identity
\(\xi_{i_\ast}
=\gamma_{\ell}(t,\xi_{(i_\ast)},s)\) and
\eqref{eq:capacity-exact-graph-translation}
imply
\begin{equation}\label{eq:capacity-exact-Z-formula}
 {\cal Z}_{\ell}(t,e^\xi)
 =\exp\left(
     L^{-1}\bigl(
       \xi_{i_\ast}
       -\gamma_{\ell,0}(t,\xi_{(i_\ast)})
            \bigr)
          \right),
 \qquad \ell\in\{I,D\}.
\end{equation}
Since \(L'(s_0)=\lambda(e^{s_0})>0\), the local inverse \(L^{-1}\)
is smooth near \(L(s_0)\). Hence the right-hand side is smooth near
\((t_0,\xi_0)\). Since the point and
 \(\ell\) were
arbitrary, the result follows.
\end{proof}

\section{Numerical scheme}\label{app:numerics}

The computations in Sections~\ref{sec:Numeric} and~\ref{sec:implications}
solve \eqref{eq:double-obs-n-dim} for $n=2$ by a standard projected implicit
finite-difference scheme.  Since \eqref{eq:double-obs-n-dim} contains no
derivative in $z$, it is solved independently at each capacity level of the
grid.  Below, $h_i$ is the mesh width in $x_i$, $\Delta t$ the time step, and
$Q^k_{i,j}$ the approximation of ${\cal Q}(t_k,x_{1,i},x_{2,j},z)$.
\begin{center}
\small
\begin{tabular}{@{}p{0.2\textwidth}p{0.65\textwidth}@{}}
\toprule
Item & Choice \\
\midrule
Grid
  & Except where otherwise stated: Uniform in $(x_1,x_2)\in[0,100]^2$ with $400\times400$ nodes, including
    $x_i=0$; $N_y$ uniform capacity levels in $[0.01,200]$; $N_t=60$ uniform
    time steps, $\Delta t=T/N_t$. \\
Time stepping
  & Backward Euler from ${\cal Q}(T,\cdot)=K^-$: $Q^k$ solves
    $(Q^{k+1}-Q^k)/\Delta t+{\cal L}_hQ^k-\rho Q^k+\pi=0$ at the interior
    nodes, by sparse LU factorisation. \\
Obstacles
  & After each linear step,
    $Q^k\leftarrow\min\{\max\{Q^k,K^-\},K^+\}$ (explicit projection; no
    complementarity problem is solved). \\
Second derivatives
  & Central: $(Q_{i+1}-2Q_i+Q_{i-1})/h_1^2$, and likewise in $x_2$. \\
First derivatives
  & Forward, which is upwind since $\mu_i\ge0$: $(Q_{i+1}-Q_i)/h_1$, and
    likewise in $x_2$. \\
Mixed derivative
  & Seven-point stencil chosen by the sign of $a_{12}$; for $a_{12}<0$,
    $(Q_{i+1,j}+Q_{i-1,j}+Q_{i,j+1}+Q_{i,j-1}-2Q_{i,j}-Q_{i+1,j-1}-Q_{i-1,j+1})/(2h_1h_2)$,
    and for $a_{12}\ge0$ the reflection with $Q_{i+1,j+1}$, $Q_{i-1,j-1}$. \\
Boundary $x_i=100$
  & The equation is kept, with every derivative involving $x_i$ (first,
    second and mixed) set to zero. \\
Free boundaries
  & At each node in $(x_1,x_2)$, ${\cal Z}_{\cal I}$ is the largest level
    with $Q=K^+$ and ${\cal Z}_{\cal D}$ the smallest level with $Q=K^-$. \\
Cross-partials
  & For Section~\ref{sec:implications}, the boundary surfaces are smoothed
    in $(x_1,x_2)$ by a separable Gaussian kernel with a standard deviation
    of $3$ grid cells, truncated at $3$ standard deviations, and evaluated
    off the grid by bilinear interpolation. \\
\bottomrule
\end{tabular}
\end{center}

\bibliographystyle{abbrvnat}
\bibliography{reference}

\begin{thebibliography}{58}
\providecommand{\natexlab}[1]{#1}
\providecommand{\url}[1]{\texttt{#1}}
\expandafter\ifx\csname urlstyle\endcsname\relax
  \providecommand{\doi}[1]{doi: #1}\else
  \providecommand{\doi}{doi: \begingroup \urlstyle{rm}\Url}\fi

\bibitem[Abel and Eberly(1996)]{AbelEberly1996}
A.~B. Abel and J.~C. Eberly.
\newblock Optimal investment with costly reversibility.
\newblock \emph{Review of Economic Studies}, 63\penalty0 (4):\penalty0
  581--593, 1996.

\bibitem[Akian et~al.(1996)Akian, Menaldi, and Sulem]{AkianMenaldiSulem1996}
M.~Akian, J.~L. Menaldi, and A.~Sulem.
\newblock On an investment-consumption model with transaction costs.
\newblock \emph{SIAM Journal on Control and Optimization}, 34\penalty0
  (1):\penalty0 329--364, 1996.

\bibitem[Altarovici et~al.(2017)Altarovici, Reppen, and
  Soner]{AltaroviciReppenSoner2017}
A.~Altarovici, A.~M. Reppen, and H.~M. Soner.
\newblock Optimal consumption and investment with fixed and proportional
  transaction costs.
\newblock \emph{SIAM Journal on Control and Optimization}, 55\penalty0
  (3):\penalty0 1673--1710, 2017.

\bibitem[Arrow(1968)]{Arrow1968}
K.~J. Arrow.
\newblock Optimal capital policy with irreversible investment.
\newblock In J.~N. Wolfe, editor, \emph{Value, Capital and Growth}, pages
  1--19. Edinburgh University Press, Edinburgh, 1968.

\bibitem[Baldursson and Karatzas(1997)]{BaldurssonKaratzas1997}
F.~M. Baldursson and I.~Karatzas.
\newblock Irreversible investment and industry equilibrium.
\newblock \emph{Finance and Stochastics}, 1:\penalty0 69--89, 1997.

\bibitem[Bank and El~Karoui(2004)]{BankElKaroui2004}
P.~Bank and N.~El~Karoui.
\newblock A stochastic representation theorem with applications to optimization
  and obstacle problems.
\newblock \emph{Annals of Probability}, 32\penalty0 (1B):\penalty0 1030--1067,
  2004.

\bibitem[Belak and Sass(2019)]{BelakSass2019}
C.~Belak and J.~Sass.
\newblock Finite-horizon optimal investment with transaction costs:
  Construction of the optimal strategies.
\newblock \emph{Finance and Stochastics}, 23:\penalty0 861--888, 2019.

\bibitem[Bensoussan(1984)]{Bensoussan1984}
A.~Bensoussan.
\newblock On the theory of option pricing.
\newblock \emph{Acta Applicandae Mathematicae}, 2:\penalty0 139--158, 1984.

\bibitem[Bensoussan and Lions(1982)]{BensoussanLions1982}
A.~Bensoussan and J.-L. Lions.
\newblock \emph{Applications of Variational Inequalities in Stochastic
  Control}.
\newblock North-Holland, Amsterdam, 1982.

\bibitem[Broadie and Detemple(1997)]{BroadieDetemple1997}
M.~Broadie and J.~Detemple.
\newblock The valuation of {American} options on multiple assets.
\newblock \emph{Mathematical Finance}, 7\penalty0 (3):\penalty0 241--286, 1997.

\bibitem[Burdzy et~al.(2009)Burdzy, Kang, and Ramanan]{BurdzyKangRamanan2009}
K.~Burdzy, W.~Kang, and K.~Ramanan.
\newblock The {Skorokhod} problem in a time-dependent interval.
\newblock \emph{Stochastic Processes and their Applications}, 119\penalty0
  (2):\penalty0 428--452, 2009.

\bibitem[Caffarelli et~al.(2004)Caffarelli, Petrosyan, and Shahgholian]{CPS04}
L.~Caffarelli, A.~Petrosyan, and H.~Shahgholian.
\newblock Regularity of a free boundary in parabolic potential theory.
\newblock \emph{J. Amer. Math. Soc.}, 17\penalty0 (4):\penalty0 827--869, 2004.

\bibitem[Carr et~al.(1992)Carr, Jarrow, and Myneni]{CarrJarrowMyneni1992}
P.~Carr, R.~Jarrow, and R.~Myneni.
\newblock Alternative characterizations of {American} put options.
\newblock \emph{Mathematical Finance}, 2\penalty0 (2):\penalty0 87--106, 1992.

\bibitem[Chen and Dai(2013)]{ChenDai2013}
X.~Chen and M.~Dai.
\newblock Characterization of optimal strategy for multiasset investment and
  consumption with transaction costs.
\newblock \emph{SIAM Journal on Financial Mathematics}, 4\penalty0
  (1):\penalty0 857--883, 2013.

\bibitem[Choi et~al.(2013)Choi, S{\^\i}rbu, and {\v
  Z}itkovi{\'c}]{ChoiSirbuZitkovic2013}
J.~H. Choi, M.~S{\^\i}rbu, and G.~{\v Z}itkovi{\'c}.
\newblock Shadow prices and well-posedness in the problem of optimal investment
  and consumption with transaction costs.
\newblock \emph{SIAM Journal on Control and Optimization}, 51\penalty0
  (6):\penalty0 4414--4449, 2013.

\bibitem[Cvitanic and Karatzas(1996)]{CvitanicKaratzas1996}
J.~Cvitanic and I.~Karatzas.
\newblock Hedging and portfolio optimization under transaction costs: A
  martingale approach.
\newblock \emph{Mathematical Finance}, 6\penalty0 (2):\penalty0 133--165, 1996.

\bibitem[Dai and Yi(2009)]{DaiYi2009}
M.~Dai and F.~Yi.
\newblock Finite-horizon optimal investment with transaction costs: A parabolic
  double obstacle problem.
\newblock \emph{Journal of Differential Equations}, 246\penalty0 (4):\penalty0
  1445--1469, 2009.

\bibitem[Dai et~al.(2009)Dai, Jiang, Li, and Yi]{DaiJiangLiYi2009}
M.~Dai, L.~Jiang, P.~Li, and F.~Yi.
\newblock Finite horizon optimal investment and consumption with transaction
  costs.
\newblock \emph{SIAM Journal on Control and Optimization}, 48\penalty0
  (2):\penalty0 1134--1154, 2009.

\bibitem[Davis and Norman(1990)]{DavisNorman1990}
M.~H.~A. Davis and A.~R. Norman.
\newblock Portfolio selection with transaction costs.
\newblock \emph{Mathematics of Operations Research}, 15\penalty0 (4):\penalty0
  676--713, 1990.

\bibitem[{De Angelis} and Ferrari(2014)]{DeAngelisFerrari2014}
T.~{De Angelis} and G.~Ferrari.
\newblock A stochastic partially reversible investment problem on a finite
  time-horizon: Free-boundary analysis.
\newblock \emph{Stochastic Processes and their Applications}, 124\penalty0
  (12):\penalty0 4080--4119, 2014.

\bibitem[Dixit and Pindyck(1994)]{DixitPindyck1994}
A.~K. Dixit and R.~S. Pindyck.
\newblock \emph{Investment under Uncertainty}.
\newblock Princeton University Press, Princeton, 1994.

\bibitem[{El Karoui} and Karatzas(1991)]{ElKarouiKaratzas1991}
N.~{El Karoui} and I.~Karatzas.
\newblock A new approach to the {Skorohod} problem, and its applications.
\newblock \emph{Stochastics and Stochastics Reports}, 34\penalty0
  (1--2):\penalty0 57--82, 1991.

\bibitem[Federico and Pham(2014)]{FedericoPham2014}
S.~Federico and H.~Pham.
\newblock Characterization of the optimal boundaries in reversible investment
  problems.
\newblock \emph{SIAM Journal on Control and Optimization}, 52\penalty0
  (4):\penalty0 2180--2223, 2014.

\bibitem[Guasoni and Muhle-Karbe(2013)]{GuasoniMuhleKarbe2013}
P.~Guasoni and J.~Muhle-Karbe.
\newblock Portfolio choice with transaction costs: A user's guide.
\newblock In V.~Henderson and R.~Sircar, editors, \emph{Paris--Princeton
  Lectures on Mathematical Finance 2013}, volume 2081 of \emph{Lecture Notes in
  Mathematics}, pages 169--201. Springer, Berlin, 2013.

\bibitem[Guo and Pham(2005)]{GuoPham2005}
X.~Guo and H.~Pham.
\newblock Optimal partially reversible investment with entry decision and
  general production function.
\newblock \emph{Stochastic Processes and their Applications}, 115\penalty0
  (5):\penalty0 705--736, 2005.

\bibitem[Guo and Tomecek(2008)]{Guo2008}
X.~Guo and P.~Tomecek.
\newblock Connections between singular control and optimal switching.
\newblock \emph{SIAM Journal on Control and Optimization}, 47\penalty0
  (1):\penalty0 421--443, 2008.

\bibitem[Han and Yi(2025)]{HanYi2025}
X.~Han and F.~Yi.
\newblock A reversible investment problem with general cost function in finite
  horizon: Free boundaries analysis.
\newblock \emph{Journal of Mathematical Analysis and Applications},
  546:\penalty0 129221, 2025.

\bibitem[Han et~al.(2024)Han, Yi, and Zhang]{HanYiZhang2024}
X.~Han, F.~Yi, and J.~Zhang.
\newblock A reversible investment problem with capacity and demand in finite
  horizon: Free boundary analysis.
\newblock \emph{SIAM Journal on Control and Optimization}, 62\penalty0
  (2):\penalty0 1207--1234, 2024.

\bibitem[Hobson et~al.(2019)Hobson, Tse, and Zhu]{HobsonTseZhu2019}
D.~Hobson, A.~S.~L. Tse, and Y.~Zhu.
\newblock A multi-asset investment and consumption problem with transaction
  costs.
\newblock \emph{Finance and Stochastics}, 23:\penalty0 641--676, 2019.

\bibitem[Jacka(1991)]{Jacka1991}
S.~D. Jacka.
\newblock Optimal stopping and the {American} put.
\newblock \emph{Mathematical Finance}, 1\penalty0 (2):\penalty0 1--14, 1991.

\bibitem[Jaillet et~al.(1990)Jaillet, Lamberton, and
  Lapeyre]{JailletLambertonLapeyre1990}
P.~Jaillet, D.~Lamberton, and B.~Lapeyre.
\newblock Variational inequalities and the pricing of {American} options.
\newblock \emph{Acta Applicandae Mathematicae}, 21\penalty0 (3):\penalty0
  263--289, 1990.

\bibitem[Jane{\v c}ek and Shreve(2004)]{JanecekShreve2004}
K.~Jane{\v c}ek and S.~E. Shreve.
\newblock Asymptotic analysis for optimal investment and consumption with
  transaction costs.
\newblock \emph{Finance and Stochastics}, 8:\penalty0 181--206, 2004.

\bibitem[Jeon and Kim(2026)]{JeonKim2026}
J.~Jeon and T.~Kim.
\newblock The finite-horizon reversible investment problem with the constant
  elasticity of variance model.
\newblock \emph{SIAM Journal on Control and Optimization}, 64\penalty0
  (2):\penalty0 843--868, 2026.

\bibitem[Kabanov and Kl{\"u}ppelberg(2004)]{KabanovKluppelberg2004}
Y.~Kabanov and C.~Kl{\"u}ppelberg.
\newblock A geometric approach to portfolio optimization in models with
  transaction costs.
\newblock \emph{Finance and Stochastics}, 8:\penalty0 207--227, 2004.

\bibitem[Kallsen and Muhle-Karbe(2010)]{KallsenMuhleKarbe2010}
J.~Kallsen and J.~Muhle-Karbe.
\newblock On using shadow prices in portfolio optimization with transaction
  costs.
\newblock \emph{Annals of Applied Probability}, 20\penalty0 (4):\penalty0
  1341--1358, 2010.

\bibitem[Karatzas(1988)]{Karatzas1988}
I.~Karatzas.
\newblock On the pricing of {American} options.
\newblock \emph{Applied Mathematics and Optimization}, 17:\penalty0 37--60,
  1988.

\bibitem[Karatzas and Shreve(1984)]{KaratzasShreve1984}
I.~Karatzas and S.~E. Shreve.
\newblock Connections between optimal stopping and singular stochastic control
  {I}. {M}onotone follower problems.
\newblock \emph{SIAM Journal on Control and Optimization}, 22\penalty0
  (6):\penalty0 856--877, 1984.

\bibitem[Karatzas and Shreve(1985)]{KaratzasShreve1985}
I.~Karatzas and S.~E. Shreve.
\newblock Connections between optimal stopping and singular stochastic control
  {II}. {R}eflected follower problems.
\newblock \emph{SIAM Journal on Control and Optimization}, 23\penalty0
  (3):\penalty0 433--451, 1985.

\bibitem[Kim(1990)]{Kim1990}
I.~J. Kim.
\newblock The analytic valuation of {American} options.
\newblock \emph{Review of Financial Studies}, 3\penalty0 (4):\penalty0
  547--572, 1990.

\bibitem[Klimsiak and Rozkosz(2016)]{KlimsiakRozkosz2016}
T.~Klimsiak and A.~Rozkosz.
\newblock The early exercise premium representation for {American} options on
  multiply assets.
\newblock \emph{Applied Mathematics \& Optimization}, 73:\penalty0 99--114,
  2016.

\bibitem[Kukuljan(2022)]{Kukuljan2022}
T.~Kukuljan.
\newblock Higher order parabolic boundary {Harnack} inequality in {$C^1$} and
  {$C^{k,\alpha}$} domains.
\newblock \emph{Discrete and Continuous Dynamical Systems}, 42\penalty0
  (6):\penalty0 2667--2698, 2022.

\bibitem[Laurence and Salsa(2009)]{LaurenceSalsa2009}
P.~Laurence and S.~Salsa.
\newblock Regularity of the free boundary of an {American} option on several
  assets.
\newblock \emph{Communications on Pure and Applied Mathematics}, 62\penalty0
  (7):\penalty0 969--994, 2009.

\bibitem[Lindgren and Monneau(2015)]{LindgrenMonneau2015}
E.~Lindgren and R.~Monneau.
\newblock Pointwise regularity of the free boundary for the parabolic obstacle
  problem.
\newblock \emph{Calculus of Variations and Partial Differential Equations},
  54\penalty0 (1):\penalty0 299--347, 2015.

\bibitem[Magill and Constantinides(1976)]{MagillConstantinides1976}
M.~J.~P. Magill and G.~M. Constantinides.
\newblock Portfolio selection with transaction costs.
\newblock \emph{Journal of Economic Theory}, 13\penalty0 (2):\penalty0
  245--263, 1976.

\bibitem[McKean(1965)]{McKean1965}
H.~P. McKean.
\newblock Appendix: A free boundary problem for the heat equation arising from
  a problem in mathematical economics.
\newblock \emph{Industrial Management Review}, 6\penalty0 (2):\penalty0 32--39,
  1965.

\bibitem[Merhi and Zervos(2007)]{MerhiZervos2007}
A.~Merhi and M.~Zervos.
\newblock A model for reversible investment capacity expansion.
\newblock \emph{SIAM Journal on Control and Optimization}, 46\penalty0
  (3):\penalty0 839--876, 2007.

\bibitem[Merton(1969)]{Merton1969}
R.~C. Merton.
\newblock Lifetime portfolio selection under uncertainty: The continuous-time
  case.
\newblock \emph{Review of Economics and Statistics}, 51\penalty0 (3):\penalty0
  247--257, 1969.

\bibitem[Merton(1971)]{Merton1971}
R.~C. Merton.
\newblock Optimum consumption and portfolio rules in a continuous-time model.
\newblock \emph{Journal of Economic Theory}, 3\penalty0 (4):\penalty0 373--413,
  1971.

\bibitem[Oksendal(2000)]{Oksendal2000}
A.~Oksendal.
\newblock Irreversible investment problems.
\newblock \emph{Finance and Stochastics}, 4:\penalty0 223--250, 2000.

\bibitem[Pindyck(1988)]{Pindyck1988}
R.~S. Pindyck.
\newblock Irreversible investment, capacity choice, and the value of the firm.
\newblock \emph{American Economic Review}, 78\penalty0 (5):\penalty0 969--985,
  1988.

\bibitem[Pindyck(1991)]{Pindyck1991}
R.~S. Pindyck.
\newblock Irreversibility, uncertainty, and investment.
\newblock \emph{Journal of Economic Literature}, 29\penalty0 (3):\penalty0
  1110--1148, 1991.

\bibitem[Possama{\"i} et~al.(2015)Possama{\"i}, Soner, and
  Touzi]{PossamaiSonerTouzi2015}
D.~Possama{\"i}, H.~M. Soner, and N.~Touzi.
\newblock Homogenization and asymptotics for small transaction costs: The
  multidimensional case.
\newblock \emph{Communications in Partial Differential Equations}, 40\penalty0
  (11):\penalty0 2005--2046, 2015.

\bibitem[Reppen et~al.(2025)Reppen, Soner, and
  Tissot-Daguette]{ReppenSonerTissotDaguette2025}
A.~M. Reppen, H.~M. Soner, and V.~Tissot-Daguette.
\newblock Neural optimal stopping boundary.
\newblock \emph{Mathematical Finance}, 35\penalty0 (2):\penalty0 441--469,
  2025.

\bibitem[Riedel and Su(2011)]{RiedelSu2011}
F.~Riedel and X.~Su.
\newblock On irreversible investment.
\newblock \emph{Finance and Stochastics}, 15:\penalty0 607--633, 2011.

\bibitem[Shreve and Soner(1994)]{ShreveSoner1994}
S.~E. Shreve and H.~M. Soner.
\newblock Optimal investment and consumption with transaction costs.
\newblock \emph{Annals of Applied Probability}, 4\penalty0 (3):\penalty0
  609--692, 1994.

\bibitem[Soner and Touzi(2013)]{SonerTouzi2013}
H.~M. Soner and N.~Touzi.
\newblock Homogenization and asymptotics for small transaction costs.
\newblock \emph{SIAM Journal on Control and Optimization}, 51\penalty0
  (4):\penalty0 2893--2921, 2013.

\bibitem[Steg(2012)]{Steg2012}
J.-H. Steg.
\newblock Irreversible investment in oligopoly.
\newblock \emph{Finance and Stochastics}, 16:\penalty0 207--224, 2012.

\bibitem[Villeneuve(1999)]{Villeneuve1999}
S.~Villeneuve.
\newblock Exercise regions of {American} options on several assets.
\newblock \emph{Finance and Stochastics}, 3:\penalty0 295--322, 1999.

\end{thebibliography}

\end{document}